\documentclass[11pt,letterpaper]{article}

\usepackage[margin=1in]{geometry}

\usepackage{dsfont, amssymb,amsmath,amscd,latexsym, amsthm, amsxtra,amsfonts}

\usepackage[dvipsnames,svgnames]{xcolor}

\usepackage{lineno}
\usepackage[active]{srcltx}

\usepackage{tikz}

\usepackage{natbib}
 \bibpunct[, ]{(}{)}{,}{a}{}{,}%

\usepackage{bbm}
\usepackage{enumerate}
\usepackage{mathrsfs}

\usepackage{comment}
\usepackage{cases}
\usepackage{circuitikz}

\usepackage{subcaption}

\usepackage{verbatim}
\usepackage{graphicx}
\usepackage{epstopdf}
\usepackage{bm}%\usetikzlibrary{shapes,arrows,positioning}
\usetikzlibrary{automata,arrows,positioning,calc}
\usepackage[pdfstartview=FitH, bookmarksnumbered=true,bookmarksopen=true, colorlinks=true, pdfborder={0 0 1}, citecolor=blue, linkcolor=blue,urlcolor=blue]{hyperref}
\numberwithin{equation}{section}
\usepackage{graphics}
\usepackage{booktabs}
\usepackage{algorithm,algorithmic}

\usepackage{mathtools}
\mathtoolsset{showonlyrefs=true}

\usepackage{adjustbox}                  % auto-shrink figure

\usetikzlibrary{arrows.meta, positioning, calc, shadows.blur}

\definecolor{brandBlue}{HTML}{0072B2}
\definecolor{brandOrange}{HTML}{E69F00}
\definecolor{brandGreen}{HTML}{009E73}
\definecolor{brandPurple}{HTML}{6F42C1}

\allowdisplaybreaks

\newtheorem{theorem}{Theorem}[section]
\newtheorem*{theorem*}{Theorem}
\newtheorem{corollary}[theorem]{Corollary}
\newtheorem{proposition}[theorem]{Proposition}
\newtheorem{lemma}[theorem]{Lemma}

\theoremstyle{definition}
\newtheorem{definition}[theorem]{Definition}

\theoremstyle{definition}
\newtheorem{example}{Example}
\theoremstyle{definition}
\newtheorem{remark}{Remark}
\theoremstyle{definition}
\newtheorem{assumption}{Assumption}

\usepackage{amsmath}

\DeclareMathOperator*{\argmin}{arg\,min}
\newcommand{\md}{\mathrm{d}}

\newcommand{\mR}{\mathbb{R}}

\newcommand{\mN}{\mathbb{N}}

\newcommand{\mE}{\mathbb{E}}
\newcommand{\mF}{\mathbb{F}}
\newcommand{\mP}{\mathbb{P}}
\newcommand{\mQ}{\mathbb{Q}}

\newcommand{\ind}{\mathbbm{1}}

\renewcommand{\epsilon}{\varepsilon}

\newcommand{\N}{\mathcal{N}}
\newcommand{\F}{\mathcal{F}}
\newcommand{\E}{\mathcal{E}}
\newcommand{\D}{\mathcal{D}}

\newcommand{\Ti}{\mathcal{T}}

\newcommand{\cZ}{\mathcal{Z}}
\newcommand{\cA}{\mathcal{A}}
\newcommand{\cW}{\mathcal{W}}
\newcommand{\cL}{\mathcal{L}}

\newcommand{\cP}{\mathcal{P}}

\newcommand{\bz}{\mathbf{z}}

\newcommand{\mt}{\top}

\newcommand{\TV}{\mathrm{TV}}
\newcommand{\Var}{\mathrm{Var}}
\newcommand{\Cov}{\mathrm{Cov}}
\newcommand{\KL}{\mathrm{KL}}
\newcommand{\targ}{\mathrm{targ}}
\newcommand{\Npre}{\mathrm{N_{pre}}}
\newcommand{\Ncor}{\mathrm{N_{cor}}}
\newcommand{\bep}{\bm{\epsilon}}
\newcommand{\tep}{\tau_0,\epsilon_{\rm score},T}

\title{ Dynamic data generation and dynamic portfolio selection: an application of a score-based diffusion model\thanks{Erhan Bayraktar is supported in part by National Science Foundation grant DMS-2507940 and by the Susan M. Smith Professorship. Fengyi Yuan is supported by the Chinese University of Hong Kong (Shenzhen) start-up fund
under UDF01004253.}}

\author{
	 Ahmad Aghapour\thanks{Department of Mathematics, University of Michigan, Ann Arbor, MI, USA. %\textbf{e-mail}:
		\url{aghapour@umich.edu}}
	\and Erhan Bayraktar\thanks{Department of Mathematics, University of Michigan, Ann Arbor, MI, USA. %\textbf{e-mail}:
		\url{erhan@umich.edu}}
	\and Fengyi Yuan\thanks{School of Science and Engineering, the Chinese University of Hong Kong (Shenzhen), Shenzhen, Guangdong, China.  %\textbf{e-mail}:
		\url{yuanfengyi@cuhk.edu.cn}}
}

\date{}

\begin{document}
\maketitle

\begin{abstract}
We study dynamic data generation and its application to model-free dynamic portfolio selection. Existing score-based diffusion models are typically designed to learn a static data distribution, whereas dynamic decision problems require generated trajectories that preserve the sequential information structure of the underlying process and support conditional sampling. To address this gap, we develop an adaptive score-based diffusion framework for dynamic data. Given samples from an unknown data-generating model $\mP$, the framework learns a generative model $\mQ$ through conditional score matching and generates trajectories sequentially by updating the conditioning information over time. We establish quantitative error bounds between $\mP$ and $\mQ$ under the adapted Wasserstein metric $\cA\cW_2$, which is tailored to nonanticipative dynamic problems, and show that the same adaptive sampling scheme provides conditional path generators. We then apply this dynamic data generation framework to dynamic mean-variance portfolio selection with limited historical price data. We prove stability of the dynamic mean-variance problem with respect to $\cA\cW_2$, thereby translating the generative approximation error into performance control for portfolio policies. Building on these results, we implement a policy-gradient algorithm in the learned generative environment, where adaptively sampled paths serve as training scenarios. A synthetic ARMA experiment shows that the proposed adaptive sampling scheme generates distributions close to the true data-generating process. On real market data, the proposed approach outperforms several benchmarks, including the Markowitz portfolio, the equal-weight portfolio, and the S\&P 500.
\end{abstract}
\bigskip
\noindent\textbf{Subject Classification:}
dynamic programming: optimal control,
finance: portfolio,
computer science: artificial intelligence.

\section{Introduction}

Portfolio selection, a fundamental task in financial engineering, has been a subject of both theoretical and practical interest since the classical work of \citet{Markowitz1952}. Because dynamic formulations account for intertemporal hedging and rebalancing in nonstationary markets, formulating portfolio selection as a {\it dynamic} multiperiod problem has become increasingly appealing; see \citet{Mossin1968} and \citet{Merton1971} for early work. Researchers have explored this problem under various specifications, such as continuous-time dynamic mean-variance problems, in which return models are prescribed by stochastic differential equations \citep{ZhouLi2000}, and discrete-time problems with independent intertemporal returns \citep{LiNg2000}. However, because information is incomplete, estimation errors are unavoidable, and fully modeling market uncertainty is infeasible, we generally lack access to the true model required by model-based approaches. Instead, we have access only to a limited sample of financial data. Consequently, investigating portfolio selection problems in a {\it model-free} manner has been a long-standing topic in financial engineering; see discussions of the `universal portfolio' \citep{Cover1991}, stochastic portfolio theory \citep{Karatzas2017,CSW2017}, and applications of machine learning \citep{BEL2018,GKX2020}. However, dynamic portfolio selection problems are rather involved, even in model-based settings, due to their inherent intertemporal structure \citep{CV1999}, high dimensionality \citep{BGSS2005}, and time-inconsistency \citep{LiNg2000,BasakChabakauri2010,BMZ2014}\footnote{ \citet{BasakChabakauri2010} and \citet{
BMZ2014} adopt a different solution concept---equilibrium solutions---to obtain time-consistent policies. In this paper, we focus on time-0-optimal (precommitted) policies.}. Therefore, model-free dynamic portfolio selection is a natural and important problem to explore.

To address this challenge, we leverage score-based diffusion models. These models can generate additional samples whose distributional properties are similar to those of the original data. Specifically, when the law $\mP$ of $S^{1:T}$ is unknown, we use a data-driven approach to train a model $\mQ$ that closely resembles $\mP$. We use $\mQ$ as a scenario generator for the downstream portfolio step: we sample a large number of paths $\tilde S^{1:T}\sim \mQ$ and train a reinforcement learning (RL) agent on these simulated paths. Theoretically, we show that our method achieves near-optimality under Markowitz's mean-variance (MV) criterion. The analysis rests on two key contributions: (1) a quantitative error bound between $\mP$ and $\mQ$, and (2) model stability. From an engineering perspective, the approximate model $\mQ$ consists of a score network $s_\theta$, which is the core of the score-matching technique used in generative models, and a recurrent neural network (RNN) encoder $R_\theta$, which encodes the current state of the market. These two networks are trained simultaneously using historical price data. The reinforcement learning agent is trained on data generated from $\mQ$ by adaptively evaluating $s_\theta$ and $R_\theta$. Finally, the agent, fed with the encoded and updated market state, outputs portfolio allocations whose performance is evaluated in real-world experiments. For a summary of the theoretical results, see Figure \ref{fig:roadmap}; for the workflow implementing the proposed solution for model-free dynamic portfolio selection, see Figure \ref{fig:roadmap:implementation}.

\begin{figure}[t]
\centering
\begin{tikzpicture}[
  box/.style={
    rectangle, draw, very thick, minimum width=3cm, minimum height=1cm, align=center
  },
  arrow/.style={
    ->, thick, shorten >=2pt, shorten <=2pt
  },
  dashedarrow/.style={
    ->, thick, dashed, shorten >=2pt, shorten <=2pt
  },
  node distance=1.5cm and 3cm
]
  % nodes
  \node[box]                        (MVP) {MV problem\\under \(\mP\)};
  \node[box, below=of MVP]         (MVQ) {MV problem\\under \(\mQ\)};
  \node[box, below=of MVQ]         (QH)  {Quadratic hedging\\under \(\mQ\)};
  \node[box, right=5cm of MVP]     (Data){Financial Data\\from \(\mP\)};
  \node[box, below=of Data]        (PM)  {Pretrained model\\\(\mQ\)};
  \node[box, right=of QH]          (PG)  {Policy gradient\\algorithm\\ (Algorithm \ref{FVI-Q})};

  % arrows and labels
  \draw[arrow] (MVQ) -- node[left]{\shortstack{Model Stability\\(Corollary \ref{coro:MVstability})}} (MVP);
  \draw[arrow] (QH)  -- node[left]{Duality \eqref{MV-hedging-duality}}               (MVQ);
  \draw[arrow] (PG)  --                (QH);
  \draw[arrow] (PM)  -- node[right]{\shortstack{Adaptive sampling\\(Section \ref{sec:sampling})}} (PG);
  \draw[arrow] (Data) -- node[right]{\shortstack{Quantitative error bounds\\in \(\cA\cW_2\)\\ (Section \ref{sec:mainresults})}} (PM);
  \draw[dashedarrow] (Data) to[bend left=25] node[above]{\shortstack{Model-based approach\\(omniscient solution)}} (MVP);

\end{tikzpicture}
\caption{Model-free dynamic MV portfolio selection: theory}
\label{fig:roadmap}
\end{figure}

\begin{figure}[htbp]
\tikzset{
  box/.style = {
    rectangle, draw, very thick,
    minimum width=3cm, minimum height=1cm,
    align=center
  },
  arrow/.style = {
    ->, thick,
    shorten >=2pt, shorten <=2pt
  },
  dashedarrow/.style = {
    ->, thick, dashed,
    shorten >=2pt, shorten <=2pt
  }
}

\centering

\begin{tikzpicture}[
  every node/.style = {font=\scriptsize, align=center},
  box/.style   = {draw, rounded corners=3pt, thick},
  data/.style  = {box, top color=white, bottom color=brandBlue!10,
                  minimum width=25mm, minimum height=12mm},
  proc/.style  = {box, top color=white, bottom color=brandPurple!8,
                  minimum width=25mm, minimum height=12mm},
  pathbox/.style={box, top color=white, bottom color=brandOrange!12,
                  minimum width=27mm, minimum height=12mm},
  arrow/.style = {->, >=Latex, line width=.9pt, color=brandPurple},
  node distance=10mm
]

% --------------------- top row ---------------------
\node[data] (hist)  {\includegraphics[width=10mm]{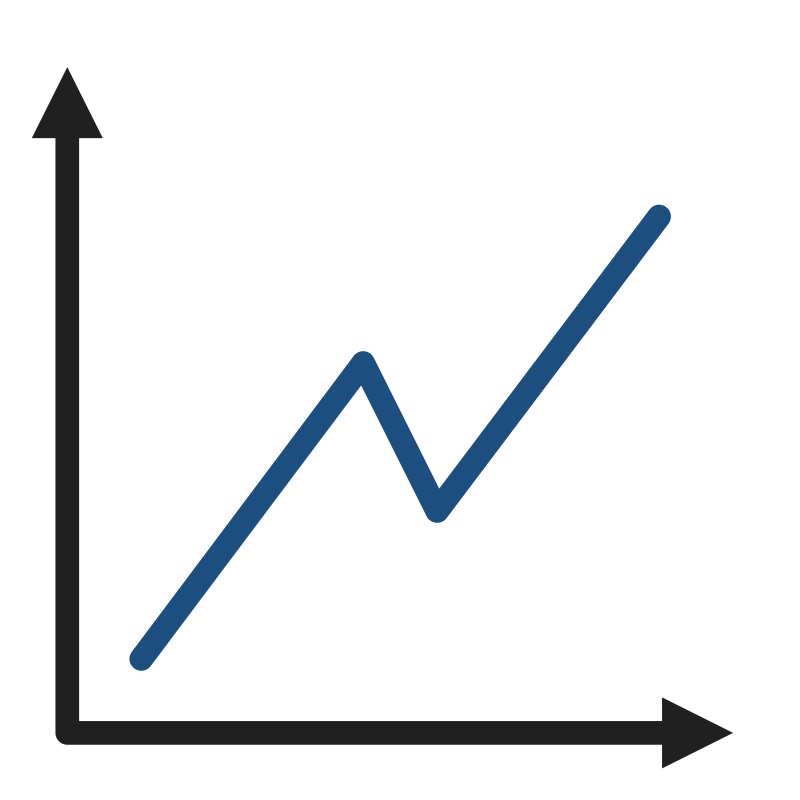}\\[-1pt] Historical\\Prices\\$S^{1:T}$};

\node[proc, right=of hist] (rnn) {RNN\\Encoder $R_\theta$};

\node[proc, right=of rnn] (rl) {%
  \includegraphics[width=10mm]{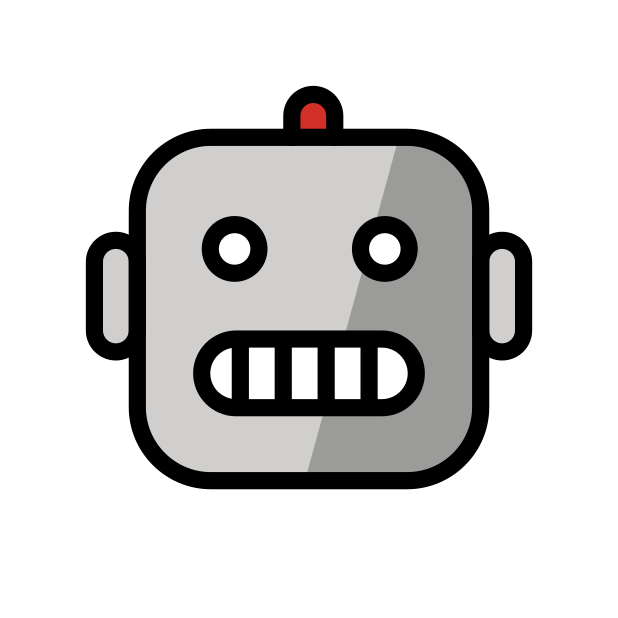}\\[-1pt]
  RL Agent\\(policy $\pi_\beta$)};

\node[data, top color=white, bottom color=brandGreen!10,
      right=of rl] (alloc) {Optimal\\Weights\\$a_t$};

% --------------------- bottom row (directly under RNN) ---------------------
\node[proc, below=of rnn] (diff) {Diffusion\\Model $s_\theta$};

\node[pathbox, right=of diff] (paths) {%
  \includegraphics[width=10mm]{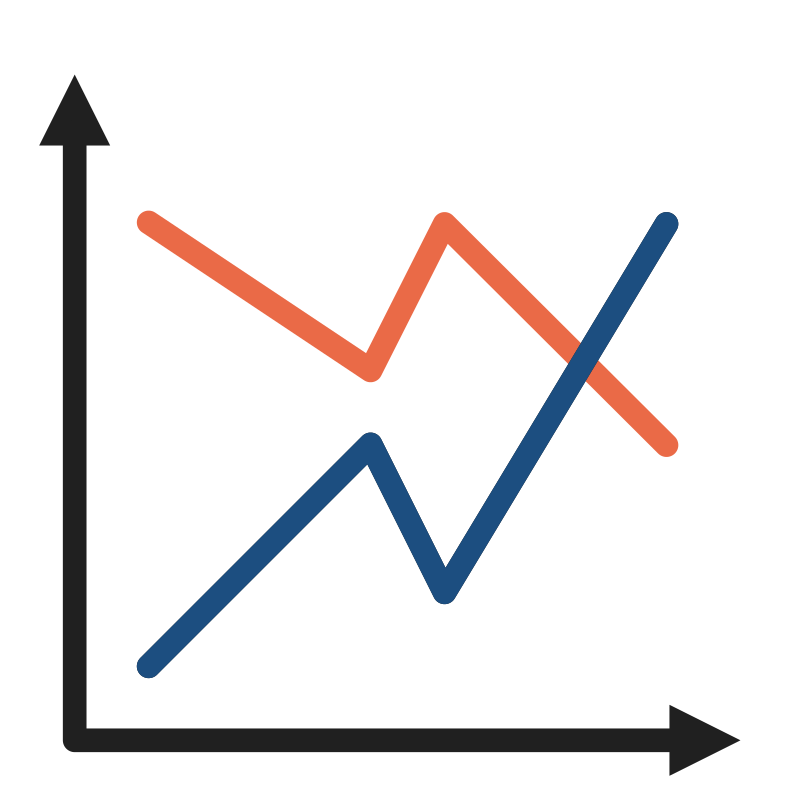}\\[-1pt]
  Sampled Paths\\$\tilde S^{1:T}$};

% --------------------- arrows ---------------------
\draw[arrow] (hist) -- (rnn);          % input to encoder
\draw[arrow] (rnn)  -- (diff);         % encoder → diffusion
\draw[arrow] (diff) -- (paths);        % diffusion → paths
\draw[arrow] (paths) -- (rl);          % paths → RL agent (upward/slanted)
\draw[arrow] (rnn) -- (rl);
\draw[arrow] (rl)    -- (alloc);       % RL agent → weights

\end{tikzpicture}
% \begin{tikzpicture}[node distance=1.5cm and 3cm, auto]

%   % Nodes
%   \node[box] (prices) {Historical prices data};
%   \node[box, right=of prices] (rnn) {RNN encoder $R_\theta$};
%   \node[box, below=of prices] (diffusion) {Diffusion model \\(score network $s_\theta)$};
%   \node[box, below=of diffusion] (sampled) {Generated prices data\\$\widetilde s^{1:T}$};
%   \node[box, below=of rnn] (output) {Portfolio strategies};
%   \node[box, below=of output] (rl) {RL Agent\\(policy  network $\pi_\beta$)};

%   \draw[arrow] (prices) -- (rnn);
%   \draw[arrow] (prices) -- (diffusion);
%   \draw[arrow] (rnn) -- (output);
%   \draw[arrow] (rl) -- (output);
%   \draw[arrow] (diffusion) -- (sampled);
%   \draw[arrow] (sampled) -- (rl);

% \end{tikzpicture}

\caption{Model-free dynamic MV portfolio selection: implementation}
\label{fig:roadmap:implementation}
\end{figure}

\subsection{\texorpdfstring{  Main results and contributions}{Main results and contributions}}

{ 
Score matching has become a standard approach to training generative models since the work of \citet{hyvarinen2005score}; see, for example, \citet{SWMG2015} and \citet{SE2019}. The score-based framework of \citet{Song2021b}, which unifies score-based generative models and diffusion probabilistic models \citep{ho2020denoising}, provides a stochastic sampling mechanism and admits quantitative approximation guarantees; see \citet{tang2025scorebaseddiffusionmodelsstochastic} for a survey and \citet{gao2025wasserstein} for recent Wasserstein bounds. This methodology has been used in image and video generation \citep{Rombach2022,Melnik2024}, molecular simulation \citep{HSVW2022}, and health applications \citep{alphafold}. In finance, \citet{factormodel} recently studied score-based generation for high-dimensional return data with low-dimensional structure. These static models provide the starting point for our construction, and Section \ref{subsec:static-diffusion} reviews the formal diffusion framework.
}

The static construction does not directly address the decision problem studied here. Dynamic portfolio selection is driven by a time series $X=(X^1,X^2,\cdots,X^T)=X^{1:T}$, and the generated trajectories must preserve the information available at each decision time. Treating the whole path as a vector in $\mR^{dT}$ ignores this nonanticipative structure. { This distinction is crucial, as value functions in dynamic decision problems can be unstable under ordinary Wasserstein perturbations of the path law; see \citet{BBBE2020} and Example \ref{exm:instability}}. We therefore measure generative error by the adapted Wasserstein metric $\cA\cW_2$, which compares probability laws while respecting the nonanticipative information structure. A second requirement is conditional sampling. At time $t$, the policy-gradient algorithm needs samples from the future conditional law given the observed history $X^{1:t}=x^{1:t}$, not only unconditional samples from the joint law.

{ 
To meet these requirements, we train and sample adaptively and sequentially. At the conceptual level, the stage-$t$ diffusion is initialized from the clean next-state law $\mP_{x^{1:t}}$, conditional on the observed history $X^{1:t}=x^{1:t}$. To account for the possible degeneracy of the data distribution, we use an early-stopping technique as in \citet{chen2023sampling}. In the formal analysis, we replace the clean history by a noisy history, train against a posterior-mixture target score function, and stop the reverse process early, at time $\Ti-\tau_0$; Section \ref{subsec:conditional-score-matching} gives the precise conditional VP-SDE. Given the trained scores, the sampler generates a path recursively: it first samples $y^1$ by running the reverse VP-SDE from noise, and then samples $y^{t+1}$ by running the reverse equation with conditional input $y^{1:t}$ and the corresponding conditional score approximation. Algorithm \ref{adaptive-sampling} gives the full procedure. The resulting law $\mQ^{\Ti-\tau_0}$ is close to the data-generating law $\mP$ in the metric relevant for dynamic decisions:
}
\begin{theorem*}[Theorem \ref{thm:AWbounds} in the main text]
Denote by $\mQ^{\Ti-\tau_0}$ the output law of the early-stopped adaptive sampler. Then $\cA\cW_2(\mP,\mQ^{\Ti-\tau_0})$ is bounded explicitly in terms of the early-stopping parameter $\tau_0$, the reverse diffusion horizon $\Ti$, and the score-matching error $\epsilon_{\rm score}$. The bound converges to zero under a suitable scaling with $\tau_0\downarrow0$, $\Ti\to\infty$, and $\epsilon_{\rm score}\downarrow0$. 
\end{theorem*}
{ The same construction also provides conditional sampling. Given a conditioning history, we use it as the input and run the corresponding reverse SDE under the learned model. The proof combines conditional total-variation estimates, the clean-to-noisy early-stopping comparison, assumptions on the joint time-series law, and tail estimates for the reverse SDE; see Section \ref{sec:mainresults} and the appendices.}

The generative guarantee alone is not sufficient for portfolio selection or other dynamic decision problems. The learned law $\mQ$ is only a surrogate for the unknown market law $\mP$, and the portfolio obtained under $\mQ$ is useful only if its performance transfers back to $\mP$. We establish this transfer through value stability and policy performance bounds for the dynamic mean-variance problem:
\begin{theorem*}[Theorem \ref{DPP-stability} and Corollaries \ref{coro:MVstability} and \ref{coro:regular-policy-transfer} in the main text]
Denote by $v(\mP)$ and $v(\mQ)$ the optimal values of the mean-variance problem under $\mP$ and $\mQ$, respectively. Then, under appropriate technical conditions, $|v(\mP)-v(\mQ)|\leq C\cA\cW_2(\mP,\mQ)$. Moreover, if a sufficiently regular feedback policy is $\epsilon$-optimal under the surrogate model $\mQ$, then it is $(\epsilon+C\cA\cW_2(\mP,\mQ))$-suboptimal under the data-generating model $\mP$.
\end{theorem*}
The proof uses the adapted nature of $\cA\cW_2$ and a duality argument from \citet{LH2007} that reduces the dynamic mean-variance objective to a quadratic-hedging problem. This reduction also restores a dynamic-programming structure that is not available directly for the time-inconsistent mean-variance criterion. The policy-transfer result then combines value stability with a fixed-policy stability estimate for regular feedback policies.

It remains to solve the mean-variance problem under the implicit model $\mQ$. To this end, we use $\mQ$ as a generative environment and implement a policy-gradient algorithm; see Algorithm \ref{FVI-Q}. This application requires two adaptations of standard policy-gradient methods. First, financial portfolio selection lacks a natural episodic environment that can be reset repeatedly without relying on long historical samples that may be stale. The diffusion sampler supplies fresh paths conditional on the current market state. Second, standard policy-gradient methods rely on dynamic programming, whereas the mean-variance objective violates the Bellman principle \citep{BasakChabakauri2010}. We use the duality result of \citet{LH2007} to solve the associated quadratic-hedging problem and then recover the mean-variance portfolio through the dual relation.

Figure \ref{fig:roadmap:implementation} summarizes the implemented pipeline. Because the conditioning histories $h^{1:t}$ and $y^{1:t}$ grow with $t$, adaptive sampling would otherwise require history inputs in $\mR^{dt}$. We therefore use an RNN encoder to map the history into a fixed-dimensional representation and train this encoder jointly with the score network; see \citet{rasul2021,yan2021}. Section \ref{sec:exp} evaluates the method on synthetic and real data. In the synthetic ARMA experiment, the generated distribution closely matches the oracle distribution in terms of conditional moments and dependence diagnostics. Across the two real-data experiments, both GenMarkowitz and GenTD3 deliver higher annualized returns and stronger risk-adjusted performance than the S\&P~500, an equal-weight portfolio, and a history-based Markowitz portfolio. The FF-30 experiment further indicates that the proposed methods remain effective as the asset dimension increases from 10 to 30 after modest architectural adjustments, including an increase in the score-network base width.

{ 
To conclude, the main contribution of this work is a dynamic data-generation framework with decision-relevant guarantees. More specifically, the paper makes the following three contributions:
\begin{enumerate}
\item We develop an adaptive score-based diffusion framework for dynamic data distributions. The approach trains the model and samples from it adaptively, fully respects the information structure of the data, and provides an explicit $\cA\cW_2$ error bound.

\item We show why the $\cA\cW_2$ guarantee is the relevant error notion for downstream dynamic portfolio decisions. For dynamic mean-variance selection, we prove value stability and policy performance transfer: an approximately optimal regular feedback policy under the learned law $\mQ$ remains near-optimal under the data-generating law $\mP$, with the performance gap controlled by $\cA\cW_2(\mP,\mQ)$.

\item We evaluate conditional-generation fidelity against an exact oracle in the synthetic experiment and downstream portfolio performance in the real-data experiments.
\end{enumerate}}

\subsection{\texorpdfstring{  Related work}{Related work}}

{ 
Our work is related to three streams of research: simulation-based portfolio optimization, generative models for financial time series, and learning methods for dynamic portfolio decisions.

\paragraph{Simulation and resampling for portfolio selection.}
Simulation has long been used to construct scenarios for multistage decision problems and financial risk analysis \citep{hoyland2001generating,glasserman2004}. In portfolio optimization, resampling methods address estimation error in mean-variance inputs by averaging decisions over bootstrap samples \citep{michaud1998,Efron1993}. A complementary line of work fits structured return models, simulates fresh paths, and optimizes portfolios over the resulting scenarios. Examples include \citet{AngBekaert2002} for regime-switching return models, \citet{PaolellaPolakWalker2021,SahamkhadamStephanOstermark2018,Guastaroba2009} for non-Gaussian GARCH and EVT-copula models, and \citet{Lim2004} for stochastic-parameter diffusion models in portfolio selection. Simulation also plays a central role in approximate dynamic programming and information-relaxation bounds for stochastic dynamic programs \citep{BrownSmith2014}. These methods are effective when the structured or parametrized data-generating mechanism or resampling scheme fits the real data well, a condition that is difficult to satisfy in practice. In contrast, the present paper proposes a {\it data-driven} approach to learn a generator and studies its error under a metric tailored to dynamic decisions.

\paragraph{Data-driven financial scenario generators.}
A large literature focuses on data-driven generative models for financial applications. The TimeGAN and QuantGAN papers show that adversarial generators can reproduce stylized features such as heavy tails, volatility clustering, and cross-sectional dependence \citep{yoon2019timegan,wiese2020quantgan}. Representative specialized GAN papers include \citet{marti2020corrgan} on correlation matrices, \citet{liao2024sig} on conditional path generation under a signature metric, \citet{VolGAN,VolGANhedging} on implied-volatility surfaces and hedging applications, and \citet{cont2022tailgan} on tail-risk scenarios. These papers demonstrate the value of flexible generators for financial scenario construction. Their guarantees, however, are usually formulated in terms of distributional realism or scenario quality, rather than a bound that transfers a learned dynamic law to the performance of downstream dynamic decisions.

Score-based diffusion models, which provide the foundation for our paper, offer a different generative paradigm. Denoising diffusion probabilistic models and score-SDE models supply the basic methodology \citep{ho2020denoising,Song2021b}, while recent theory gives convergence guarantees in total variation, Wasserstein distance, or KL divergence under structural assumptions on the data law and score error \citep{de2022convergence,chen2023sampling,gao2025wasserstein}. Time-series diffusion models have been developed for probabilistic forecasting, imputation, and synthetic sequence generation. Representative works include \citet{rasul2021} on TimeGrad, \citet{yan2021} on ScoreGrad, \citet{tashiro2021csdi} on CSDI, \citet{shen2023non} on TimeDiff, \citet{kollovieh2023tsdiff} on TSDiff, and \citet{huang2024ftsdiff} on FTS-Diffusion. Finance-focused diffusion works include \citet{WangVentre2024} on denoising for financial time series, \citet{Kubiak2024} on synthetic correlation matrices, \citet{factormodel} on factor-structured high-dimensional returns, \citet{TakahashiMizuno2024} on wavelet-based synthetic financial time series, and \citet{KimChoiKim2025} on price-dependent noising mechanisms based on geometric Brownian motion. These works focus on forecasting accuracy and imputation quality, with most of the results being experimental. In contrast to the static framework, our analysis controls a sequentially generated law in $\cA\cW_2$, the metric that preserves the information structure used by downstream policies.

\paragraph{Conditional diffusion models for dynamic data.}
Conventional conditional diffusion models generate a target from an externally supplied condition. Classifier-free guidance is a representative mechanism for conditional generation \citep{HoSalimans2022}, and a recent study analyzes the statistical behavior of such conditional diffusion models \citep{FuYangWangChen2024}. In time-series applications, conditional diffusion typically uses past observations or partially observed entries as conditioning information for forecasting or imputation \citep{rasul2021,yan2021,tashiro2021csdi,shen2023non}. For stochastic dynamics, \citet{LiuChenXiuZhang2024} study a training-free conditional diffusion method for learning stochastic differential equations, and \citet{gao2025generatingsolutionpathsmarkovian} generate Markovian SDE solution paths by composing conditional diffusion kernels in an autoregressive fashion and controlling the error in KL divergence. Motivated by the widely acknowledged nonstationarity of financial data, our setting requires us to consider the time-series joint law rather than a single conditional kernel. Thus, the conditioning history is itself generated and then reused as input for future kernels. When the history is externally fixed, our result implies an average $\cW_2$ error between conditional kernels; see Corollary \ref{coro:conditional-error}. This paper thus goes beyond existing conditional diffusion models by analyzing the error of a {\it recursively generated path law} in $\cA\cW_2$.

\paragraph{Reinforcement learning for dynamic portfolio selection.}
Our portfolio application is also connected to dynamic mean-variance optimization and reinforcement learning. \citet{TamarDiCastroMannor2016} extend temporal-difference methods to estimate the variance of the reward-to-go under a fixed policy and discuss how this policy-evaluation step may serve as a component of variance-aware policy-optimization methods. Recent studies and surveys document progress as well as challenges in sample efficiency, trading frictions, and risk-aware exploration \citep{jiang2017deep,cong2021alphaportfolio,RLReview}. The role of these tools in our paper is specific. We use a learned distribution $\mQ$ as a generative environment, solve the precommitted portfolio problem under $\mQ$, and prove value stability and policy performance transfer under the data-generating law $\mP$. This links the quality of a learned dynamic simulator to downstream mean-variance performance, rather than treating the simulator and the portfolio-learning method as separate components.
}

\subsection{Organization of the paper}

{ The rest of this paper is organized as follows. In Section \ref{sec:sampling}, we review the static score-based diffusion framework, explain its limitations for dynamic data, and present our adaptive training and sampling method. Section \ref{sec:mainresults} provides a quantitative $\cA\cW_2$ error bound between the output distribution and the data distribution. In Section \ref{stability}, we turn to a specific mean-variance portfolio selection problem and establish model stability in terms of the adapted Wasserstein metric (i.e., the value function is Lipschitz in the adapted Wasserstein metric). Section \ref{sec:app} implements the adaptive sampler and evaluates its generative fidelity and downstream portfolio performance, using a TD3-style policy-gradient algorithm. Section \ref{conclusion} concludes the main text. The appendices include technical discussions and proofs of the theoretical results.}

We close this section by introducing notation used frequently throughout the paper.

\subsection{Frequently used notation}\label{subsec:notation}

{ 
For convenience, Table~\ref{tab:notation-quick-reference} distinguishes
physical time from diffusive time and summarizes the principal objects used
throughout the paper. Precise definitions are provided below.
}

\begin{table}[!htbp]
 
\centering
\footnotesize
\renewcommand{\arraystretch}{1.00}
\begin{tabular}{@{}p{0.29\textwidth}p{0.65\textwidth}@{}}
\toprule
Symbol & Meaning \\
\midrule
$t\in\{1,\ldots,T\}$; $T$; $d$
& Physical-time index, path length, and observation dimension. \\

$\tau\in[0,\Ti]$; $\tau_0$
& Diffusive time, maximum diffusion horizon, and early-stopping level. \\

$X^{1:t},x^{1:t}$; $Y^{1:t},y^{1:t}$
& Clean and generated histories (random variables and their realizations);
portfolio analogues are $S^{1:t},s^{1:t}$ and
$\tilde S^{1:t},\tilde s^{1:t}$. \\

$H^{1:T},h^{1:t}$
& Gaussian-smoothed data path and a realized noisy history. \\

$\mP$; $\mP^{\tau_0}$; $\mQ^{\Ti-\tau_0}$
& Clean path law, noisy path law, and early-stopped sampler law. \\

$\mP_{1:t}$; $\mP_{x^{1:t}}$;
$\mP^{\tau_0}_{h^{1:t}}$
& Clean-history marginal and clean- and noisy-history conditional
next-state laws. \\

$\Pi_{h^{1:t}}^{\tau_0,t+1}$
& Clean next-state posterior given a noisy history; distinct from
$\mP^{\tau_0}_{h^{1:t}}$, the noisy next-state law. \\

$p^{\tau_0}_{t+1}(\tau,\cdot\mid h^{1:t})$;
$\nabla_x\log p^{\tau_0}_{t+1}$
& Forward-VP conditional density and its score. \\

$s_\theta^{t+1}$; $\epsilon_{\rm score}$
& Learned conditional score and prescribed score-matching error level. \\

$\cW_2$; $\Pi_{\rm bc}(\mP,\mQ)$;
$\cA\cW_2(\mP,\mQ)$
& Wasserstein distance, bicausal couplings, and adapted Wasserstein
distance. \\
\bottomrule
\end{tabular}
\caption{Reference table for frequently used notation.}
\label{tab:notation-quick-reference}
\end{table}

In this work, the data take the form of time series of length $T\in \mN$ and dimension $d\in \mN$. We denote by $t\in \{1,2,\cdots, T\}$ the time index. For general theoretical results, we use $X^{1:t}:=(X^1, X^2, \cdots, X^t)$ ($Y^{1:t}:=(Y^1,Y^2,\cdots, Y^t)$) to denote random variables sampled from the data distribution (the approximating model, respectively), and $x^{1:t}$ ($y^{1:t}$) to denote their realizations. For portfolio selection problems, we use $S^{1:t}$ ($\tilde S^{1:t}$) and $s^{1:t}$ ($\tilde s^{1:t}$) instead. For the diffusion model, we use $\tau\in [0,\Ti]$ to denote the {\it diffusive time}, i.e., the time parameter of the forward and reverse SDEs, with $\Ti>0$ denoting the maximum diffusive time. { For the constant schedule $\beta\equiv2$ used in the theoretical analysis,} we write
\[
h_1(\tau):=e^{-\tau},\qquad h_2(\tau):=1-e^{-2\tau}.
\]
{ For the data distribution $\mP$ defined on the time-series space $\mR^{dT}$, $\mP_{1:t}$ is the joint distribution of $(X^1, X^2, \cdots, X^t)$, so $\mP_{1:T}=\mP$. Moreover, $\mP_{x^{1:t}}$ is the conditional distribution of $X^{t+1}$ given the clean data history $X^{1:t}=x^{1:t}$. Because we do not assume that the clean data distribution has a density, we implement an {\it early-stopping} technique. More specifically, we fix a small $\tau_0>0$ throughout the paper and consider the noisy path
\[
H^{1:T}:=h_1(\tau_0)X^{1:T}+\sqrt{h_2(\tau_0)}Z^{1:T},
\qquad Z^{1:T}\sim \N(0,I_{dT\times dT}),
\]
with $Z^{1:T}$ independent of $X^{1:T}\sim\mP$. We denote by $\mP^{\tau_0}:=\cL(H^{1:T})$ the resulting full noisy path law. Thus $\mP^{\tau_0}_{h^{1:t}}$ denotes the regular conditional law of $H^{t+1}$ given $H^{1:t}=h^{1:t}$ under this full noisy law. In practice, early stopping means that we stop the reverse-time process at $\Ti-\tau_0$ and compare the output measure $\mQ^{\Ti-\tau_0}$ with $\mP^{\tau_0}$; see Theorem \ref{thm:AWbounds}. For a noisy history $h^{1:t}$, define the posterior conditional kernel
\[
\Pi_{h^{1:t}}^{\tau_0,t+1}:=\cL(X^{t+1}\mid H^{1:t}=h^{1:t}).
\]
For $\tau\in[\tau_0,\Ti]$, let $p^{\tau_0}_{t+1}(\tau,\cdot|h^{1:t})$ be the density of the forward VP process at diffusive time $\tau$ whose random initial condition has law $\Pi_{h^{1:t}}^{\tau_0,t+1}$. Define $M^{\tau_0}_p(h^{1:t}):=(\mE_{\Pi_{h^{1:t}}^{\tau_0,t+1}}[|X^{t+1}|^p])^{1/p}$, for $p\geq 1$, and $E^{\tau_0}_c(h^{1:t}) := \mE_{\Pi_{h^{1:t}}^{\tau_0,t+1}}[e^{c|X^{t+1}|}]$ for $c>0$. We retain $M_p(x^{1:t})$ and $E_c(x^{1:t})$ for the corresponding clean conditional moments under $\mP_{x^{1:t}}$. By convention, $\mP$ denotes the data-generating model and hence the distribution of $X^{1:T}$ ($S^{1:T}$), while $\mQ$ denotes the alternative model and hence the distribution of $Y^{1:T}$ ($\tilde S^{1:T}$).}

To present the quantitative error bound for the score-based diffusion model, we use the adapted Wasserstein metric. More specifically, for two probability measures $\mP$ and $\mQ$ on $\mR^{dT}$, a coupling $\pi$ between $\mP$ and $\mQ$ is said to be {\it causal in $x$} if for any $t\in\{1,2,\cdots,T\}$,
\begin{align}
    \pi(Y^{1:t}\in \md y^{1:t}|X^{1:T}=x^{1:T}) = \pi(Y^{1:t}\in \md y^{1:t}|X^{1:t}=x^{1:t}).
\end{align}
A coupling $\pi$ is said to be {\it bicausal} if it is causal in both $x$ and $y$ \citep[cf.][]{BBLZ2017,BH23}. Define $\Pi_{\rm bc}(\mP,\mQ)$ to be the set of all bicausal couplings between $\mP$ and $\mQ$. With the choice of transport cost $|x^{1:T}-y^{1:T}|:=(\sum_{t=1}^T|x^t-y^t|^2)^{1/2}$, we define the {\it adapted Wasserstein metric} as follows:
\begin{align}
\cA\cW_2(\mP,\mQ)= \bigg( \inf_{\pi \in \Pi_{\rm bc}(\mP,\mQ)} \int_{\mR^{dT}\times \mR^{dT}} |x^{1:T}-y^{1:T}|^2 \pi(\md x^{1:T},\md y^{1:T})\bigg)^{1/2}.
\end{align}

{ Throughout the paper, we use $L_{\tau_0}$ and $C_{\tau_0}$ to denote generic constants that may vary from line to line. They may depend on $\tau_0$, but are independent of $\Ti$, $\epsilon_{\rm score}$ and $T$. To focus on the trade-off among score matching, noisy initialization, and early stopping, we explicitly write out the dependence on $\tau_0$, $\Ti$ and $\epsilon_{\rm score}$ below, while hiding all dependence on $d$ and the moments of the data distribution. To better understand the scaling with respect to the sequence length, we also retain the dependence on $T$.
}

\section{\texorpdfstring{  Dynamic data generation via diffusion models}{Dynamic data generation via diffusion models}}\label{sec:sampling}

{ 
\subsection{Background on static diffusion models}\label{subsec:static-diffusion}

In the static setting, a score-based diffusion model starts from a data distribution and gradually transports it to a tractable reference distribution. The version used in this paper is the variance-preserving SDE (VP-SDE). Let $p_{\rm data}$ be a probability distribution on $\mR^d$ from which the data are drawn. The VP-SDE is
\begin{align}
\begin{cases}
   & \md X_\tau
=
-\tfrac12\,\beta(\tau)\,X_\tau\,\md\tau
+\sqrt{\beta(\tau)}\,\md B_{\tau},\\
&X_0 \sim p_{\rm data}
\end{cases}\label{vp-sde}
\end{align}
where $\beta(\tau)=\beta_{\min}+(\beta_{\max}-\beta_{\min})\tau/\Ti$ is a prescribed noise schedule. Although $p_{\rm data}$ is unknown, the law of \eqref{vp-sde} converges to the invariant distribution $\N(0,I_{d\times d})$ as $\tau\to\infty$. Sampling from the data distribution is then reduced to solving the reverse-time SDE
\begin{align}
\begin{cases}
&\md Y_\tau = \bigg(\frac{1}{2}\beta(\Ti-\tau)Y_\tau + \beta(\Ti-\tau)\nabla \log p(\Ti-\tau,Y_\tau)\bigg)\md \tau + \sqrt{\beta(\Ti-\tau)}\md \bar B_\tau, \\
&Y_0 \sim \N(0,I_{d\times d}),
\end{cases}
\label{vp-sde-reversed}
\end{align}
where $p(\tau,\cdot)$ denotes the density of $X_\tau$. Here, $B$ and $\bar B$ are $d$-dimensional standard Brownian motions driving the forward and reverse constructions, respectively. When the reverse SDE is simulated as a separate generative process, they may be taken independent. Because $p_{\rm data}$ is unknown, the score $\nabla \log p$ is estimated from data. Under suitable score-matching conditions, an accurate score estimate implies that, for sufficiently large $\Ti$, the law of $Y_\Ti$ is close to $p_{\rm data}$. See \citet{tang2025scorebaseddiffusionmodelsstochastic} for a survey and \citet{gao2025wasserstein} for recent Wasserstein bounds.

This static construction provides the diffusion mechanism used below. We next explain why generating an entire dynamic path as a single static vector does not preserve the information structure required by sequential decision problems.
}

{ 
\subsection{Limitations of existing frameworks}

To motivate our proposed adaptive training and sampling scheme for generating dynamic data, we first discuss the inadequacy of existing results, especially when they are applied to financial (time-series) data and {\it sequential} decision-making.

\subsubsection{Assumptions on data distribution}

The static framework in Section \ref{subsec:static-diffusion} can be applied to time-series data by treating a complete path as a vector in $\mR^{dT}$, in which case all existing results (including algorithms and error bounds) apply directly. Indeed, recent developments provide comprehensive error estimates in terms of total variation distance \citep{chen2023sampling} and Wasserstein distance \citep{chen2023sampling,gao2025wasserstein}, both of which yield convergence guarantees for a sufficiently large diffusive horizon $\Ti$ and a sufficiently small score-matching error $\epsilon_{\rm score}$. These results usually make structural assumptions on the {\it data distribution}. For example, obtaining non-trivial total variation bounds requires the data distribution to admit a density \citep[see][]{Song2021b,chen2023sampling}. On the other hand, Wasserstein bounds rely on either strong log-concavity of the density function \citep{tang2025scorebaseddiffusionmodelsstochastic,gao2025wasserstein} or compact support of the data distribution \citep{chen2023sampling}. However, in the dynamic setting, there are practical scenarios in which $X^{1:T}$ has no joint density function on $\mR^{dT}$ and is not compactly supported.
\begin{example}\label{exm:data:assumption}
Let $X^t\in\mR^d$ be the vector of log prices of $d$ assets. Suppose
\[
X^1\sim \N(\mu_1,\Sigma_1),\qquad
X^{t+1}=a_tX^t+A_t\epsilon^{t+1},\quad t=1,\ldots,T-1,
\]
where $\Sigma_1$ is positive definite, $a_t\in \mR$, the innovations $\epsilon^2,\ldots,\epsilon^T$ are i.i.d.\ $\N(0,I_m)$ and independent of $X^1$, and $A_t\in\mR^{d\times m}$ with $m<d$. This represents a market in which each period introduces only a few systematic shocks, which are then propagated across many assets over time. Unlike a conventional static factor model with full-rank idiosyncratic noise at each date, the low dimensionality here arises from intertemporal dependence and low-dimensional innovations. Note that, in this example, reducing the number of assets to achieve full support is not feasible because the coefficients $a_t$ and $A_t$, as well as the dimension $m$, are unknown. The joint log-price path $X^{1:T}$ is Gaussian and supported on an affine subspace of dimension at most $d+m(T-1)<dT$, so it has no Lebesgue density on $\mR^{dT}$. It also has unbounded support and exponential moments of all orders. Because Appendix \ref{app:assumption:stheta}, particularly Example \ref{exm:posterior:gaussian}, covers jointly Gaussian laws with possibly singular covariance matrices, this dynamic low-dimensional log-price model satisfies the sufficient conditions used to verify Assumption \ref{assumption:P}.
\end{example}

In recent work, \citet{gao2025generatingsolutionpathsmarkovian} propose generating SDE solution paths (on a discrete time grid) via a {\it conditional diffusion model} in an autoregressive fashion (see their Equation~(8) and Remark~1). They obtain an error bound in terms of KL divergence. Because of the special structure of KL divergence, this bound inherently encodes the distance between conditional distributions (Equation~(12) of \citet{gao2025generatingsolutionpathsmarkovian}). Their autoregressive sampling scheme is justified by its numerical feasibility (Remark~1 therein). Although promising for generating discretely sampled SDEs, their results are not directly applicable to more general discrete-time dynamic data because $\mathrm{KL}(\mu|\nu)=\infty$ when $\mu$ and $\nu$ have different supports. Our work complements theirs by showing that this adaptive sampling method also provides error bounds in terms of the adapted Wasserstein distance. Thus, our results illustrate that adaptive sampling is critical for generating dynamic data distributions (possibly beyond SDE solutions) and remains robust when the data distribution is degenerate.

\subsubsection{Instability with respect to ordinary Wasserstein metric}

When the data are indeed compactly supported, e.g., when there is a known uniform bound on possible prices, the density assumption can be relaxed to accommodate degenerate distributions such as the one in Example \ref{exm:data:assumption}, and there are well-established error bounds in Wasserstein distance; see, e.g., \citet{chen2023sampling}. However, the following example demonstrates that, even for compactly supported prices (in fact, prices supported on very few possible values), working under an alternative model that is Wasserstein-close to the real one can lead to different optimal decisions.

\begin{example}\label{exm:instability}
This example is adapted from \citet{BBBE2020}. Let $S_0=1$ and the risk-free rate $r=0$. Suppose that, under the data-generating model $\mP$,
\begin{align}
    (S_1,S_2) =
    \begin{cases}
      (1,2), & \text{with probability } 1/4, \\
      (1,1), & \text{with probability } 1/2,\\
      (1,0), &\text{with probability } 1/4.
    \end{cases}
\end{align}
Suppose also that, under an alternative model $\mP_n$,
\begin{align}
    (S_1,S_2) =
    \begin{cases}
      (1+1/n,2), & \text{with probability } \frac14(1-\frac1{n^2}), \\
      (1+1/n,1), & \text{with probability } \frac14(1-\frac1{n^2}),\\
      (1-1/n,1), &\text{with probability } \frac14(1-\frac1{n^2}),\\
      (1-1/n,0), &\text{with probability } \frac14(1-\frac1{n^2}),\\
      (1,2), &\text{with probability } 3/(4n^2),\\
      (1,0), &\text{with probability } 1/(4n^2).
    \end{cases}
\end{align}
See Figure \ref{fig:exm2} for the path diagrams. Both $\mP$ and $\mP_n$ are finite-state tree models and satisfy the no-arbitrage condition. Consider the following transport plan from $\mP_n$ to $\mP$:
\begin{align}
    &(1,1+1/n,2) \to (1,1,2),\\
    &(1,1+1/n,1) \text{ and } (1,1-1/n,1) \to (1,1,1),\\
    &(1,1-1/n,0) \to (1,1,0),\\
    &(1,1,2) \to (1,1,2) \text{ with mass } 1/(4n^2),\\
    &(1,1,2) \to (1,1,1) \text{ with mass } 1/(2n^2),\\
    &(1,1,0) \to (1,1,0).
\end{align}
Each arrow without a stated mass transports the entire mass of its source path. Because this coupling incurs a time-$1$ cost of $(1-1/n^2)/n^2$ and a time-$2$ cost of $1/(2n^2)$, we have
\begin{align}
    \cW_2^2(\mP,\mP_n)
    \leq \frac{3}{2n^2}-\frac{1}{n^4}\to 0
\end{align}
as $n\to\infty$.

Under a bicausal coupling, the time-$1$ branch under $\mP_n$ must be selected without observing the terminal outcome under $\mP$. At each of the three time-$1$ nodes under $\mP_n$, the squared Wasserstein distance between the corresponding terminal kernel and the terminal law under $\mP$ is $1/2$. Hence, after dropping the nonnegative time-$1$ cost,
\begin{align}
    \cA\cW_2^2(\mP,\mP_n)\ge \frac{1}{2}.
\end{align}

\begin{figure}[htbp]
\centering
 
\begin{subfigure}{0.46\textwidth}
\centering
\begin{tikzpicture}[
    x=2.2cm,
    y=1.4cm,
    node/.style={circle,fill=black,inner sep=1.4pt},
    edge/.style={thick},
    lab/.style={font=\small},
    prob/.style={font=\scriptsize,midway,sloped,above,
                 fill=white,inner sep=0.6pt}
]

\node[node,label=left:{\(\scriptstyle S_0=1\)}] (P0) at (0,0) {};
\node[node,label=above:{\(\scriptstyle S_1=1\)}] (P1) at (1,0) {};

\node[node,label=right:{\(\scriptstyle S_2=2\)}] (P2u) at (2,1) {};
\node[node,label=right:{\(\scriptstyle S_2=1\)}] (P2m) at (2,0) {};
\node[node,label=right:{\(\scriptstyle S_2=0\)}] (P2d) at (2,-1) {};

\draw[edge] (P0) -- (P1) node[prob] {\(1\)};
\draw[edge] (P1) -- (P2u) node[prob] {\(1/4\)};
\draw[edge] (P1) -- (P2m) node[prob] {\(1/2\)};
\draw[edge] (P1) -- (P2d) node[prob,below] {\(1/4\)};

\node[lab] at (0,-1.45) {\(t=0\)};
\node[lab] at (1,-1.45) {\(t=1\)};
\node[lab] at (2,-1.45) {\(t=2\)};

\end{tikzpicture}
\caption{Real model \(\mP\)}
\end{subfigure}
\hfill
\begin{subfigure}{0.46\textwidth}
\centering
\begin{tikzpicture}[
    x=2.15cm,
    y=0.82cm,
    node/.style={circle,fill=black,inner sep=1.4pt},
    edge/.style={thick},
    lab/.style={font=\small},
    prob/.style={font=\scriptsize,midway,sloped,above,
                 fill=white,inner sep=0.6pt}
]

\node[node,label=left:{\(\scriptstyle S_0=1\)}] (Q0) at (0,0) {};

\node[node,label=above left:{\(\scriptstyle S_1=1+\frac1n\)}] (Q1u) at (1,1.6) {};
\node[node,label=above:{\(\scriptstyle S_1=1\)}] (Q1m) at (1,0) {};
\node[node,label=below left:{\(\scriptstyle S_1=1-\frac1n\)}] (Q1d) at (1,-1.6) {};

\node[node,label=right:{\(\scriptstyle S_2=2\)}] (Q2uu) at (2,2.2) {};
\node[node,label=right:{\(\scriptstyle S_2=1\)}] (Q2um) at (2,1.05) {};
\node[node,label=right:{\(\scriptstyle S_2=2\)}] (Q2mu) at (2,0.4) {};
\node[node,label=right:{\(\scriptstyle S_2=0\)}] (Q2md) at (2,-0.4) {};
\node[node,label=right:{\(\scriptstyle S_2=1\)}] (Q2dm) at (2,-1.05) {};
\node[node,label=right:{\(\scriptstyle S_2=0\)}] (Q2dd) at (2,-2.2) {};

\draw[edge] (Q0) -- (Q1u)
    node[prob] {\(\frac12(1-\frac1{n^2})\)};
\draw[edge] (Q0) -- (Q1m)
    node[prob] {\(\frac1{n^2}\)};
\draw[edge] (Q0) -- (Q1d)
    node[prob,below] {\(\frac12(1-\frac1{n^2})\)};

\draw[edge] (Q1u) -- (Q2uu) node[prob] {\(1/2\)};
\draw[edge] (Q1u) -- (Q2um) node[prob,below] {\(1/2\)};
\draw[edge] (Q1m) -- (Q2mu) node[prob] {\(3/4\)};
\draw[edge] (Q1m) -- (Q2md) node[prob,below] {\(1/4\)};
\draw[edge] (Q1d) -- (Q2dm) node[prob] {\(1/2\)};
\draw[edge] (Q1d) -- (Q2dd) node[prob,below] {\(1/2\)};

\node[lab] at (0,-2.7) {\(t=0\)};
\node[lab] at (1,-2.7) {\(t=1\)};
\node[lab] at (2,-2.7) {\(t=2\)};

\end{tikzpicture}
\caption{Alternative model \(\mP_n\)}
\end{subfigure}

\caption{Path diagrams of $\mP$ and $\mP_n$. Edge labels are conditional transition probabilities.}
\label{fig:exm2}
\end{figure}

Let us now consider mean-variance portfolio selection under these two models, with initial endowment $x_0=1$. We take the numbers of shares of $S$ held at $t=0,1$ as the decision variables and denote them by $\theta = (\theta_0,\theta_1)$. Clearly, under $\mP$, $\theta_0$ is irrelevant, and for any portfolio $\theta$, the terminal wealth $X^\theta_2$ has the following distribution:
\begin{align}
    X^\theta_2 = \begin{cases}
        1+\theta_1, & \text{with probability } 1/4,\\
        1, & \text{with probability } 1/2,\\
        1-\theta_1, &\text{with probability } 1/4.
    \end{cases}
\end{align}
Therefore, with risk aversion $\gamma>0$, the mean-variance objective under $\theta$ is $J(\theta) = 1-\frac{\gamma\theta_1^2}{4}$, and the non-participation decision yields the optimal mean-variance value $v(\mP)=1$.

On the other hand, under $\mP_n$, let $\theta_1^+$, $\theta_1^0$, and $\theta_1^-$ denote the time-$1$ risky holdings after observing $S_1=1+1/n$, $S_1=1$, and $S_1=1-1/n$, respectively. Write
\begin{align}
    c=\frac{\theta_1^++\theta_1^-}{2},\qquad
    k=\frac{\theta_1^+-\theta_1^-}{2},\qquad
    q=\frac{\theta_0}{n}.
\end{align}
A direct calculation gives
\begin{align}
    \mE[X_2^\theta]
    &=1+\left(1-\frac{1}{n^2}\right)
    \left(\frac12-\frac1n\right)k
    +\frac{\theta_1^0}{2n^2},\\
    \Var(X_2^\theta)
    &=\left(1-\frac{1}{n^2}\right)
    \left[
    \frac{c^2+k^2}{4}
    +\left(q+\left(\frac12-\frac1n\right)c\right)^2
    \right]
    +\frac{3(\theta_1^0)^2}{4n^2}\\
    &\quad+\frac{1}{n^2}\left(1-\frac{1}{n^2}\right)
    \left[
    \left(\frac12-\frac1n\right)k-\frac{\theta_1^0}{2}
    \right]^2.
\end{align}
Consequently, for $n>2$, the mean-variance objective is maximized at
\begin{align}
    q^*=c^*=0,\qquad
    k^*
    &=\frac{4n^3(n-2)}
    {\gamma(2n^4+n^2-6n+6)},\\
    \theta_1^{0,*}
    &=\frac{2n^2(n^2-2n+2)}
    {\gamma(2n^4+n^2-6n+6)}.
\end{align}
Equivalently, the optimal holdings on the three $\mP_n$-reachable time-$1$ nodes are uniquely given by
\begin{align}
    \theta_0^*=0,\qquad
    \theta_1^*=
    \begin{cases}
    k^*, & S_1=1+1/n,\\
    \theta_1^{0,*}, & S_1=1,\\
    -k^*, & S_1=1-1/n.
    \end{cases}
\end{align}
The corresponding optimal value is
\begin{align}
    v(\mP_n)
    =1+\frac{1}{2}\left(1-\frac{1}{n^2}\right)
    \left(\frac12-\frac1n\right)k^*
    +\frac{\theta_1^{0,*}}{4n^2}.
\end{align}
In particular, $k^*\to2/\gamma$, $\theta_1^{0,*}\to1/\gamma$, and $v(\mP_n)-v(\mP)\to1/(2\gamma)$. When the same feedback strategy is deployed under $\mP$, the state $S_1=1$ occurs with probability one, and its mean-variance value is
\begin{align}
    1-\frac{\gamma(\theta_1^{0,*})^2}{4}
    \to 1-\frac{1}{4\gamma}.
\end{align}
Thus, although $\cW_2(\mP_n,\mP)\to0$, the performance loss of the $\mP_n$-optimal feedback strategy under $\mP$ satisfies
\begin{align}
    v(\mP)-\left[1-\frac{\gamma(\theta_1^{0,*})^2}{4}\right]
    =\frac{\gamma(\theta_1^{0,*})^2}{4}
    \to\frac{1}{4\gamma}>0.
\end{align}
In particular, both the optimal-value gap and the policy-performance loss remain bounded away from zero.
\end{example}

Example \ref{exm:instability} demonstrates that {\it sequential decision-making} is unstable under model deviations measured by the usual Wasserstein distance. The intuition behind this phenomenon is that a general transport plan between two models does not necessarily respect the information structure (or conditional distributions), whereas sequential decision-making crucially exploits additional information as time evolves. In Example \ref{exm:instability}, a small Wasserstein perturbation changes the information at $t=1$ and creates the illusory signals $S_1=1\pm 1/n$. As a consequence, the $\mP_n$-optimal policy incurs a nonvanishing performance loss under $\mP$. This example motivates us to consider the adapted Wasserstein distance for error analysis in this section and to investigate the stability of dynamic mean-variance portfolio selection with respect to the adapted Wasserstein distance in Section \ref{stability}.

To conclude this subsection, Table \ref{tab:comparison-existing-frameworks} compares the preceding results with ours:

\begin{table}[htbp]
\centering
\begin{tabular}{lcccc}
\toprule
Reference & Data type & Metric & Has density & Bounded support \\
\midrule
\citet{chen2023sampling} & Static  & $\TV$ & Yes & No \\
\citet{chen2023sampling}  & Static & $\cW_2$ & No & Yes \\
\citet{gao2025wasserstein} & Static & $\cW_2$ & Yes & No \\
\citet{factormodel} & Static & $\TV$ & Yes & No\\
\citet{gao2025generatingsolutionpathsmarkovian} & Dynamic (SDE solutions) & $\KL$ & Yes  & No \\
This paper & Dynamic (discrete-time) & $\cA\cW_2$ & No & No \\
\bottomrule
\end{tabular}
\small
\caption{Comparison of existing error bounds and our framework.}
\label{tab:comparison-existing-frameworks}
\end{table}
}

\subsection{Conditional score matching}\label{subsec:conditional-score-matching}

{ 
For an observed clean history $x^{1:t}$, the law $\mP_{x^{1:t}}$ is the ideal conditional next-state distribution required for generation. The formal early-stopped analysis conditions instead on the noisy history $H^{1:t}=h^{1:t}$ and uses the posterior $\Pi_{h^{1:t}}^{\tau_0,t+1}$. For each $t=1,2,\cdots,T-1$, consider the conditional VP-SDE
\begin{align}
\begin{cases}
&\md X_\tau^{t+1}
=-\tfrac12\beta(\tau)X_\tau^{t+1}\md\tau
+\sqrt{\beta(\tau)}\md B_\tau^{t+1},\\
&X_0^{t+1}\sim\Pi_{h^{1:t}}^{\tau_0,t+1}
=\cL(X^{t+1}\mid H^{1:t}=h^{1:t}).
\end{cases}
\label{vp-sde-conditional}
\end{align}
Its density at diffusive time $\tau$ is denoted by $p^{\tau_0}_{t+1}(\tau,\cdot\mid h^{1:t})$. The corresponding reverse-time SDE is obtained from \eqref{vp-sde-reversed} using this conditional density and its score. For simplicity, the theoretical analysis below takes $\beta\equiv2$. The same arguments extend naturally to a time-varying VP schedule after replacing the corresponding transition coefficients.
}

To facilitate our adaptive sampling method, we assume that we have access to a family of pretrained approximate score functions $s^{t}_{\theta}:[0,\Ti]\times \mR^{dt}\to \mR^d$, $t=1,2,\cdots, T$, satisfying the following {\it noisy-history conditional score-matching error bounds}:

\begin{assumption}\label{assumption:score-matching-error}
	For any $\tau\in [\tau_0,\Ti]$ and $t=1,2,\cdots, T-1$, we have
	 	\begin{align}\label{score-error-initial}
		&\mE_{X\sim p_1(\tau,\cdot)}[|s^1_{\theta}(\tau,X)- \nabla_x \log p_1(\tau,X)|^2]\leq \epsilon^2_{\rm score},\\
		&\mE_{H^{1:t}\sim \mP^{\tau_0}_{1:t}}\big[\mE_{X_\tau^{t+1}\sim p^{\tau_0}_{t+1}(\tau,\cdot|H^{1:t})}[|s^{t+1}_\theta(\tau,H^{1:t},X^{t+1}_\tau)-\nabla_x \log p^{\tau_0}_{t+1}(\tau,X_\tau^{t+1}|H^{1:t})|^2]\big]\leq \epsilon^2_{\rm score}.\label{score-error}
	\end{align}
Here, $p^{\tau_0}_{t+1}(\tau,\cdot|h^{1:t})$ is the probability density function of the forward process whose random initial condition has law $\Pi_{h^{1:t}}^{\tau_0,t+1}=\cL(X^{t+1}\mid H^{1:t}=h^{1:t})$.
\end{assumption}

\begin{figure} 
\centering
\resizebox{0.95\textwidth}{!}{%
\begin{tikzpicture}[
    >=Latex,
    every node/.style={font=\small},
    lab/.style={font=\footnotesize, inner sep=1pt}
]
    \node (x1) at (0,1.4) {$X^1$};
    \node (x2) at (2.0,1.4) {$X^2$};
    \node (xdots) at (4.0,1.4) {$\cdots$};
    \node (xt) at (6.0,1.4) {$X^t$};
    \node (xp) at (9.1,1.4) {$X^{t+1}$};
    \node (xmore) at (10.8,1.4) {$\cdots$};

    \node (h1) at (0,0) {$H^1$};
    \node (h2) at (2.0,0) {$H^2$};
    \node (hdots) at (4.0,0) {$\cdots$};
    \node (ht) at (6.0,0) {$H^t$};
    \node (hp) at (9.1,0) {$H^{t+1}$};
    \node (hmore) at (10.8,0) {$\cdots$};

    \node (q1) at (0,-1.4) {$X^1_\tau$};
    \node (q2) at (2.0,-1.4) {$X^2_\tau$};
    \node (qdots) at (4.0,-1.4) {$\cdots$};
    \node (qt) at (6.0,-1.4) {$X^t_\tau$};
    \node (qp) at (9.1,-1.4) {$X^{t+1}_\tau$};
    \node (qmore) at (10.8,-1.4) {$\cdots$};

    \draw[->] (x1) -- (x2);
    \draw[->] (x2) -- (xdots);
    \draw[->] (xdots) -- (xt);
    \draw[->] (xt) -- node[lab, above] {$\mP_{x^{1:t}}$} (xp);
    \draw[->] (xp) -- (xmore);

    \draw[->] (h1) -- (h2);
    \draw[->] (h2) -- (hdots);
    \draw[->] (hdots) -- (ht);
    \draw[->] (ht) -- node[lab, below] {$\mP^{\tau_0}_{h^{1:t}}$} (hp);
    \draw[->] (hp) -- (hmore);

    \draw[->] (q1) -- (q2);
    \draw[->] (q2) -- (qdots);
    \draw[->] (qdots) -- (qt);
    \draw[->] (qt) -- (qp);
    \draw[->] (qp) -- (qmore);

    \draw[->] (x1) -- (h1);
    \draw[->] (x2) -- (h2);
    \draw[->] (xt) -- (ht);
    \draw[->] (xp) -- (hp);

    \draw[->] (x1.south east) to[out=-65,in=65] node[lab, right, pos=0.25] {$\phi(\tau,\cdot|X^1)$} (q1.north east);
    \draw[->] (x2.south east) to[out=-65,in=65] (q2.north east);
    \draw[->] (xt.south east) to[out=-75,in=75] (qt.north east);
    \draw[->] (xp.south east) to[out=-60,in=60] node[lab, right, pos=0.25] {$\phi(\tau,\cdot|X^{t+1})$} (qp.north east);

    \draw[->] (ht) to[bend left=10] node[lab, pos=0.52, above] {$\Pi_{h^{1:t}}^{\tau_0,t+1}$} (xp);
\end{tikzpicture}
}
\caption{{ The clean data variables $X$, the noisy history $H$, and the VP-diffused variables $X_\tau$.}}
\label{fig:noisy-history-kernels}
\end{figure}

We remark that \eqref{score-error-initial} is the classical score error bound for the first marginal of $\mP$; see \citet{Song2021b}. Due to the early-stopping technique, the score-matching error \eqref{score-error} is imposed along noisy histories $H^{1:t}\sim \mP^{\tau_0}_{1:t}$, and the target is the posterior-mixture score $\nabla_x\log p^{\tau_0}_{t+1}(\tau,\cdot|H^{1:t})$. { For clarity, Figure \ref{fig:noisy-history-kernels} summarizes the relationships among the various random variables and conditional kernels.}

Although we use the noisy history variables $H$ in \eqref{score-error}, the MSE objective remains fully trainable from joint data by denoising score matching; see Proposition \ref{prop:DSM} below. The proof idea is borrowed from classical denoising score matching \citep[e.g.,][]{hyvarinen2005score}, with the necessary modification that the conditioning variable is the noisy history. We provide the proof in Appendix \ref{app:proof:conditionalDSM} for completeness.

\begin{proposition}\label{prop:DSM}
For any $t=1,\ldots,T-1$ and $\tau\in[\tau_0,\Ti]$, let $H^{1:t}=h_1(\tau_0)X^{1:t}+\sqrt{h_2(\tau_0)}Z^{1:t}$, where $Z^{1:t}\sim \N(0,I_{dt\times dt})$ is independent of $X^{1:t+1}\sim \mP_{1:t+1}$. The following two score-matching problems are equivalent:
\begin{alignat}{2}
&\min_{\theta}
&&\mE_{H^{1:t}\sim \mP^{\tau_0}_{1:t}}\big[\mE_{X_\tau^{t+1}\sim p^{\tau_0}_{t+1}(\tau,\cdot|H^{1:t})}[|s^{t+1}_\theta(\tau,H^{1:t},X_{\tau}^{t+1})-\nabla_x \log p^{\tau_0}_{t+1}(\tau,X_\tau^{t+1}|H^{1:t})|^2]\big]\label{conditional-ESM}\\
&\min_{\theta}
&&\mE_{X^{1:t+1}\sim \mP_{1:t+1}}\Big[\mE_{Z^{1:t}}\big[\mE_{X_\tau^{t+1}\sim \phi(\tau,\cdot|X^{t+1})}[|s^{t+1}_\theta(\tau,H^{1:t},X_{\tau}^{t+1})\notag\\
&&&\hspace{17em}{}-\nabla_x \log \phi(\tau,X_\tau^{t+1}|X^{t+1})|^2]\big]\Big].\label{conditional-DSM}
\end{alignat}
Here, $\phi(\tau,\cdot|x_0)$ is the probability density function of the forward process, with initial condition $X_0=x_0$. In particular,
\[
\phi(\tau,x\mid x_0)
=(2\pi h_2(\tau))^{-d/2}
\exp\!\left(-\frac{|x-h_1(\tau)x_0|^2}{2h_2(\tau)}\right).
\]
\end{proposition}
\begin{proof}
    See Appendix \ref{app:proof:conditionalDSM}.
\end{proof}

\subsection{Adaptive sampling algorithm}

Because of the explicit form of $\phi(\tau,x|x_0)$, $X^{t+1}_\tau\sim \phi(\tau,\cdot|X^{t+1})$ is equivalent to $X^{t+1}_\tau = h_1(\tau)X^{t+1}+\sqrt{h_2(\tau)}\bz$, where $\bz\sim \N(0,I)$ is independent of the other random variables. Therefore, \eqref{conditional-DSM} is further equivalent to
\begin{equation}\label{conditional:DSM:noise}
\min_{\theta} \mE_{X^{1:t+1}\sim \mP_{1:t+1}}\mE_{Z^{1:t}}\mE_{\bz\sim \N(0,I)} \big| \sqrt{h_2(\tau)}s^{t+1}_\theta\big(\tau,H^{1:t},h_1(\tau)X^{t+1}+\sqrt{h_2(\tau)}\bz \big)+\bz\big|^2,
\end{equation}
where $H^{1:t}=h_1(\tau_0)X^{1:t}+\sqrt{h_2(\tau_0)}Z^{1:t}$. In the implementation, $s_\theta$ is trained to minimize the empirical version of \eqref{conditional:DSM:noise}, with additional averaging over $\tau$ to account for different diffusive times; see \eqref{conditional:DSM:implementation}.

Suppose the true model for time-series data is $\mP$, and we have pretrained approximate score functions $s_\theta^t$ satisfying \eqref{score-error-initial} and \eqref{score-error}. We note that, for a fixed noisy history $h^{1:t}$, $s^{t+1}_\theta(\tau,h^{1:t},\cdot)$ serves as an approximation to the noisy-history conditional score function $\nabla _x\log p^{\tau_0}_{t+1}(\tau,\cdot|h^{1:t})$. This observation motivates adaptive sampling from the output distribution; see Algorithm \ref{adaptive-sampling}. In the next section, we present and discuss our main theoretical results for this adaptive, autoregressive-style sampling algorithm.

\begin{algorithm}
\caption{Adaptive sampling}

\textbf{Inputs}: pretrained approximate score functions $s^{t}_{\theta}$, $t=1,2,\cdots,T$; a reverse-time SDE simulator.

\textbf{Initialization}: independently sample
$(Z_{(n)}^1,\ldots,Z_{(n)}^T)\sim
\N(0_{\mR^{dT}},I_{dT\times dT})$ for $n=1,\ldots,N$.

\begin{algorithmic}[1]
\FOR{$n=1$ \TO $N$}

\STATE{
With initial condition $\bar Y_0=Z_{(n)}^1$, run the reverse-time SDE using the score function { $s_\theta^1$}:
}
\begin{align}
\md \bar Y_\tau = [\bar Y_\tau +2s_\theta^{1}(\Ti-\tau,\bar Y_\tau)]\md \tau { + \sqrt{2}\md \bar B_\tau}.
\label{condition-reversed-1}
\end{align}

\STATE{
$y_{(n)}^1\leftarrow \bar Y_{\Ti-\tau_0}$.
}

\FOR{$t=1$ \TO $T-1$}
\STATE{$y^{1:t}\leftarrow y^{1:t}_{(n)}$}
\STATE{
With initial condition { $\bar Y_{0}=Z_{(n)}^{t+1}$}, run the reverse-time SDE using the score function $(\tau,x)\mapsto s^{t+1}_\theta(\tau,y^{1:t},x)$:
}
\begin{align}
\md \bar Y_\tau = [\bar Y_\tau +2s_\theta^{t+1}(\Ti-\tau,y^{1:t},\bar Y_\tau)]\md \tau { + \sqrt{2}\md \bar B_\tau}.\label{condition-reversed-t}
\end{align}

\STATE{
$y_{(n)}^{t+1}\leftarrow \bar Y_{\Ti-\tau_0}$.
}

\ENDFOR
\ENDFOR
\end{algorithmic}

\textbf{Output}: $\{(y_{(n)}^1,y_{(n)}^2,\cdots, y_{(n)}^T)\}_{n=1}^N$.

\label{adaptive-sampling}
\end{algorithm}

\section{Main results}\label{sec:mainresults}
We now prove that the sampling method described by Algorithm \ref{adaptive-sampling} gives a distribution $\mQ$ close to $\mP$ in $\cA\cW_2$. Our proof relies on existing total variation bounds for conditional distributions generated by diffusion models and on a uniform tail estimate for SDEs. We need the following technical assumptions.

\begin{assumption}[\bf Assumptions on the data distribution] \label{assumption:P}
\ \\
\vspace{-20pt}
\begin{enumerate}
    \item  There exists a constant $L_{\tau_0}>0$ depending on $\tau_0$ such that for any $\tau\in [\tau_0,\Ti]$, $x,y\in \mR^d$, $t\in \{1,2,\cdots,T-1\}$ and $h^{1:t},k^{1:t}\in \mR^{dt}$,
\begin{align}
   &|\nabla _x \log p_{1}(\tau,x)-\nabla _x \log p_{1}(\tau,y)|\leq L_{\tau_0}|x-y|, \label{Lip:forwardscore}\\
   &|\nabla _x \log p^{\tau_0}_{t+1}(\tau,x|h^{1:t})-\nabla _x \log p^{\tau_0}_{t+1}(\tau,y|k^{1:t})|\leq L_{\tau_0}\big((1+|x|)|h^{1:t}-k^{1:t}|+|x-y|\big).\label{Lip:conditionalscore}
   \end{align}
   \item There exists a constant $C_{\tau_0}>0$ such that for any $t\in \{1,2,\cdots,T-1\}$, $\tau\in [\tau_0,\Ti]$ and $h^{1:t}\in \mR^{dt}$,
   \begin{align}
       |\nabla _x \log p^{\tau_0}_{t+1}(\tau,x|h^{1:t})|\leq C_{\tau_0}(1+|x|+|h^{1:t}|).\label{conditionalscore:growth}
   \end{align}
    \item  There exist constants $c>0$ and $K>0$ such that
    \begin{align}
 \sup_{t\in \{1,2,\cdots, T\}}\mE_{\mP}[e^{c|X^t| }]\leq K.
 \end{align}
 \item There exists a constant $K>0$ such that for any $t\in \{1,2,\cdots, T-1\}$ and $x^{1:t},y^{1:t}\in \mR^{dt}$,
 \begin{align}
     \cW_2(\mP_{x^{1:t}},\mP_{y^{1:t}})\leq K|x^{1:t}-y^{1:t}|.
 \end{align}

\end{enumerate}

 \end{assumption}

\begin{assumption}[\bf Assumptions on the approximating network] \label{assumption:stheta}
\ \\
\vspace{-20pt}
\begin{enumerate}
    \item There exists a constant $L_{\tau_0}>0$ such that for any $\tau\in [\tau_0,\Ti]$, $x,y\in \mR^d$, $t\in \{1,2,\cdots,T-1\}$ and $h^{1:t},k^{1:t}\in \mR^{dt}$,
\begin{align}
&|s^{1}_\theta(\tau,x)-s^{1}_\theta(\tau,y)|\leq L_{\tau_0}|x-y|,\\
&|s^{t+1}_\theta(\tau,h^{1:t},x)-s^{t+1}_\theta(\tau,k^{1:t},y)|\leq L_{\tau_0}(1+|x|)(|h^{1:t}-k^{1:t}|+|x-y|).	\label{Lip:stheta}
\end{align}

%\item { 
%   There exists a constant $C_{\tau_0}>0$ such that for any $t\in \{1,2,\cdots,T\}$, $\tau\in [\tau_0,\Ti]$ and $x^{1:t}\in \mR^{dt}$,
%   \begin{align}
%       |s_\theta(\tau,x^{1:t-1},x)|\leq C_{\tau_0}(1+|x|+|x^{1:t}|).
%   \end{align}}

\item There exist constants $\delta>0$ and $D_{\tep}>0$ (which may depend on $\tau_0$, $\epsilon_{\rm score}$ and $T$) such that for any $t\in \{1,2,\cdots,T-1\}$, $\tau\in [\tau_0,\Ti]$, $h^{1:t}\in \mR^{dt}$ and $x\in \mR^d$,
\begin{align}
	&2 x\cdot s^{1}_\theta(\tau,x)\leq -(1+\delta)|x|^2 +{D_{\tep}},\\
    &2 x\cdot s^{t+1}_\theta(\tau,h^{1:t},x)\leq -(1+\delta)|x|^2 +D_{\tep}.
\end{align}
\end{enumerate}
\end{assumption}

% { Todo: It is better to make the dependence on $D$ explicit because it depends on $\epsilon_{\rm score}$ when constructing networks. It turns out an additional $e^{cM_{\rm diss}}$ in $\alpha(\Ti)$. Finally we will take $M_{\rm diss}\sim \log(1/\epsilon_{\rm score})$ so everything is fine.
% }

\begin{remark}
\begin{enumerate}
\item {  Although Lipschitz continuity assumptions on the forward-process score functions, such as \eqref{Lip:forwardscore}, are fairly standard in the literature (for example, see Assumption A1 of \citet{chen2023sampling} and Assumption 1 of \citet{gao2025wasserstein}), the Lipschitz continuity in the noisy observation variable $h^{1:t}$ required by \eqref{Lip:conditionalscore}, together with the growth condition \eqref{conditionalscore:growth}, appears to be new. In Appendix \ref{app:assumption:stheta}, we show that these requirements can be interpreted as stability and growth conditions on the posterior mean under $\Pi_{h^{1:t}}^{\tau_0,t+1}$ when the forward process is observed with noise. The Wasserstein stability of the clean conditional kernels in Assumption \ref{assumption:P}-4 is used separately in the proof of Theorem \ref{thm:AWbounds} to obtain the first-order smoothing estimate $\cA\cW_2^2(\mP,\mP^{\tau_0})\leq (K')^T\tau_0$ for some $K'>1$. We also show in Appendix \ref{app:assumption:stheta} that certain weak {\it log-concavity} conditions on $\Pi_{h^{1:t}}^{\tau_0,t+1}$, which are weaker than the usual strong log-concavity assumption on the density, make the Lipschitz constant with respect to $x$ uniform in $h^{1:t}$. In the finite-state and Gaussian examples, the posterior kernel can be computed directly from the data law. The non-Gaussian exponential example instead provides a tractable check of the sufficient conditions at the posterior level.

More generally, Assumption \ref{assumption:P}-1 imposes convenient technical conditions for our error analysis. We leave their relaxation to future work.}

\item One generally seeks to avoid assumptions directly on the approximation network because they may restrict its expressive power. Here, however, we show that under Assumption \ref{assumption:P}, which concerns only the data distribution, the approximating network $s_\theta$ can be chosen to satisfy Assumption \ref{assumption:stheta} (with an appropriate choice of $D_{\tep}$) while preserving the approximation-error bounds \eqref{score-error-initial} and \eqref{score-error}; see Proposition \ref{prop:assumption:stheta} below.
	    To keep the main development focused, we defer a detailed discussion of this result to Appendix \ref{app:assumption:stheta}. {  As $\epsilon_{\rm score}\to 0$, the required score matching becomes increasingly accurate, and the constant $D_{\tep}$ may diverge. Some deterioration of this kind appears unavoidable without strong log-concavity of the data distribution. Nevertheless, the resulting error bound makes the trade-off among $\Ti$, $\tau_0$ and $\epsilon_{\rm score}$ explicit and yields a convergence guarantee under a suitable joint scaling of these parameters. In particular, Proposition \ref{prop:assumption:stheta} shows that the growth of $D_{\tep}$ is polylogarithmic (quadratic in $\log(\epsilon_{\rm score}^{-1})$ for fixed $\tau_0$ and $T$), so its contribution to the final bound is comparatively mild. See also Remark \ref{remark:mainthm}-(1).}
\end{enumerate}
\end{remark}

\begin{proposition}\label{prop:assumption:stheta}
Suppose Assumption \ref{assumption:P} holds. For any $\tau_0>0$ and $\epsilon_{\rm score}>0$, there exists an $s_\theta$ satisfying Assumption \ref{assumption:score-matching-error} as well as Assumption \ref{assumption:stheta} with

\[
D_{\tep}=C\cdot C_{\tau_0}^2T\Big(\log(eT)+\log\big(e+h_2(\tau_0)^{-2}\epsilon^{-2}_{\rm score}\big)\Big)^2,
\]
where $C_{\tau_0}$ is the growth constant in \eqref{conditionalscore:growth} and $C>0$ is independent of $\tau_0$, $\Ti$, $\epsilon_{\rm score}$ and $T$.
\end{proposition}

We are now ready to present one of our main results: quantitative error bounds between the output distribution of the adaptive sampling method and the true data distribution. Recall the notation for noisy-history conditional moments: $M^{\tau_0}_p(h^{1:t}):=(\mE_{\Pi_{h^{1:t}}^{\tau_0,t+1}}[|X^{t+1}|^p])^{1/p}$, for $p\geq 1$, and $E^{\tau_0}_c(h^{1:t}) = \mE_{\Pi_{h^{1:t}}^{\tau_0,t+1}}[e^{c|X^{t+1}|}]$.
\begin{theorem} \label{thm:AWbounds}
		For $t=1,\ldots,T-1$, denote by $\{Y^{y^{1:t}}_{\tau}\}_{\tau\in [0,\Ti]}$ the backward process with initial condition $Y^{y^{1:t}}_{0}\sim \N(0_{\mR^d},I_{d\times d})$ and score function $(\tau,x)\mapsto s^{t+1}_{\theta}(\tau,y^{1:t},x)$, i.e., $Y^{y^{1:t}}$ satisfies \eqref{condition-reversed-t}. Define $\mQ^{\Ti-\tau_0}_{y^{1:t}}=\cL(Y^{y^{1:t}}_{\Ti-\tau_0})$. Moreover, define $\mQ^{\Ti-\tau_0}_1=\cL(Y_{\Ti-\tau_0})$ with $Y$ satisfying \eqref{condition-reversed-1}.
	\begin{itemize}
\item[(1)] There exists a constant $C_{\tau_0}>0$, which may depend on $\tau_0>0$, such that, for $t=1,2,\cdots,T-1$, we have
\begin{align}
\cW_2^2\big(\mP^{\tau_0}_{h^{1:t}},\mQ^{\Ti-\tau_0}_{y^{1:t}}\big)\leq & C_{\tau_0}\Big(D_{\tep}e^{-\Ti}\big(1+M^{\tau_0}_2(h^{1:t})\big)+e^{-c'D_{\tep}} \\
&+ D_{\tep}\E^{\tau_0}(\tau_0,h^{1:t})^{1/2} \\
&+e^{-c'\sqrt{D_{\tep}}}\big(1+E^{\tau_0}_c(h^{1:t})\big) \\
&+ D_{\tep}\Ti^{1/2}\big(1+M^{\tau_0}_2(h^{1:t})\big)|h^{1:t}-y^{1:t}|\Big),
		\label{conditional-bound}
\end{align}
	where
		\begin{align}
		\E^{\tau_0}(\tau_0,h^{1:t}):=\int_{\tau_0}^\Ti \mE_{X_\tau^{t+1}\sim p^{\tau_0}_{t+1}(\tau,\cdot|h^{1:t})}\big[\big|s^{t+1}_\theta(\tau,h^{1:t},X_\tau^{t+1})-\nabla_x \log p^{\tau_0}_{t+1}(\tau,X_\tau^{t+1}|h^{1:t})\big|^2\big] \md \tau.
	    \end{align}
    Here, $c'>0$ is independent of $\tau_0$, $\Ti$, $\epsilon_{\rm score}$ and $T$.
\item[(2)] Define $\mQ^{\Ti-\tau_0}$ to be the joint distribution with regular conditional kernels $\{\mQ^{\Ti-\tau_0}_{y^{1:t}}:y^{1:t}\in \mR^{dt}, t=1,2,\cdots, T-1\}$, i.e.,
\begin{align}
\mQ^{\Ti-\tau_0}(\md y^{1:T})
:=\mQ^{\Ti-\tau_0}_1(\md y^1)
\prod_{t=1}^{T-1}\mQ^{\Ti-\tau_0}_{y^{1:t}}(\md y^{t+1}).\label{Q:def}
\end{align}
	Moreover, define
	\begin{align}
	    \mathfrak b_{\tep}
	    =C_{\tau_0}\big(D_{\tep}e^{-\Ti}+e^{-c'D_{\tep}}+e^{-c'\sqrt{D_{\tep}}}\big).
	\end{align}
		Let $a_T:=2^{-(T-1)}$. Assume that
		\begin{align*}
		    &\Ti>1,\qquad C_{\tau_0}>1,\qquad
		    D_{\tep}>1,\qquad \mathfrak b_{\tep}<1,\\
		    &C_{\tau_0}D_{\tep}\Ti^{1/2}\epsilon_{\rm score}<1.
		\end{align*}
		Then
	\begin{align}
	\cA\cW_2^2(\mP,\mQ^{\Ti-\tau_0})\leq &(K')^T\tau_0
	+C_{\tau_0}T^2D_{\tep}^2\Ti\big(D_{\tep}e^{-\Ti}\big)^{a_T}\notag\\
	&+C_{\tau_0}T^2D_{\tep}^2\Ti e^{-c'a_T\sqrt{D_{\tep}}}\notag\\
	&+C_{\tau_0}T^2D_{\tep}^2\Ti^{1+a_T/2}\big(D_{\tep}\epsilon_{\rm score}\big)^{a_T}.
	\label{eq:AWbounds}
	\end{align}
    Here $K'>1$ and $c'>0$ are independent of $\tau_0$, $\Ti$, $\epsilon_{\rm score}$ and $T$, while $C_{\tau_0}$ is independent of $\Ti$, $\epsilon_{\rm score}$ and $T$.
	\end{itemize}
	\end{theorem}
	\begin{proof}
	See Appendix \ref{app:proof-AWbounds}.
	\end{proof}

Several remarks on Theorem \ref{thm:AWbounds} follow.

\begin{remark}\phantomsection\label{remark:mainthm}
	\begin{itemize}
	\item[(1)] (Convergence guarantee) { 
	    In this remark, we first fix $T$ and focus on the trade-off among $\tau_0$, $\Ti$ and $\epsilon_{\rm score}$. The scale of $T$ will be discussed in (4). Let $\eta>0$ be arbitrary. We first choose $\tau_0>0$ such that $(K')^T\tau_0<\eta/4$. Set $a_T=2^{-(T-1)}$ and $L_\epsilon:=\log(\epsilon_{\rm score}^{-1})$. By Proposition \ref{prop:assumption:stheta}, with this $\tau_0$ and $T$ fixed, we may take $D_{\tep}\sim_{\tau_0,T}L_\epsilon^2$. Then, up to constants depending on $\tau_0$ and $T$, the remaining three terms in \eqref{eq:AWbounds} have the rough orders $L_\epsilon^{4+2a_T}\Ti e^{-a_T\Ti}$, $L_\epsilon^{4}\Ti\epsilon_{\rm score}^{c_{\tau_0,T}}$, and $L_\epsilon^{4+2a_T}\Ti^{1+a_T/2}\epsilon_{\rm score}^{a_T}$, respectively. Here $c_{\tau_0,T}>0$ is a constant depending on $\tau_0$ and $T$. Therefore, choosing $\Ti=\epsilon_{\rm score}^{-r}$ with
	    \[
	        0<r<\min\left\{c_{\tau_0,T},\frac{a_T}{1+a_T/2}\right\}
	    \]
		    makes each of the three terms smaller than $\eta/4$ for sufficiently small $\epsilon_{\rm score}$. This scale also gives $D_{\tep}\Ti^{1/2}\epsilon_{\rm score}\to0$, so the technical smallness condition in Theorem \ref{thm:AWbounds} is satisfied. With such a choice of parameters, we achieve $\cA\cW_2^2(\mP,\mQ^{\Ti-\tau_0})\to0$. This provides a {\it theoretical convergence guarantee} as $\Ti\to \infty$, $\epsilon_{\rm score}\to 0$ and $\tau_0\to0$}. Even for static data, establishing Wasserstein convergence for diffusion models is known to be challenging; see Remark 1 of \citet{Kwon2022}. Under various technical assumptions on the data distribution, such as strong log-concavity, Wasserstein convergence results are available \citep[e.g.,][]{gao2025wasserstein,Tang2024}, and it is possible to bound $\cA\cW_2$ by $\cW_2$ up to some explicit error terms; see \citet{blanchet2024}. However, with this conventional approach, obtaining bounds for conditional distributions requires imposing strong log-concavity on {\it every} conditional distribution of $\mP$, uniformly in the observations $x^{1:t}$, which seems to be rather restrictive for real-world modeling. We circumvent these assumptions and replace them with certain dissipativity conditions on the {\it networks}, which can be enforced by carefully designing the network structure (i.e., Assumption \ref{assumption:stheta}).

	    \item[(2)] (Wasserstein bounds for static data) { To consider the static data case, take $T=1$ in \eqref{eq:AWbounds}. Using the relation $D_{\tep}\sim \log^2\epsilon_{\rm score}^{-1}$, we obtain the following Wasserstein bound for static data generation:
    \begin{align}
	    \cW_2^2(\mP,\mQ^{\Ti-\tau_0})\lesssim &K'\tau_0+\log^6(\epsilon_{\rm score}^{-1})\Ti e^{-\Ti}
    +\Ti\epsilon_{\rm score}^{c'}\log^4(\epsilon_{\rm score}^{-1})
    \\
    &+\Ti^{3/2}\epsilon_{\rm score}\log^6(\epsilon_{\rm score}^{-1})\\
    =:& {\tt EarlyStop\_Err}+{\tt Noise\_Err}+{\tt Cutoff\_Err}+{\tt Score\_Err},\label{bound:static}
    \end{align}
	  with {\tt EarlyStop\_Err} denoting the early-stopping error, {\tt Noise\_Err} the noisy approximation error, {\tt Score\_Err} the score-matching error, and {\tt Cutoff\_Err} the newly introduced cut-off error (see Appendix \ref{app:proof-AWbounds} for details). We emphasize that when $T=1$, Assumption \ref{assumption:P}-1 and Assumption \ref{assumption:P}-2 reduce to the usual Lipschitz continuity and growth conditions for score functions. In particular, we do not rely on any structural assumptions on $\mP$, such as strong log-concavity. Our result is achieved by bounding the Wasserstein metric by total variation, up to a tail error, using Proposition 20 of \citet{blanchet2024} and Theorem 6.15 of \citet{villani2008}, and then applying uniform tail estimates for SDEs with dissipative coefficients. The price to pay is that the noise approximation error term $\log^6(\epsilon_{\rm score}^{-1})\Ti e^{-\Ti}$ diverges when $\epsilon_{\rm score}\to 0$ with a {\it fixed} $\Ti$.}

	%\item[(3)] It is evident that, as outlined in Algorithm \ref{adaptive-sampling}, $\{y_{(n)}^1\}_{n=1}^N$ is sampled from the distribution $\mQ_1^\Ti$. Furthermore, given the sequence $y_{(n)}^{1:t}$, the next sample $y_{(n)}^{t+1}$ is drawn from the distribution $\mQ^\Ti_{y_{(n)}^{1:t}}$. By definition, $\{y_{(n)}^1,y_{(n)}^2,\cdots, y_{(n)}^T\}_{n=1}^N$ is one sample from the distribution $\mQ^\Ti$, thereby establishing that $\mQ^\Ti$ is the output distribution of Algorithm \ref{adaptive-sampling}. Consequently, Theorem \ref{thm:AWbounds} provides $\cA\cW_2$ bounds between the output distribution of Algorithm \ref{adaptive-sampling} and the data distribution.

	\item[(3)] (Conditional sampling) By the construction of $\mQ^{\Ti-\tau_0}$ in \eqref{Q:def}, for a random variable $Y^{1:T}\sim \mQ^{\Ti-\tau_0}$, $\mQ^{\Ti-\tau_0}_{y^{1:t}}$ is the conditional distribution of $Y^{t+1}$ given $Y^{1:t}=y^{1:t}$. Consequently, it can be sampled by selecting the score function $x\mapsto s^{t+1}_\theta(\tau,y^{1:t},x)$. Thus, Algorithm \ref{adaptive-sampling} can also be used for conditional sampling, provided that $y^{1:t}$ is included as input and the sampling process starts with $Z_{(n)}^{t+1}$.

	    \item[(4)]  (The exponential dependence on $T$) We remark that \eqref{eq:AWbounds} contains two types of horizon dependence. The first is the separate smoothing term $(K')^T\tau_0$, which is induced by the clean-kernel Lipschitz property in Assumption \ref{assumption:P}-4 and a linear smoothing iteration. The rest of \eqref{eq:AWbounds} has no hidden exponential constant in $T$. Proposition \ref{prop:assumption:stheta} gives
	    \[
	    D_{\tep}
	    =O_{\tau_0}\!\left(
	    T\big(\log(eT)+\log(e+h_2(\tau_0)^{-2}\epsilon_{\rm score}^{-2})\big)^2
	    \right).
	    \]
	    Hence, these prefactors are polynomial in $T$ and polylogarithmic in $\epsilon_{\rm score}^{-1}$. The severe dependence on $T$ in the remaining small quantities enters through the exponent $a_T=2^{-(T-1)}$. This exponent is induced by the proof strategy, which first controls conditional reverse SDEs in total variation and then converts total variation to $\cW_2$. More precisely, the conditional score-matching error is controlled only in mean along the noisy data histories $H^{1:t}\sim \mP^{\tau_0}_{1:t}$, whereas the adaptive sampler evaluates the learned conditional scores along generated histories $Y^{1:t}$. Without a direct stability estimate for the learned conditional kernels $y^{1:t}\mapsto \mQ_{y^{1:t}}^{\Ti-\tau_0}$, the proof estimates $\cW_2(\mP_{h^{1:t}}^{\tau_0},\mQ_{y^{1:t}}^{\Ti-\tau_0})$ through a total-variation bound. Since $\cW_2^2$ is controlled by $\TV$ only after a square-root loss, the adapted Wasserstein iteration contains a term of the form
    \begin{align}
        A_{t+1}\leq A_t+C\sqrt{A_t}+C\epsilon_{\rm score},
        \qquad
        A_t:=\mE |X^{1:t}-Y^{1:t}|^2,
    \end{align}
    inducing the cumulative square-root exponent. Nevertheless, it should be interpreted as a worst-case consequence of the present general assumptions rather than an intrinsic rate for every dynamic model. Under additional (though quite restrictive) contraction conditions on the data and learned kernels, the smoothing and approximation recursions become contractive; see Assumption \ref{ass:contraction} and Proposition \ref{prop:exp-in-T}.
\end{itemize}
\end{remark}

{ 
It is useful to relate Theorem \ref{thm:AWbounds} to the conventional conditional diffusion model, which is widely used, for instance, in text-to-image generation. In the standard conditional setting, one observes a condition $c$ and generates only from the conditional law of the target given this condition; the condition variable itself is not generated. Thus, at the level of distributions, the goal is to learn a static conditional kernel rather than a full dynamic law. This corresponds to the special case of our framework in which the history is externally fixed and only the next coordinate is sampled. Therefore, the error bound for conditional diffusion models is expressed as the average Wasserstein distance between the conditional kernels, rather than the adapted Wasserstein distance between joint laws. The following corollary records this type of bound and applies more generally at every stage of a {\it nonstationary} environment. This result is useful for {\it static} or {\it myopic} portfolio selection, in which we make a single decision at each stage based on the observed history. Nevertheless, we are mainly interested in the {\it dynamic} setting, in which decisions are planned several steps ahead.
}

\begin{corollary}\label{coro:conditional-error}
For $t=1,2,\cdots, T-1$, we have
\begin{equation}
\mE_{H^{1:t}\sim \mP^{\tau_0}_{1:t}}[\cW_2^2(\mP^{\tau_0}_{H^{1:t}},\mQ^{\Ti-\tau_0}_{H^{1:t}})]\leq C_{\tau_0}\big(D_{\tep}(\Ti^{1/2}\epsilon_{\rm score}+e^{-\Ti})+e^{-c'\sqrt{D_{\tep}}}\big).
\end{equation}

\end{corollary}

\begin{proof}
    Taking $h^{1:t}=y^{1:t}$ in \eqref{conditional-bound}, integrating with respect to $\mP^{\tau_0}_{1:t}$, and using \eqref{score-error} gives the desired result.
\end{proof}

{  To conclude this section, we present an improved bound under additional structural assumptions.

\begin{assumption}\label{ass:contraction}
   Suppose that the clean conditional kernels of $\mP$ are Markovian and contractive: for $t=1,2,\cdots,T-1$,
    $\mP_{x^{1:t}}=\mP_{x^t}$ and there exists $\kappa_P\in(0,1)$ such that, for all $x,y\in\mR^d$,
    \begin{align}
    \cW_2(\mP_x,\mP_y)\leq \kappa_P |x-y|.
    \end{align}
   Suppose also that the conditional kernels $\mQ^{\Ti-\tau_0}$ are Markovian, i.e.,
    $\mQ^{\Ti-\tau_0}_{y^{1:t}}=\mQ^{\Ti-\tau_0}_{y^t}$ for $t=1,2,\cdots,T-1$, and there exists $\kappa_Q\in(0,1)$ such that, for any $t=1,2,\cdots,T-1$ and $x,y\in \mR^d$,
    \begin{align}
    \cW_2(\mQ^{\Ti-\tau_0}_{x},\mQ^{\Ti-\tau_0}_{y})\leq \kappa_Q |x-y|.
    \end{align}
\end{assumption}

\begin{proposition}\label{prop:exp-in-T}
    Suppose that Assumption \ref{ass:contraction} holds in addition to Assumptions \ref{assumption:score-matching-error}-\ref{assumption:stheta}. Define the full-history conditional error
    \begin{align}
        \Delta_{\tep}:=&\ \cW_2^2(\mP^{\tau_0}_1,\mQ_1^{\Ti-\tau_0}) \notag\\
        &+\max_{1\leq t\leq T-1}\mE_{H^{1:t}\sim \mP^{\tau_0}_{1:t}}\Big[\cW_2^2\big(\mP^{\tau_0}_{H^{1:t}},\mQ_{H^{1:t}}^{\Ti-\tau_0}\big)\Big],
    \end{align}
    where $\mP^{\tau_0}_{1:t}$ is the marginal distribution of $H^{1:t}$ under $\mP^{\tau_0}$. Then there exist constants $C_{\kappa_P},C_{\kappa_Q}>0$, independent of $\tau_0,\Ti,\epsilon_{\rm score}$ and $T$, such that
    \begin{align}
        \cA\cW_2^2(\mP,\mQ^{\Ti-\tau_0})\leq C_{\kappa_P}T\tau_0+C_{\kappa_Q}T\Delta_{\tep}.
    \end{align}
In particular, combining the static estimate \eqref{bound:static} with the conditional estimate in Corollary \ref{coro:conditional-error}, both the iterated exponent $1/2^{T-1}$ and the exponential smoothing factor $(K')^T$ in \eqref{eq:AWbounds} can be removed under the additional contraction conditions above.
\end{proposition}

\begin{proof}
    See Appendix \ref{app:proof:sec:2}.
\end{proof}}

\section{\texorpdfstring{ Stability of dynamic mean-variance portfolio selection}{Stability of dynamic mean-variance portfolio selection}}\label{stability}

This section connects diffusion models to {\it downstream decision problems}. We first formulate the dynamic mean-variance portfolio selection problem, allowing a general constraint set $\mathcal K$. We employ a dual argument and then establish the dynamic programming principle for the dual quadratic-hedging problem in Subsection \ref{subsec:DPP}. In Subsection \ref{sec:stability}, we provide an optimal-value stability result and a performance-gap estimate between the data-generating model $\mP$ and an alternative model $\mQ$, using the $\cA\cW_2$ metric.

{ Throughout this section, the stock-price process is exogenous, i.e., an admissible portfolio strategy affects the investor's wealth and gain processes, while leaving the conditional law of future prices unchanged. Thus, the model represents a price-taking investor without market impact.}

Suppose the stock prices $(S^1,S^2,\cdots,S^T)$ follow some model $\mP$ and the risk-free interest rate is zero. For an investment strategy $\vartheta$ (in some admissible set or constraint set), the associated gain process is $(\vartheta \cdot S)_t := \sum_{l=1}^{t-1} \vartheta_l ^\mt(S^{l+1}-S^l)$. We want to maximize the { time-0-optimal (or precommitted)} mean-variance objective:
\begin{align}
&v(\mP;\vartheta):= \mE_\mP[(\vartheta\cdot S)_T]-\frac{\gamma}{2} \Var_{\mP}[(\vartheta\cdot S)_T].\\
&v^*(\mP):= \sup_{\vartheta}v(\mP;\vartheta).\label{MV-problem}
\end{align}
By a classical result, if we know the data-generating model $\mP$, it is sufficient to solve the {\bf quadratic-hedging problem}
\begin{equation}\label{auxiliary_MV}
V(\mP,c)=\min_{\vartheta}\mE_{\mP}\big[|c-(\vartheta\cdot S)_T|^2 \big],
\end{equation}
for some $c>0$. This is because, for $\vartheta$ such that $\mE_\mP[(\vartheta\cdot S)_T]>0$,
\begin{equation}\label{MV-hedging-relation}
	v(\mP;\vartheta)=\sup_{a\in \mR_+}\bigg\{-\frac{\gamma}{2}\mE_\mP\Big[\Big|(\vartheta\cdot S)_T-\Big(\frac{1}{\gamma}+a\Big)\Big|^2\Big]+\frac{1}{2\gamma}+a\bigg\}.
\end{equation}
Consequently, we have the duality relation between these two value functions:
\begin{equation}\label{MV-hedging-duality}
    v^*(\mP)=\sup_{a\in \mR_+} \bigg\{ -\frac{\gamma}{2}V\bigg(\mP,\frac{1}{\gamma}+a\bigg)+\frac{1}{2\gamma}+a\bigg\}.
\end{equation}

\begin{remark}
If we consider a {\it cone} constraint on the portfolio, then it is sufficient to study \eqref{auxiliary_MV} with $c=1$, and there is an explicit relation between solutions of these two problems. For generality, we choose to consider a generic convex constraint set $\mathcal K$. In particular, we can take $\mathcal K$ to be bounded.
\end{remark}
%Then the optimal MV portfolio is given by
%\begin{equation}
%\vartheta^*_t = \frac{1}{\gamma}\frac{\hat \vartheta_t}{1-\mE_\mP[(\hat \vartheta\cdot S)_T]},
%\end{equation}
%where $\hat\vartheta$ solves \eqref{auxiliary_MV}.

\subsection{Technical preparations and existence of an optimizer}\label{technical}
 This subsection summarizes the technical assumptions and results needed for the quadratic-hedging problem \eqref{auxiliary_MV}. There is a large literature on this type of problem; see, e.g., \citet{LH2007,CS2012,CCK2024}. Here, we consider a general bounded convex constraint set, which makes some classical results inapplicable. For this reason, we restate certain well-known results (such as the DPP) and provide short proofs for completeness.

 Denote $\triangle S^t = S^{t+1}-S^t$. Let $\mF=(\F_t)_{t=1}^T$ be the natural filtration of $S$, with $\F_1$ trivial.

\begin{assumption}\label{assumption-S} 
 $\mP(S^1=s^1)=1$ for some $s^1\in(0,\infty)^d$. Moreover, for some $M_S>0$,
\begin{align}
    M_S:=\mE_{\mP}[|S^{1:T}|^{2T}]<\infty.
\end{align}

\end{assumption}
\begin{assumption}\label{assumption-ND}
There exists $\delta\in (0,1)$ such that
\begin{equation}\label{ND}
\bigg(\mE_\mP\big[\triangle S^t|\F_t \big]\bigg)\bigg(\mE_\mP\big[\triangle S^t|\F_t \big]\bigg)^\mt\preceq (1-\delta)\mE_\mP\big[\triangle S^t(\triangle S^t)^\mt|\F_t \big], \quad t=1,2,\cdots, T-1,
\end{equation}
Equivalently, for any $\phi\in \mR^d$,
\begin{equation}\label{ND-equivalent}
	\Big|\big(\mE_\mP\big[\phi^\mt\triangle S^t|\F_t \big]\big)\Big|^2\leq (1-\delta)\mE_\mP\big[|\phi^\mt\triangle S^t|^2 |\F_t \big]
\end{equation}

\end{assumption}

\begin{definition}
	$L^2_{\mP}(S)$ is the set of $\mF$-adapted, $\mR^d$-valued processes $\{\vartheta_t\}_{t=1}^{T-1}$ such that
	\begin{equation}
		\mE_\mP\big[|\vartheta^\mt _t\triangle S^t|^2\big]<\infty,\forall t=1,2,\cdots, T-1.
	\end{equation}
\end{definition}

\begin{remark} 
\begin{enumerate}
\item Assumption \ref{assumption-S} is a normalization and integrability condition. The deterministic initial price $S^1=s^1$ simply fixes the observed initial state from which the finite-horizon portfolio problem starts. The $2T$-moment condition ensures that, for bounded portfolio constraints, the terminal gain and the mean-variance criterion are well-defined; in particular, constant strategies are in $L^2_{\mP}(S)$.

\item Assumption \ref{assumption-ND} is a technical nondegeneracy condition. It requires the conditional drift of each one-period price increment to be dominated by its conditional second moment, uniformly over time. Its role is the same as the structure-type conditions used in the discrete-time quadratic-hedging and mean-variance literature: following \citet{variance-hedging} (with the standard multidimensional extension), \eqref{ND} implies that the set of unconstrained attainable terminal gains is a closed subspace of $L^2(\mP)$; Proposition \ref{prop:closedness} establishes the corresponding result under bounded convex constraints. This closedness yields the existence of an optimizer for the constrained quadratic-hedging problem \eqref{auxiliary_MV}. Thus \eqref{ND} should be read as a sufficient technical condition for the well-posedness of the auxiliary quadratic-hedging problem.
\end{enumerate}
\end{remark}

We consider {\bf portfolio constraints}.
\begin{definition}
For a closed convex subset $\mathcal K\subset \mR^d$, define $\Theta_{\mathcal K,\mP} = \{\vartheta\in L^2_{\mP}(S): \vartheta_t\in \mathcal K,\ \mP{\rm -a.s.}  \}$, and $G(\Theta_{\mathcal K,\mP}):=\{(\vartheta\cdot S)_T:\vartheta \in \Theta_{\mathcal K,\mP} \}$. In this paper, we assume that $\mathcal K$ is bounded and write $K:=\sup_{\theta\in\mathcal K}|\theta|<\infty$.
\end{definition}

\begin{lemma}\label{predictable-projection}
		Denote $\Pi^S_t:=\sqrt{\mE_\mP[\triangle S^t (\triangle S^t)^\mt |\F_t]}, t=1,2,\cdots, T-1$. Then for any $\vartheta, \bar\vartheta \in L^2_{\mP}(S)$, $(\vartheta \cdot S)_t=(\bar \vartheta\cdot S)_t$, $\mP$-a.s., for every $t=1,2,\cdots, T$ if and only if $\Pi^S_t\vartheta_t = \Pi^S_t \bar\vartheta_t$, $\mP$-a.s., for every $t=1,2,\cdots, T-1$.
\end{lemma}

\begin{proof}
    See Appendix \ref{app:proof:stability}.
\end{proof}

\begin{proposition}\label{prop:closedness}
$G(\Theta_{\mathcal K,\mP})$ is a closed convex subset of $L^2(\mP)$.
\end{proposition}
\begin{proof}
    See Appendix \ref{app:proof:stability}.
\end{proof}

\begin{corollary}
	There exists an optimizer in $\Theta_{\mathcal K,\mP}$ for the quadratic hedging problem \eqref{auxiliary_MV}, and the optimizer is unique in value, i.e., $(\vartheta\cdot S)_T=(\tilde\vartheta\cdot S)_T$, $\mP$-a.s., for any two optimizers $\vartheta$ and $\tilde \vartheta$.
\end{corollary}

\subsection{Dynamic programming}\label{subsec:DPP}
To derive the DPP, we define a conditional version of \eqref{auxiliary_MV} (for $t=1,2,\cdots,T$):
\begin{equation}\label{dynamic_auxiliary}
V(t,w,s^{1:t}):= \inf_{\vartheta\in \Theta^{t,s^{1:t}}_{\mathcal K,\mP}}	\mE_{\mP_{s^{1:t}}}\bigg[\bigg| c-w-\sum_{l=t}^{T-1}\vartheta_{l}^\mt\triangle S^l   \bigg|^2     \bigg].
\end{equation}
\begin{definition}
$\Theta^{t,s^{1:t}}_{\mathcal K,\mP}$ is the set of processes $\vartheta_l$, $l=t,\cdots, T-1$, that take values in $\mathcal K$, $\mP_{s^{1:t}}$-a.s., and are adapted to $\{\F^t_l\}_{l=t}^{T-1}$, with $\F^t_l=\sigma(S^{t+1},\cdots, S^l)$\footnote{By convention, we assume $\F^t_t$ is trivial.}
\end{definition}
Then, assuming $(\vartheta\cdot S)_1=0$ and $S^1=s^1$ is a constant, we have
\begin{equation}
	\inf_{\vartheta\in \Theta_{\mathcal K,\mP}}\mE_{\mP}\big[|c-(\vartheta\cdot S)_T|^2 \big] = V(1,0,s^1).
\end{equation}
Note that \eqref{ND} is written in terms of conditional expectations, where $\F_t=\F^0_t=\sigma(S^1,\cdots,S^t)$. It follows from Subsection \ref{technical} that \eqref{dynamic_auxiliary} also has an optimizer that is unique in value.

\begin{theorem}\label{DPP:thm}
The terminal condition holds for every $(w,s^{1:T})\in\mR\times\mR^{dT}$:
\begin{equation}
V(T,w,s^{1:T}) = |c-w|^2.\label{DPP-1}
\end{equation}
For each $t=1,\ldots,T-1$, the conditional value functions may be chosen measurable so that, for every $w\in\mR$ and $\mP_{1:t}$-almost every $s^{1:t}$, with an exceptional set independent of $w$,
\begin{equation}
V(t,w,s^{1:t}) = \inf_{\theta\in \mathcal K} \mE_{\mP_{s^{1:t}}}\Big[V\Big(t+1,w+\theta^\mt\triangle S^t,(s^{1:t},S^{t+1})\Big) \Big].	\label{DPP-2}
\end{equation}
Moreover, the infimum in \eqref{DPP-2} admits a measurable selector $\hat\theta$ satisfying, for every $w\in\mR$ and the same histories,
\begin{equation}\label{hattheta}\hat\theta(t,w,s^{1:t}) \in \argmin_{\theta\in \mathcal K}\mE_{\mP_{s^{1:t}}}\Big[V\Big(t+1,w+\theta^\mt\triangle S^t,(s^{1:t},S^{t+1})\Big) \Big].
\end{equation}
For the initial problem $V(1,0,s^1)$, an optimizer is given recursively by $\hat \vartheta _t = \hat \theta(t,(\hat\vartheta\cdot S)_t,S^{1:t})$.
\end{theorem}

\begin{proof}
    See Appendix \ref{app:proof:stability}.
\end{proof}

\subsection{Stability}\label{sec:stability}

For $\mu\in\{\mP,\mQ\}$, denote by $V(\cdot;c,\mu)$ the function obtained from DPP \eqref{DPP-1}-\eqref{DPP-2} under the model $\mu$. The next theorem bounds the difference between the two value functions by $\cA\cW_2(\mP,\mQ)$. For $t=1,\ldots,T-1$, write the local future second moments as
\[
M_2^\mP(s^{1:t}):=\Big(\mE_{\mP_{s^{1:t}}}[|S^{t+1:T}|^2]\Big)^{1/2},\qquad
M_2^{\mQ}(\tilde s^{1:t}):=\Big(\mE_{\mQ_{\tilde s^{1:t}}}[|\tilde S^{t+1:T}|^2]\Big)^{1/2}.
\]
We use the convention $M_2^\mP(s^{1:T})=M_2^{\mQ}(\tilde s^{1:T})=0$. For $\pi\in \Pi_{\rm bc}(\mP,\mQ)$, set
\[
D^\pi_t(s^{1:t},\tilde s^{1:t}):=
\Big(\mE_{\pi_{s^{1:t},\tilde s^{1:t}}}[|S^{t+1:T}-\tilde S^{t+1:T}|^2]\Big)^{1/2}.
\]

\begin{theorem}\label{DPP-stability}
		For any $\pi\in \Pi_{\rm bc}(\mP,\mQ)$, any $t\in\{1,2,\ldots,T-1\}$, and any $w,\tilde w\in\mR$, the following bound holds for $\pi_{1:t}$-almost every $(s^{1:t},\tilde s^{1:t})$:
	\begin{align}
	|V(t,w,s^{1:t};c,\mP)-& V(t,\tilde w,\tilde s^{1:t};c,\mQ)|\leq C\Big\{\big(|w-\tilde w|+|s^t-\tilde s^t|\big)B_t\\
	&+\big(B_t+|w-\tilde w|+|s^t-\tilde s^t|\big)D^\pi_t(s^{1:t},\tilde s^{1:t})+\big(D^\pi_t(s^{1:t},\tilde s^{1:t})\big)^2\Big\},
	\label{DPP-stability-eq}
	\end{align}
where
\[
B_t:=1+c+|w|+|\tilde w|+|s^t|+|\tilde s^t|+M_2^\mP(s^{1:t})+M_2^{\mQ}(\tilde s^{1:t}),
\]
and $C$ is a constant depending only on $T$ and $K$. In particular, when $t=1$, we have
\begin{equation}\label{stability-initial}
|V(1,0,s^1;c,\mP)-V(1,0,s^1;c,\mQ)|\leq C_0\big\{(1+|s^1|+c)\cA\cW_2(\mP,\mQ)+\cA\cW_2(\mP,\mQ)^2\big\},
\end{equation}
where $C_0$ depends only on $T$, $K$ and $M_S$.
\end{theorem}

\begin{proof}
    See Appendix \ref{app:proof:stability}.
\end{proof}

Recall the MV value function $v^*(\mP)$ defined in \eqref{MV-problem}. We now establish the stability of $v^*$ in terms of the adapted Wasserstein metric.

\begin{corollary}\label{coro:MVstability}
Assume $\cA\cW_2(\mP,\mQ)<1$. Then
\begin{equation}\label{stability-vstar}
	|v^*(\mP)-v^*(\mQ)|\leq C\cA\cW_2(\mP,\mQ),
\end{equation}
where $C$ depends only on $\gamma$, $T$, $K$ and $M_S$.
\end{corollary}

\begin{proof}
    See Appendix \ref{app:proof:stability}.
\end{proof}

\begin{remark}
The dynamic programming method also applies to the {\it fully invested problem}. In this case, admissibility requires $(\theta\cdot S)_t=\theta_t^\mt S^t$, and the dynamic program minimizes over $\theta\in \mathcal K(w,s^t):=\{\theta:\theta^\mt s^t=w\}$. Thus, at each period, we reinvest all of our wealth in the risky assets $S$.
\end{remark}

{ 
Based on the optimal-value stability result in Corollary \ref{coro:MVstability}, we now discuss performance-gap bounds for sufficiently regular feedback policies obtained from the approximate model $\mQ$. To this end, we need the following fixed-policy stability result. In fact, we obtain stability under the ordinary Wasserstein metric because the nonanticipative feedback policy itself ensures the admissibility of the induced trading strategies under both $\mP$ and $\mQ$, without requiring a bicausal coupling. However, we emphasize that our results in Section \ref{sec:mainresults} are crucial for establishing the desired performance-gap bound because the optimal-value gap is controlled by $\cA\cW_2$.

Let $\phi=\{\phi_t\}_{t=1}^{T-1}$ be a feedback policy with
\[
\phi_t:\mR\times \mR^{dt}\to \mathcal K.
\]
For each deterministic path $s^{1:T}$, we define the induced wealth path iteratively by $(\phi\cdot s)_1=0$ and
\begin{equation}
(\phi\cdot s)_{t+1}=(\phi\cdot s)_t+\phi_t((\phi\cdot s)_t,s^{1:t})^\mt(s^{t+1}-s^t),\quad t=1,\ldots,T-1.
\end{equation}
We say that $\phi$ is an $L$-Lipschitz feedback policy for some $L>0$ if for every $t=1,\ldots,T-1$,
\begin{equation}\label{regular-feedback-primitive}
|\phi_t(w,s^{1:t})-\phi_t(\tilde w,\tilde s^{1:t})|
\leq L\big(|w-\tilde w|+|s^{1:t}-\tilde s^{1:t}|\big).
\end{equation}
As a deterministic function of the path $s^{1:T}$, we can evaluate the performance of $\phi$ under any model $\mu$ by
\begin{equation}
v(\mu;\phi):=\mE_\mu[(\phi\cdot S)_T]-\frac{\gamma}{2}\Var_\mu((\phi\cdot S)_T).
\end{equation}

\begin{lemma}\label{lemma:fixed-policy-stability}
Assume $\mathcal K$ is bounded and let $K:=\sup_{\theta\in \mathcal K}|\theta|<\infty$. Suppose that $\phi$ is a Lipschitz feedback policy satisfying \eqref{regular-feedback-primitive}. Then, for any $\mQ\in \cP(\mR^{dT})$ such that
\begin{equation}\label{price-moment-fixed-policy}
\mE_\mQ[|S^{1:T}|^{2T}]\leq M_S'
\end{equation}
for some constant $M_S'>0$,
we have
\begin{equation}
|v(\mP;\phi)-v(\mQ;\phi)|\leq C\cW_2(\mP,\mQ),
\end{equation}
where $C$ only depends on $\gamma$, $T$, $K$, $L$, $M_S$ and $M_S'$.
\end{lemma}
\begin{proof}
    See Appendix \ref{app:proof:stability}.
\end{proof}
\begin{corollary}\label{coro:regular-policy-transfer}
Assume $\cA\cW_2(\mP,\mQ)<1$. Let $\phi^\epsilon$ be an $L_{\phi^\epsilon}$-Lipschitz model-independent feedback policy satisfying \eqref{regular-feedback-primitive}, and suppose that the $2T$-moment condition \eqref{price-moment-fixed-policy} holds for this pair $(\mP,\mQ)$. If
\begin{equation}
v(\mQ;\phi^\epsilon)\geq v^*(\mQ)-\epsilon,
\end{equation}
then
\begin{equation}
v(\mP;\phi^\epsilon)\geq v^*(\mP)-\epsilon-C\cA\cW_2(\mP,\mQ),
\end{equation}
where $C$ depends only on $\gamma$, $T$, $K$, $L_{\phi^\epsilon}$, $M_S$ and $M_S'$.
\end{corollary}

\begin{proof}
    See Appendix \ref{app:proof:stability}.
\end{proof}

\begin{remark}
One can avoid imposing the $2T$-moment condition on $\mQ$ separately by additionally showing that the generative model $\mQ$ satisfies \eqref{price-moment-fixed-policy}, for example, by using dissipativity or tail estimates for the sampler.
\end{remark}

\begin{remark}
Corollary \ref{coro:regular-policy-transfer} connects the generative model to downstream decision-making problems via $\cA\cW_2$. It states that if we can solve the problem under $\mQ$ and obtain a sufficiently regular feedback policy that is $\epsilon$-suboptimal (for example, by obtaining $\phi^\epsilon$ from a reinforcement-learning algorithm under the alternative model $\mQ$), then this policy will also perform well under the real model $\mP$. The regularity condition is crucial for this transfer result, and whether it can be relaxed is an interesting question for future research.
\end{remark}

}

{ 
If the real data law $\mP$ and the generated law $\mQ^{\Ti-\tau_0}$ in Section \ref{sec:mainresults} are interpreted as laws of log prices instead of prices, the following elementary transfer result converts the log-price error into a price-law error. As a result, we still have a performance-gap bound similar to that in Corollary \ref{coro:regular-policy-transfer} for the MV portfolio selection problem.

Let $\Phi:\mR^{dT}\to(0,\infty)^{dT}$ be the componentwise exponential map, $\Phi(x^{1:T})=(e^{x^1},\ldots,e^{x^T})$.

\begin{lemma}\label{lemma:log-price-transfer}
Let $\mu$ and $\nu$ be two laws on $\mR^{dT}$. Suppose that, for some $a>2$ and $K_a<\infty$,
\[
\sup_{t=1,\ldots,T}\left\{\mE_\mu[e^{a|X^t|}]+\mE_\nu[e^{a|Y^t|}]\right\}\leq K_a.
\]
Then, for every $R>0$,
\[
\cA\cW_2^2(\Phi_\#\mu,\Phi_\#\nu)
\leq e^{2R}\cA\cW_2^2(\mu,\nu)+C_{d,T}K_a e^{-(a-2)R}.
\]
Consequently, if $\cA\cW_2(\mu,\nu)\leq 1$, then
\[
\cA\cW_2(\Phi_\#\mu,\Phi_\#\nu)
\leq C_{a,d,T,K_a}\cA\cW_2(\mu,\nu)^{(a-2)/a}.
\]
If $a>2T$, then $\Phi_\#\mu$ and $\Phi_\#\nu$ also have finite $2T$-moments.
\end{lemma}

\begin{proof}
    See Appendix \ref{app:proof:stability}.
\end{proof}
}

\section{Applications}\label{sec:app}

In this section, we illustrate the application of the diffusion-based adaptive sampling method proposed in Section \ref{sec:sampling} to MV portfolio selection problems. Specifically, we draw many paths from the sampler and use them to approximate price scenarios under the data-generating model. {The results in Section \ref{sec:mainresults} show that $\mP$ and $\mQ$ are close in $\cA\cW_2$ for a sufficiently long diffusion time and a sufficiently small score-matching error.} Meanwhile, Corollary \ref{coro:MVstability} shows that the value functions under these two models are also close. This motivates a deep reinforcement-learning approach that treats $\mQ$ as an approximate environment, which we call a {\it generative environment}. Because we can generate a large number of samples with statistical properties similar to those of the limited data ($\mP\approx \mQ$), our approach can potentially address the scarcity and time-sensitivity of financial data.

\subsection{Implementation of Algorithm \ref{adaptive-sampling}}

{ Real financial time series typically exhibit path dependence, and the conditional sampler should retain information from the observed or generated price history $s^{1:t}$ when generating the next coordinate. In the formal analysis, this history enters through the full noisy history $H^{1:t}=h^{1:t}$, and, in the implementation, a natural first choice is to use the clean history $s^{1:t}$ as its proxy. Feeding the whole history to the score network, however, creates an input whose dimension grows with $t$ and leads to the curse of dimensionality. Thus, we use an RNN module $R_\theta$ to compress the information in $s^{1:t}$ into a fixed-dimensional feature $f^t$. More specifically, inspired by papers on conditional diffusion models, such as \citet{rasul2021} and \citet{yan2021}, we define feature variables $f^{1:T}$ recursively by $f^t = R_\theta(s^t,f^{t-1})$, where $R_\theta:\mR^d\times \mR^{d'}\to \mR^{d'}$ and $d'$ is the dimension of the feature space. With this approach, we propose to jointly train the RNN encoder $R_\theta$ and the score network $s_\theta: [0,\Ti]\times \mR^{d'}\times \mR^d\to \mR^d$ using the following variant of \eqref{conditional:DSM:noise}}:
\begin{equation}
\label{conditional:DSM:implementation}
\frac{1}{(T-1)M}\sum_{t=1}^{T-1}\sum_{m=1}^{M}\Big|\sqrt{h_2(\tau_{(m)})}s_\theta\Big(\tau_{(m)},f^t_{(m)},s^{t+1}_{(m)}e^{-\tau_{(m)}}+\sqrt{h_2(\tau_{(m)})}\bz_{(m)}\Big)+\bz_{(m)}\Big|^2.
\end{equation}
Here, $\{s^{1:T}_{(m)}\}_{m=1}^M$ is a batch of data, and $\{f^{1:T}_{(m)}\}_{m=1}^M$ are the corresponding feature variables obtained by feeding the data to the RNN encoder $R_\theta$. The variables $\{\tau_{(m)}\}_{m=1}^M$ are randomly sampled from $[0,\Ti]$, and $\{\bz_{(m)}\}_{m=1}^M$ are sampled from the noise distribution $\N(0,I)$. Although only $s_\theta$ explicitly appears in \eqref{conditional:DSM:implementation}, $R_\theta$ is implicitly trained by backpropagation. Algorithm \ref{algo:training-s-r} provides the pseudocode for the training procedure.

{ We emphasize that, even though the theoretical guarantees in Section \ref{sec:mainresults} are stated for full-history conditional kernels, the RNN encoder used here is only an implementation-level representation of the conditioning history rather than a new modeling perspective. If the representation is sufficient, i.e., $s^{1:t}=(s^{1:t})'$ if and only if $f^t=(f^t)'$, replacing the full history by $f^t$ is a reparametrization. If the RNN compression is lossy, the resulting representation error is part of the overall neural-network approximation error $\epsilon_{\rm score}$. This term may enlarge the numerical error attainable by the chosen architecture, while leaving the theoretical results unchanged because they take $\epsilon_{\rm score}$ as an input. We leave the detailed analysis to future work.}

\begin{algorithm}
\caption{Training the score network and the RNN encoder}
\label{algo:training-s-r}

{\bf Initialization:} Initialize $s_\theta$ and $R_\theta$; collect the training data.

\begin{algorithmic}[1]
\FOR{\texttt{epochs}}
\STATE{Sample a batch of data $\{s^{1:T}_{(m)}\}_{m=1}^M$.}

\STATE{Compute $\{f^{1:T}_{(m)}\}_{m=1}^M$ using $R_\theta$.}

\STATE{Sample $\{\tau_{(m)}\}_{m=1}^M$ uniformly from $[0,\Ti]$.}

\STATE{Sample $\{\bz_{(m)}\}_{m=1}^M$ from $\N(0,I)$.}

\STATE{Update $\theta$ by minimizing \eqref{conditional:DSM:implementation}.}

\ENDFOR

\end{algorithmic}

{\bf Outputs:} Trained score network $s_\theta$ and RNN encoder $R_\theta$.

\end{algorithm}

For the sampling procedure, we consider an implementable version of Algorithm \ref{adaptive-sampling}, which has the following two refinements: we use the feature variables obtained from the RNN encoder to address high dimensionality, and we use the conventional predictor-corrector sampler to reduce discretization error from the SDE solver. The sampling procedure is summarized in Algorithm \ref{adaptive_sampling_implementation}.

\begin{algorithm}
\caption{Adaptive sampling: implementation}
\label{adaptive_sampling_implementation}

{\bf Inputs:} Trained score network $s_\theta$ and RNN encoder $R_\theta$; number of time-discretization (predictor) steps $\Npre$; number of corrector steps $\Ncor$; { predictor grid $\{\tau_k\}_{k=0}^{\Npre}$; corrector step sizes $\{\eta_k\}_{k=0}^{\Npre-1}$}; initial feature $f^0$.

\begin{algorithmic}[1]
\STATE{$f \leftarrow f^0$.}
\STATE{Sample $s^{1:T}\sim \N(0_{\mR^{dT}},I_{dT\times dT})$.}

\FOR{$t=1,\cdots, T$}

\FOR{$k=0,1,\cdots,\Npre-1$}

\STATE{ $s^{t}\leftarrow s^{t}+ \Big(s^{t}+2s_\theta\big(\Ti-\tau_k,f, s^{t}\big)\Big)\cdot (\tau_{k+1}-\tau_k)+\sqrt{2}\bep_k$, where { $\bep_k \sim \N(0_{\mR^d},(\tau_{k+1}-\tau_k)I_{d\times d})$}.
}

\STATE{\# \texttt{The predictor step.}}

\FOR{$l=0,1,\cdots, \Ncor-1$}

\STATE{
$s^{t}\leftarrow s^{t} + { \eta_k}s_\theta\big(\Ti-\tau_k,f, s^{t}\big) + \sqrt{2}\bep'_{k,l} $, where { $\bep'_{k,l}\sim \N(0_{\mR^d},\eta_k I_{d\times d})$}.
}

\STATE{\# \texttt{The corrector step.}}

\ENDFOR
\ENDFOR

\STATE{
$f\leftarrow R_\theta(s^{t},f) $. \quad \#\texttt{Update the feature from $f^{t-1}$ to $f^{t}$.}
}
\ENDFOR

\end{algorithmic}
{\bf Outputs:} $s^{1:T}$ and the corresponding feature path $f^{1:T}$.
\end{algorithm}

\subsection{Numerical experiments}\label{sec:exp}
% {\color{red}
% \begin{figure}\label{fig:model_pipline}
%   \centering
%   \includegraphics[width=\linewidth]{model_pipline.pdf}
%   \caption{Model pipeline}
% \end{figure}
% }

This section presents numerical results demonstrating the feasibility of the proposed algorithms, including adaptive sampling through the diffusion model and the policy-gradient algorithm for solving portfolio selection problems. We conduct both synthetic-data and real-data experiments. All code is available at \href{https://github.com/fy-yuan/diffusion_dynamic_mv}{https://github.com/fy\-yuan/diffusion\_dynamic\_mv}.

{ 
\subsubsection{Synthetic data: ARMA returns}
\label{exp:arma}

We first isolate the adaptive sampling component in a setting where the conditional data law is known exactly. Let $X_t\in\mR^5$ follow a stationary vector ARMA$(4,3)$ process
\begin{equation}
X_t
=0.10X_{t-1}+0.04X_{t-2}+0.02X_{t-3}+0.01X_{t-4}
+\varepsilon_t
+0.06\varepsilon_{t-1}+0.03\varepsilon_{t-2}+0.01\varepsilon_{t-3},
\end{equation}
where
\begin{equation}
\varepsilon_t\overset{\mathrm{i.i.d.}}{\sim}\N(0,\Sigma_\varepsilon),
\qquad
\Sigma_\varepsilon
=0.01^2\left(0.70I_5+0.30\boldsymbol{1}\boldsymbol{1}^{\mt}\right).
\end{equation}
Thus, the scalar lag coefficients are shared across the five components, the innovation standard deviation is $0.01$, and every pair of contemporaneous innovation components has a correlation of $0.30$. We draw the initial state from the stationary Gaussian law and retain $8{,}000$ return observations.

The first $6{,}400$ observations are used to train the model and fit a componentwise standardizer, the next $1{,}576$ are used for fixed-seed validation, and the remaining $24$ are unused. The score model is a same-resolution U-Net with a base width of 32 and a depth of 3, and the conditioner is a two-layer LSTM with a hidden dimension of 32 and a dropout rate of $0.1$. The AdamW optimizer uses a learning rate of $10^{-3}$, a weight decay of $10^{-2}$, and a batch size of 256; an EMA with a decay rate of $0.999$ is used for validation and sampling.

The diffusion model uses the linear VP-SDE schedule $\beta_{\min}=0.01$ and $\beta_{\max}=20$. The reverse-time predictor--corrector sampler uses 500 steps, one Langevin corrector step with a signal-to-noise ratio of $0.16$, and a numerical terminal time of $\epsilon_{\rm num}=10^{-3}$. The selected context used for evaluation is encoded once. Each generated return is then used to update the persistent LSTM state, which conditions the next one-step sampling operation, thereby implementing the adaptive state recursion in Algorithm \ref{adaptive_sampling_implementation}.

\begin{figure}[H]
    \centering
    \includegraphics[width=\linewidth]{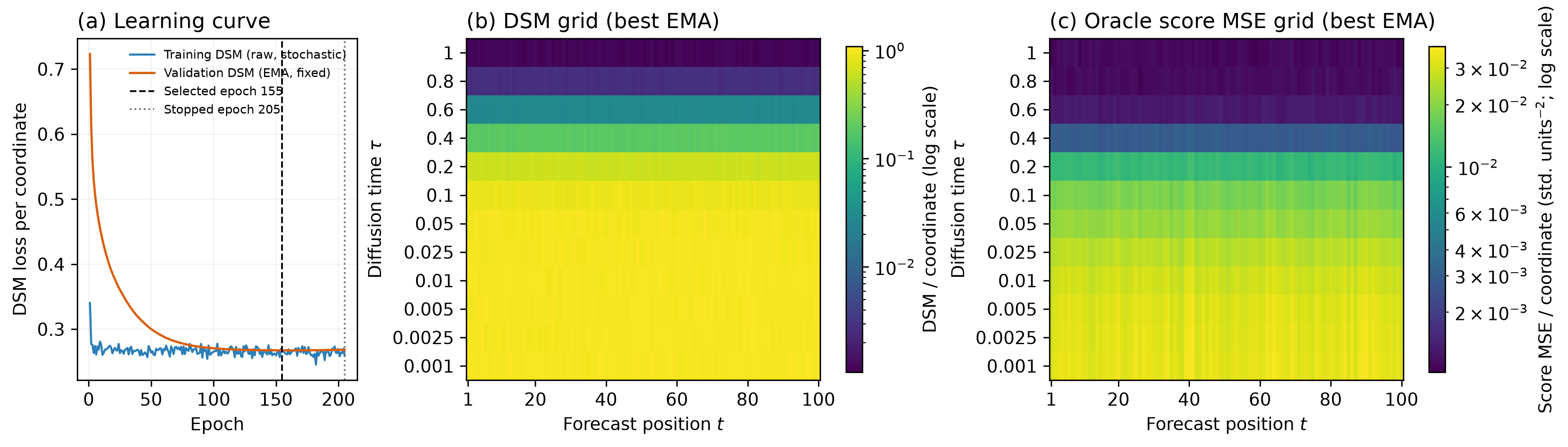}
    \caption{Training and selected-checkpoint score diagnostics for the synthetic ARMA experiment. Panel (a) shows the training DSM loss and fixed-validation DSM loss evaluated with EMA weights; the vertical lines mark the selected checkpoint at epoch 155 and the stopping point at epoch 205. Panels (b) and (c) evaluate the selected checkpoint over forecast position $t=1,\ldots,100$ and VP diffusion time $\tau$. Panel (b) reports the dimensionless DSM loss on 256 fixed validation histories, whereas panel (c) reports score MSE (in training-standardized units$^{-2}$) against the exact Gaussian conditional score on 256 independently generated stationary ARMA histories.}
    \label{arma_score_diagnostic_figure}
\end{figure}

Figure \ref{arma_score_diagnostic_figure} provides an indicative check of score learning. The fixed-validation criterion selects EMA epoch 155, with a DSM loss of $0.2674$, and training stops at epoch 205. On the independent oracle grid, the score MSE varies little with forecast position: its mean over $t=76,\ldots,100$ is $1.008$ times that over $t=1,\ldots,25$, and its mean over diffusion times at $t=100$ is $1.015$ times that at $t=1$. Thus, within the tested 100-step horizon, score-estimation error does not exhibit significant accumulation with $t$. The main variation occurs across diffusion time, with larger raw score error near the numerical low-noise endpoint. 

For evaluation, we condition both the learned sampler and an exact Gaussian oracle on the same 24-period context. We draw $N_{\rm path}=2{,}000$ independent paths of length $T_{\rm pred}=100$ from each conditional law. For any prefix length $T'\leq T_{\rm pred}$, any source $q\in\{{\rm gen},{\rm oracle}\}$, and any component $j$, let $\widehat\mu_{q,j}$ and $\widehat\sigma_{q,j}$ be the mean and the sample standard deviation obtained by pooling the $N_{\rm path}T'$ observations, let $\widehat a_{q,j}$ be the lag-one autocorrelation computed from the $N_{\rm path}(T'-1)$ within-path adjacent pairs, and let $\widehat C_q$ be the cross-sectional correlation matrix obtained by pooling the $N_{\rm path}T'$ rows. We report
\begin{align}
\Delta_{\rm mean}
&=\frac{1}{d}\sum_{j=1}^d
\left|\widehat\mu_{{\rm gen},j}-\widehat\mu_{{\rm oracle},j}\right|,
&
\Delta_{\rm sd}
&=\frac{1}{d}\sum_{j=1}^d
\left|\widehat\sigma_{{\rm gen},j}-\widehat\sigma_{{\rm oracle},j}\right|,
\nonumber\\
\Delta_{\rm acf}
&=\frac{1}{d}\sum_{j=1}^d
\left|\widehat a_{{\rm gen},j}-\widehat a_{{\rm oracle},j}\right|,
&
\Delta_{\rm corr}
&=\frac{1}{d(d-1)}\sum_{i\ne j}
\left|\widehat C_{{\rm gen},ij}-\widehat C_{{\rm oracle},ij}\right|.
\end{align}

Table \ref{output_arma_diagnostics} reports the standardized-scale errors for the forecast horizons $T'\in\{10,50,100\}$. The MAE of the sample mean is largest when only the first 10 forecast steps are pooled and decreases as the pooling horizon expands, whereas the dispersion and dependence errors vary only modestly across $T'$.

\begin{table}[H]\centering
\caption{Conditional-generation MAEs on the standardized-return scale. Each row pools forecast periods $1$ through $T'$ from the same 2,000 generated and 2,000 oracle paths.}
\label{output_arma_diagnostics}
\begin{tabular}{c|rrrr}
$T'$ & Mean & Std.\ dev. & Lag-one ACF & Off-diag.\ correlation \\ \hline
$10$  & 0.0895 & 0.0200 & 0.0146 & 0.0619 \\
$50$  & 0.0403 & 0.0250 & 0.0159 & 0.0545 \\
$100$ & 0.0334 & 0.0259 & 0.0127 & 0.0562
\end{tabular}
\end{table}

Figures \ref{arma_distribution_figures} and \ref{arma_qq_figure} compare the pooled marginal distributions on the training-standardized scale. The boxplots and QQ panels show close agreement across all five components, with modest discrepancies in location, scale, and tail behavior. These figures assess the pooled conditional marginals, while the conclusions regarding autocorrelation and cross-sectional dependence rely on the path-based diagnostics in Table \ref{output_arma_diagnostics}.

Figure \ref{arma_conditional_path_figure} provides a complementary fixed-context trajectory diagnostic for the first component on the raw-return scale. The generated mean tracks the analytical oracle mean, with the largest discrepancy occurring over the first few forecast periods, and the empirical pointwise prediction band is close to the oracle band.

\begin{figure}[H]
\centering
\includegraphics[width=0.75\linewidth]{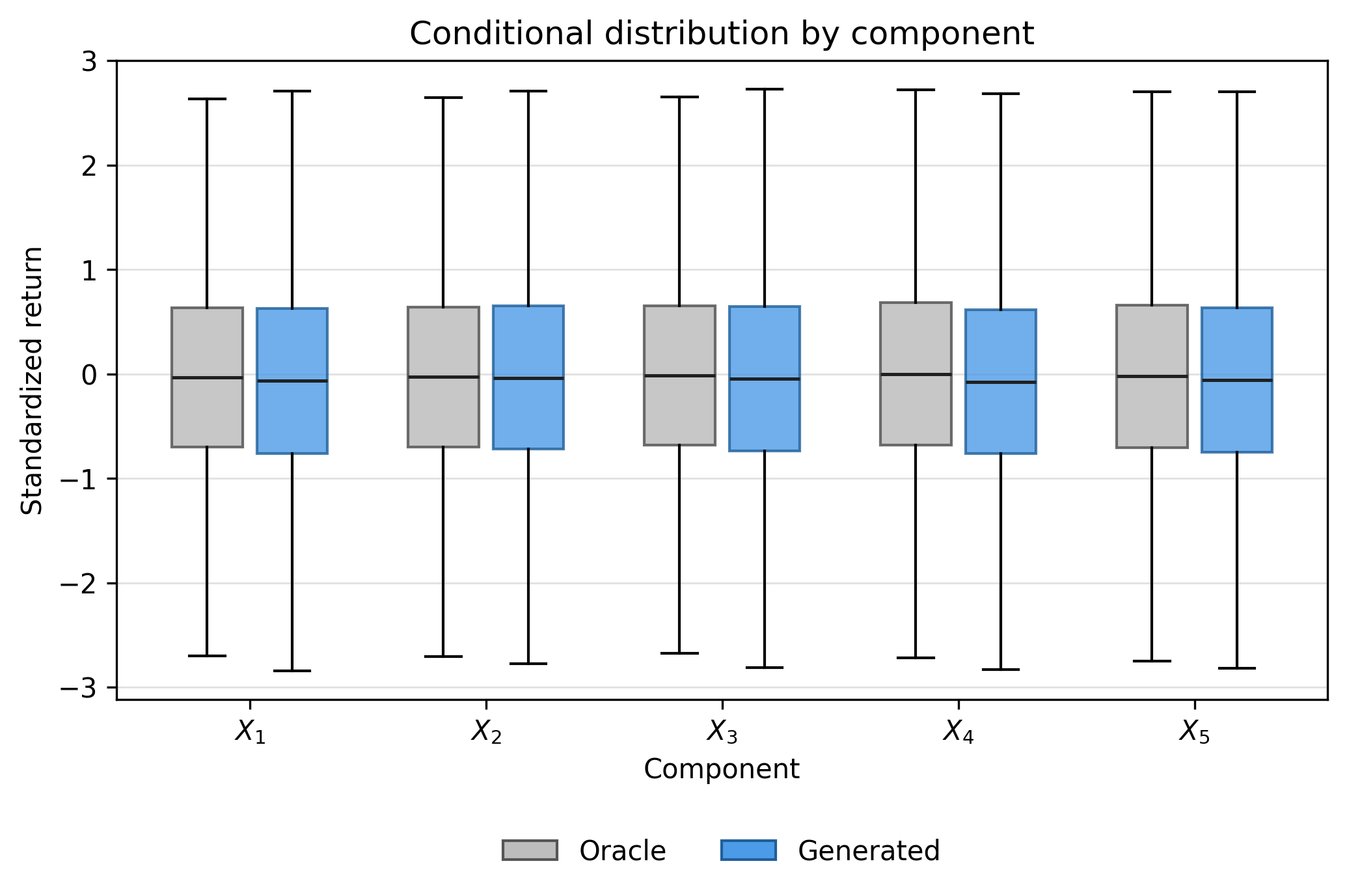}
\caption{Boxplots for generated and oracle samples.}
\label{arma_distribution_figures}
\end{figure}

\begin{figure}[H]
\centering
\includegraphics[width=0.88\linewidth]{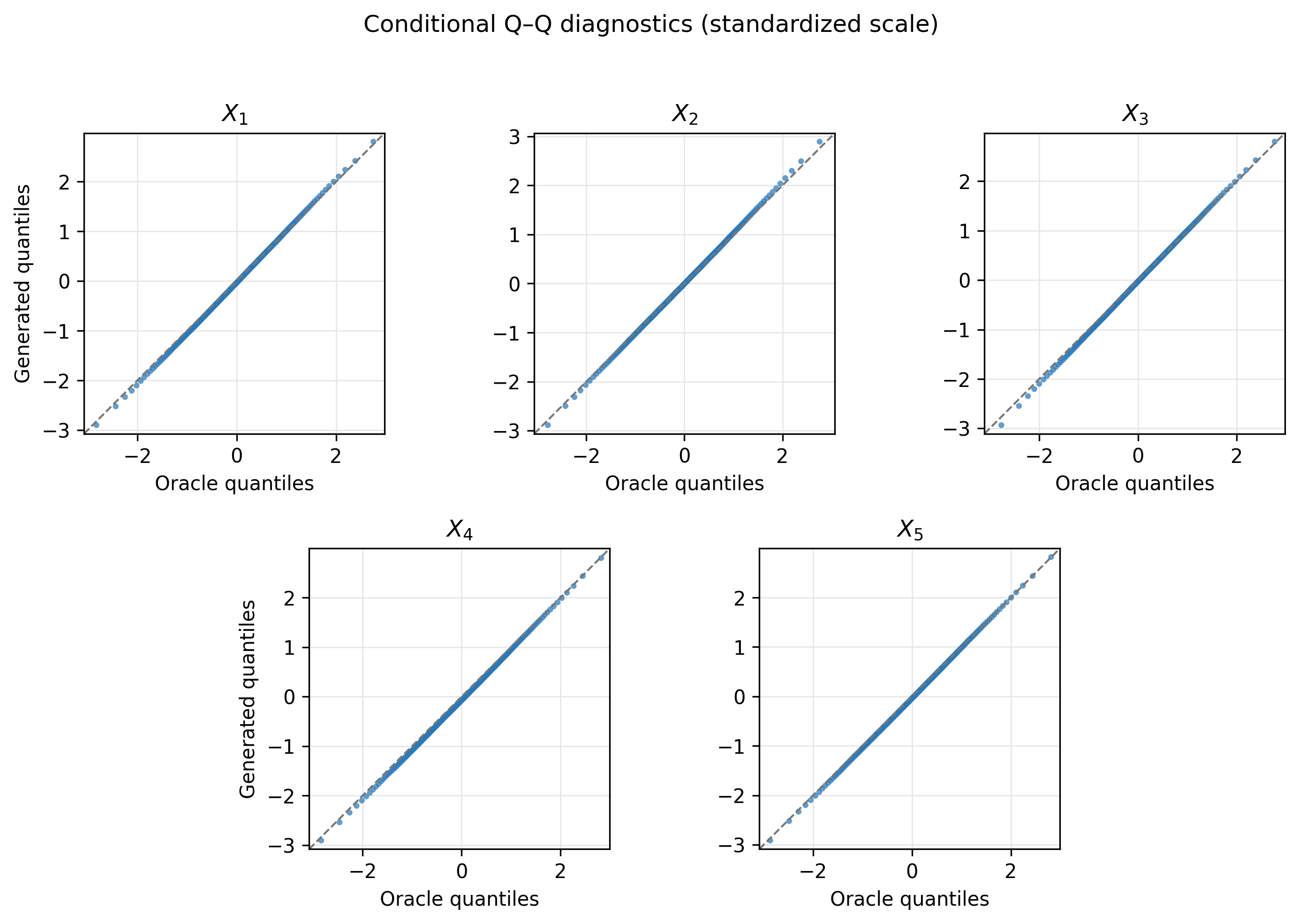}
\caption{QQ plots for generated and oracle samples.}
\label{arma_qq_figure}
\end{figure}

\begin{figure}[H]
\centering
\includegraphics[width=\linewidth]{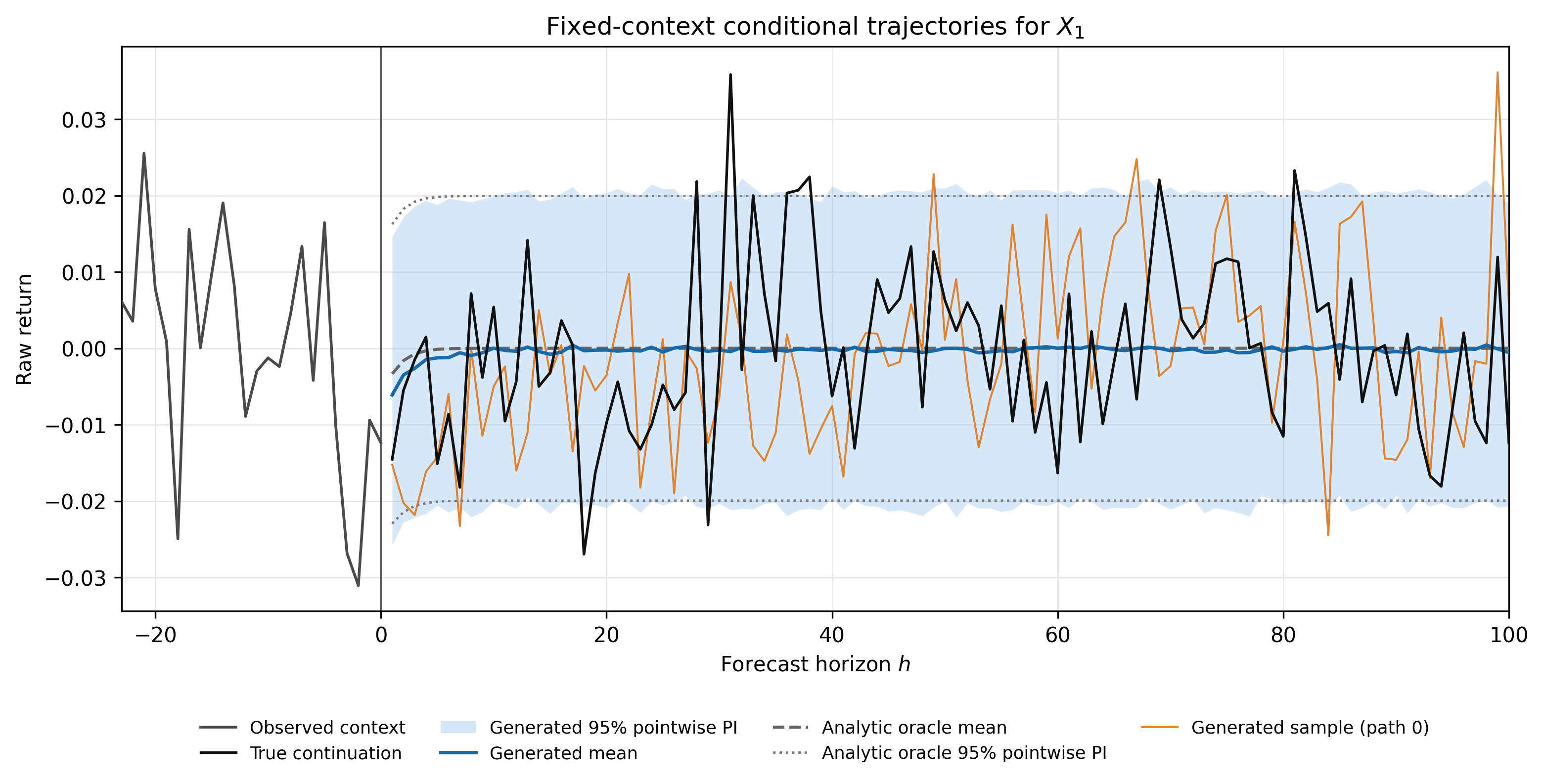}
\caption{Selected-context conditional trajectories for the first component of $X_t$ on the raw-return scale.}
\label{arma_conditional_path_figure}
\end{figure}

}
\subsubsection{Real data: FF-10 dataset}\label{exp:real}

The preceding experiment examines our adaptive sampling method and assesses its generation quality in a synthetic setting. We now use Algorithm \ref{adaptive_sampling_implementation} to construct scenario paths from distributions learned from real data and use these paths with the downstream TD3 solver summarized in Algorithm \ref{FVI-Q}. This subsection assesses the real-world applicability of the framework illustrated in Figures \ref{fig:roadmap} and \ref{fig:roadmap:implementation}. Our empirical tests employ the value-weighted 10 Industry Portfolios from the Kenneth R. French Data Library, which aggregate all common stocks on the CRSP tape into ten broad industry groups---Non-Durables (``NoDur''), Durables (``Durbl''), Manufacturing (``Manuf''), Energy (``Enrgy''), HiTech (``HiTec''), Telecommunications (``Telcm''), Utilities (``Utils''), Shops (``Shops''), Health (``Hlth''), and Others (``Other''). We construct monthly simple returns by compounding the daily value-weighted portfolio returns; no risk-free rate is subtracted. The sample spans July 1926 through February 2025 and is split chronologically as follows:
\begin{enumerate}

\item The initial test context comprises 12 months from March 2009 through February 2010. The walk-forward evaluation then covers March 2010 through February 2025, for a total of 180 monthly decisions. At every decision date, the context ends in the preceding month.

\item The validation targets cover March 1999 through February 2009, for a total of 120 months.

\item The training observations cover July 1926 through February 1999. After allowing for the 12-month context, the training targets cover July 1927 through February 1999, for a total of 860 months. All return and feature standardizers are fitted using training data only.
\end{enumerate}

For training, we use a {\it same-resolution U-Net} inspired by \citet{ronneberger2015u} as the score network $s_\theta$ and a two-layer LSTM \citep{hochreiter1997long} with a hidden dimension of 20 and a dropout rate of 0.1 as the RNN encoder $R_\theta$. The full-history generator conditions on the most recent 12 months. We also train a separate one-lag generator of the same architecture for the Markov baseline introduced below. Both models use the variance-preserving stochastic differential equation (VP-SDE):
\[
\md X_\tau
=
-\tfrac12\,\beta(\tau)\,X_\tau\,\md\tau
+\sqrt{\beta(\tau)}\,\md B_{\tau},
\qquad
\beta(\tau)=\beta_{\min}+(\beta_{\max}-\beta_{\min})\,\tau,
\]
where $\beta_{\min}=0.01$ and $\beta_{\max}=20$. Training uses continuously sampled diffusion times, whereas the reverse-time predictor--corrector sampler uses $\Npre=500$ predictor steps. The score network $s_\theta$ and the RNN encoder $R_\theta$ are trained according to Algorithm \ref{algo:training-s-r}. Detailed model architectures and training configurations are available in our GitHub repository. Optimization uses AdamW with a learning rate of $10^{-3}$, $(\beta_{1},\beta_{2})=(0.9,0.999)$, and a weight decay of $10^{-2}$. A cosine-annealing scheduler runs for at most 1,000 epochs following a 10-epoch warm-up; early stopping is triggered after 500 epochs without improvement. Gradients are clipped to a maximum norm of 1, and an EMA of the model weights with a decay rate of 0.999 is used for validation and sampling.

For testing, we set the time horizon $T=12$ (i.e., one year) for the TD3 agent and use an extended test horizon $T_{\rm test}=15T$ (i.e., 15 years). We use the S\&P 500 price return as an external market benchmark and compare the following eight FF-10 portfolio strategies:
\begin{enumerate}
    \setlength{\itemsep}{0pt}
    \setlength{\parsep}{0pt}
    \item {\bf Equal Weight (EW)}: $a\equiv (1/d,1/d,\cdots,1/d)$;
    \item {\bf History-based Markowitz (HistMarkowitz)}: For each $t=1,2,\cdots, T_{\rm test}$, the action $a$ is determined by solving the Markowitz problem with no-short-selling, no-borrowing, and full-investment constraints. The mean and covariance are estimated using the sample mean and sample covariance of the most recent 60 months of data, as in \citet{demiguel2009optimal}. In other words, this is a monthly rebalancing Markowitz strategy based on historical estimates;
    \item {\bf Generative-model-based Markowitz (GenMarkowitz)}: For each $t=1,2,\cdots, T_{\rm test}$, we collect the most recent $T$ months (i.e., one year) of data as the context window and use it as the initial feature $f^0$ in the diffusion model (see Algorithm \ref{adaptive_sampling_implementation}). We then generate 500 {\it one-month-ahead predictions} for $s^t$ and use them to solve the constrained Markowitz problem. In other words, this is a monthly rebalancing Markowitz strategy based on the generative model;
    \item {\bf Historical-scenario TD3 (HistTD3)}: At each annual cutoff, we construct 12-month scenarios from returns observed before deployment. Training and validation donor periods are disjoint, and blocks are selected without replacement according to their proximity to the latest return state. The policy uses the latest 10-dimensional raw-return vector, standardized using training-period statistics, as its market state;
    \item {\bf Parametric TD3 (ParametricTD3)}: At each annual cutoff, we fit an expanding-window Gaussian VAR(1), with a full transition matrix and a residual covariance matrix, and recursively simulate separate training and validation scenarios. The policy uses the same latest-return state as HistTD3;
    \item {\bf Markov diffusion TD3 (MarkovTD3)}: We use the separate one-lag conditional diffusion generator to construct the scenario pools. The policy state contains the RNN feature produced by this generator;
    \item {\bf No-corrector GenTD3 (NoCorrectorGenTD3)}: We retain the full-history generator, its RNN encoder, and the reverse-diffusion predictor, but disable the Langevin corrector. Thus, this method isolates the effect of the corrector while retaining the same 12-month conditioning information;
    \item {\bf Generative-model-based TD3 (GenTD3)}: At each annual cutoff, we condition the full diffusion generator on the latest 12-month context and use Algorithm \ref{adaptive_sampling_implementation} to construct one-year scenario pools. The policy uses the corresponding frozen RNN feature as its market state.
\end{enumerate}

\begin{algorithm}[H]\footnotesize
\caption{Policy-gradient solver with generative scenario pools}
\label{FVI-Q}

\textbf{Initialization}: Initialize Q-network $Q_{\alpha}$ and policy network $\pi_{\beta}$; initialize target networks $\alpha_{\targ}\leftarrow \alpha$ and $\beta_\targ\leftarrow \beta$; choose a target-update rate $\rho\in(0,1)$ and an exploration-noise distribution; set up an empty replay buffer $\D$ and a scenario pool $\Sigma$ of size $L$; train the score network $s_\theta$ and RNN encoder $R_\theta$ using Algorithm \ref{algo:training-s-r}. Collect context window data $s^{-T:-1}$.

\begin{algorithmic}[1]
\STATE{Use $s^{-T:-1}$ to compute the initial feature $f^0$.}
\STATE{Use Algorithm \ref{adaptive_sampling_implementation} to construct the indexed scenario pool
\[
\Sigma=\{(s_\ell^{1:T},f_\ell^{1:T})\}_{\ell=1}^L.
\]
}

\FOR{\texttt{steps}}
\STATE{Initialize $w=1$; sample an index $\ell$ and retrieve $(s_\ell^{1:T},f_\ell^{1:T})$ from $\Sigma$; sample $c\sim\nu_c$, where $\nu_c$ is a prescribed training distribution for terminal targets.}

\FOR{$t=1$ \TO $T-1$}
\IF{\texttt{warm-up}}
\STATE{$a\sim \mathrm{Uniform}(\mathcal K)$.}
\ELSE
\STATE{Sample exploration noise $\epsilon$ and set $a\leftarrow \operatorname{Proj}_{\mathcal K}(\pi_\beta(t,w,f_\ell^t,c)+\epsilon)$.}
\ENDIF

\STATE{Store $(t,\ell,w,a,c)$ in $\D$.}
\STATE{$w\leftarrow w+a^\mt(s_\ell^{t+1}-s_\ell^t)$.}

\IF{not \texttt{warm-up}}
\STATE{Sample a minibatch $\{(t_j,\ell_j,w_j,a_j,c_j)\}_{j=1}^B$ from $\D$.}
\STATE{For each $j$, retrieve $(s_{\ell_j}^{t_j},s_{\ell_j}^{t_j+1},f_{\ell_j}^{t_j},f_{\ell_j}^{t_j+1})$ from $\Sigma$ and set
\[
w'_j=w_j+a_j^\mt(s_{\ell_j}^{t_j+1}-s_{\ell_j}^{t_j}).
\]
}
\STATE{For each $j$, define the Bellman target
\[
y_j=
\begin{cases}
Q_{\alpha_\targ}\!\Big(t_j+1,w'_j,f_{\ell_j}^{t_j+1},
\pi_{\beta_\targ}(t_j+1,w'_j,f_{\ell_j}^{t_j+1},c_j),c_j\Big),
& t_j<T-1,\\
|w'_j-c_j|^2,
& t_j=T-1.
\end{cases}
\]
}
\STATE{Update $\alpha$ by minimizing
\[
\frac{1}{B}\sum_{j=1}^B
\Big|Q_\alpha(t_j,w_j,f_{\ell_j}^{t_j},a_j,c_j)-y_j\Big|^2.
\]
}
\STATE{Update $\beta$ by gradient descent using
\[
\frac{1}{B}\sum_{j=1}^B
\nabla_a Q_\alpha\Big(t_j,w_j,f_{\ell_j}^{t_j},
\pi_\beta(t_j,w_j,f_{\ell_j}^{t_j},c_j),c_j\Big)
\nabla_\beta\pi_\beta(t_j,w_j,f_{\ell_j}^{t_j},c_j).
\]
}
\STATE{Update target networks:
\begin{align}
&\alpha_\targ\leftarrow\rho\alpha_\targ+(1-\rho)\alpha,\\
&\beta_\targ\leftarrow\rho\beta_\targ+(1-\rho)\beta.
\end{align}
}
\ENDIF
\ENDFOR
\ENDFOR

\STATE{Find the optimal multiplier by maximizing
\[
-\frac{\gamma}{2}Q_\alpha\big(1,0,f^1,\pi_\beta(1,0,f^1,c),c\big)+c.
\]
}
\end{algorithmic}

\textbf{Output}: Policy network $\pi_\beta$.
\end{algorithm}

To instantiate the downstream decision step, we train a generative-scenario-based TD3 solver using deterministic policy-gradient machinery \citep{ddpg_original,td3}. For readability, Algorithm \ref{FVI-Q} uses raw wealth and one critic symbol. For each fixed terminal target, the implementation uses an algebraically equivalent centered-cost parameterization and standard TD3 refinements---twin critics, target-policy smoothing, and delayed actor updates---which, together with numerical hyperparameters, are omitted from the pseudocode.

All five TD3 strategies follow the same 15-stage annual walk-forward schedule. Before each stage, the latest available context is used to construct separate training and validation pools. GenTD3, ParametricTD3, MarkovTD3, and NoCorrectorGenTD3 use 384 training paths and 256 validation paths per stage; HistTD3 uses the 128 nearest eligible training blocks and the 109 complete blocks available in the disjoint validation donor period. At each deployment boundary, the terminal-wealth target is recalibrated on a 65-point grid over $[0,5]$ using only the actor, the twin critics, and the pre-deployment state. All FF-10 strategies are long-only and fully invested, are rebalanced monthly, and exclude transaction costs. If a diffusion-generated path contains an infinite value or a simple return no greater than $-100\%$, the entire return--feature path is discarded and resampled under the fixed-seed protocol.

\begin{remark}
    For clarity, the score network $s_\theta$ and the RNN encoder $R_\theta$ are trained first and remain frozen during TD3 training. During walk-forward evaluation of the diffusion-backed TD3 strategies, the generator is not resampled to select actions. Instead, the actor uses the realized context, the method-specific state, and the annually calibrated target. See Figure \ref{fig:roadmap:implementation} in the Introduction.
\end{remark}

We report annualized return, annualized volatility, Sharpe ratio, Sortino ratio, maximum drawdown, and Calmar ratio for $\gamma=3$ in Table \ref{output_reaL_metrics}; the Sharpe and Sortino ratios use a zero reference rate. Figure \ref{output_portfolio_wealth} presents the corresponding wealth trajectories.

\begin{table}[H]\centering
\caption{Walk-forward performance of the FF-10 strategies and the S\&P 500 benchmark with $\gamma = 3$}
\label{output_reaL_metrics}
\resizebox{\textwidth}{!}{%
\begin{tabular}{l|rrrrrr}
Method & Ann. return & Ann. volatility & Sharpe & Sortino & Max drawdown & Calmar \\ \hline
S\&P 500              & 11.89\% & 14.53\% & 0.8180 & 0.7585 & -24.77\% & 0.4799 \\
EW                     & 13.26\% & 15.02\% & 0.8831 & 0.8310 & -22.93\% & 0.5783 \\
HistMarkowitz          & 13.19\% & 16.71\% & 0.7892 & 0.7723 & -37.59\% & 0.3508 \\
GenMarkowitz           & {\bf 20.07\%} & 21.71\% & 0.9246 & 1.1016 & -23.12\% & 0.8681 \\
HistTD3                &  9.59\% & 17.37\% & 0.5523 & 0.5619 & -27.82\% & 0.3448 \\
ParametricTD3          & 11.14\% & 15.50\% & 0.7187 & 0.6767 & -25.52\% & 0.4364 \\
MarkovTD3              & 12.60\% & 14.56\% & 0.8657 & 0.8731 & -17.69\% & 0.7122 \\
NoCorrectorGenTD3      & 14.06\% & {\bf 14.03\%} & 1.0025 & 0.9619 & -21.56\% & 0.6522 \\
GenTD3                 & 18.79\% & 15.34\% & {\bf 1.2246} & {\bf 1.2716} & {\bf -15.71\%} & {\bf 1.1959}
\end{tabular}%
}
\end{table}

\begin{figure}[H]
\centering
\includegraphics[width=0.92\linewidth]{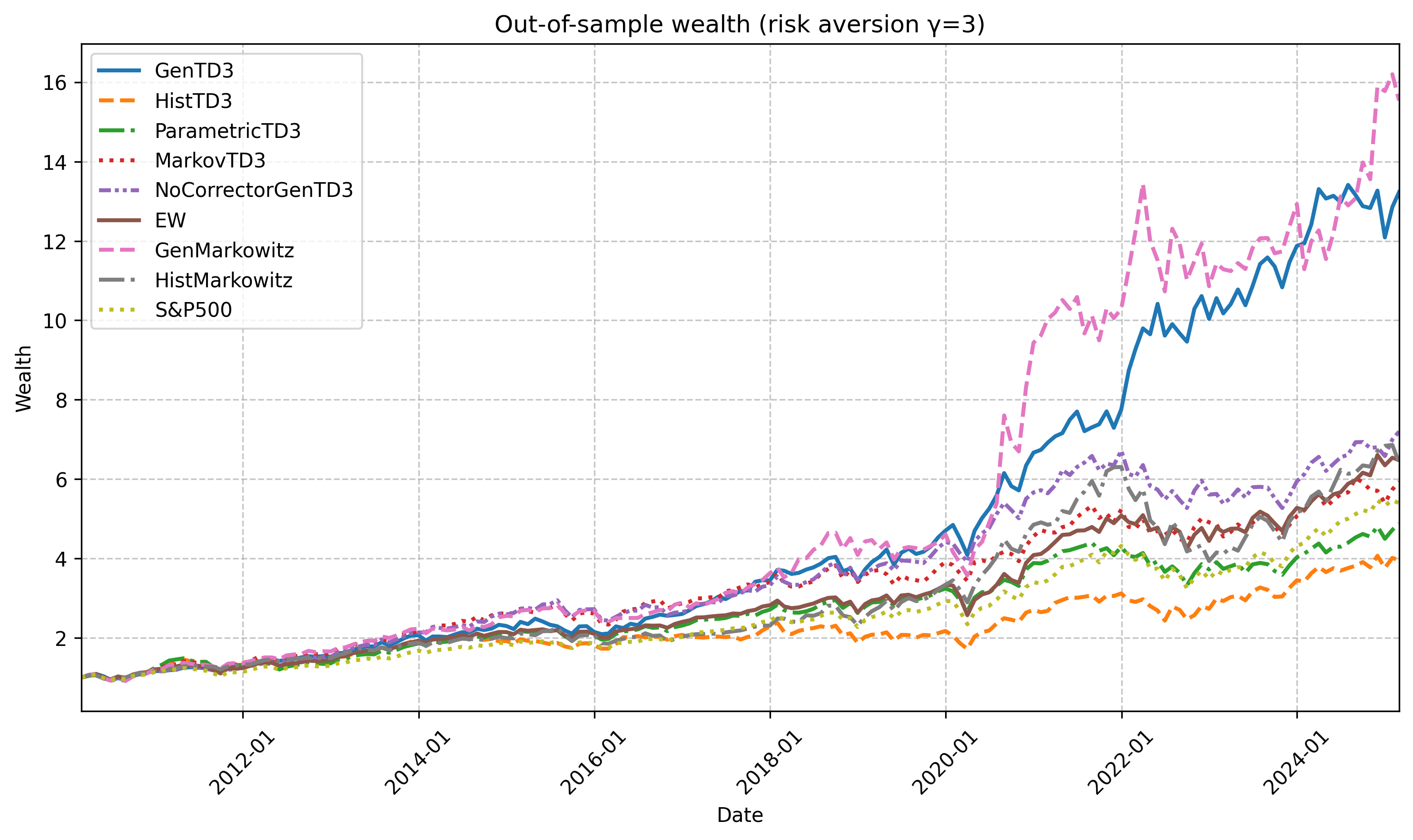}
\caption{Wealth trajectories in the FF-10 walk-forward evaluation with $\gamma=3$}
\label{output_portfolio_wealth}
\end{figure}

GenMarkowitz has both the highest annualized return and the highest volatility. GenTD3 attains an annualized return of 18.79\% and the best Sharpe ratio (1.2246), Sortino ratio (1.2716), maximum drawdown ($-15.71\%$), and Calmar ratio (1.1959). It therefore outperforms GenMarkowitz on each of the remaining five metrics, while NoCorrectorGenTD3 attains the lowest volatility. Relative to this ablation, the full GenTD3 achieves a higher return, better risk-adjusted performance, and better drawdown statistics. MarkovTD3 also trails the full-history GenTD3 on these measures, while HistTD3 and ParametricTD3 both have lower annualized returns and Sharpe ratios than EW. In this fixed-seed experiment, these rankings support the usefulness of longer-history conditioning and the Langevin corrector, while also showing that TD3 alone does not ensure superior performance.

{ 
\subsubsection{Real data: FF-30 dataset}\label{exp:real:ff30}

To assess whether the proposed methods remain effective as the number of assets increases, we repeat the real-data experiment using the value-weighted 30 Industry Portfolios from the Kenneth R. French Data Library. We use the official ``Average Value Weighted Returns -- Monthly'' panel and convert the reported returns from percentages to decimals; no risk-free rate is subtracted. The training observations span July 1926 through February 1999; after allowing for the 12-month context, the training targets cover July 1927 through February 1999, for a total of 860 months. The 120 validation targets cover March 1999 through February 2009. The initial test context covers March 2009 through February 2010, and the walk-forward evaluation covers March 2010 through February 2025, for a total of 180 monthly decisions.

To accommodate the higher-dimensional data, we increase the hidden dimension of the two-layer RNN encoder from 20 to 48 and the score-network base width from 16 to 32, while retaining a score-network depth of 3 and a time-embedding dimension of 64. We compare five methods: the S\&P~500 benchmark, EW, HistMarkowitz, GenMarkowitz, and GenTD3. The S\&P~500 series remains a price-return benchmark; all the portfolio strategies remain long-only and fully invested, and transaction costs are excluded. We report the same six performance metrics in Table \ref{output_ff30_metrics}; Figure \ref{output_ff30_portfolio_wealth} presents the corresponding wealth trajectories.

\begin{table}[H]\centering
\caption{Walk-forward performance of the FF-30 strategies and the S\&P 500 benchmark with $\gamma = 3$}
\label{output_ff30_metrics}
\resizebox{\textwidth}{!}{%
\begin{tabular}{l|rrrrrr}
Method & Ann. return & Ann. volatility & Sharpe & Sortino & Max drawdown & Calmar \\ \hline
S\&P 500         & 11.89\% & 14.53\% & 0.8180 & 0.7585 & -24.77\% & 0.4799 \\
EW                & 12.19\% & 16.49\% & 0.7391 & 0.7127 & -26.68\% & 0.4568 \\
HistMarkowitz     & 13.73\% & 18.75\% & 0.7325 & 0.7837 & -31.07\% & 0.4420 \\
GenMarkowitz      & {\bf 17.53\%} & 16.77\% & 1.0450 & {\bf 1.0662} & -25.82\% & 0.6787 \\
GenTD3            & 15.01\% & {\bf 13.69\%} & {\bf 1.0964} & 1.0573 & {\bf -20.80\%} & {\bf 0.7218}
\end{tabular}%
}
\end{table}

\begin{figure}[H]
\centering
\includegraphics[width=0.92\linewidth]{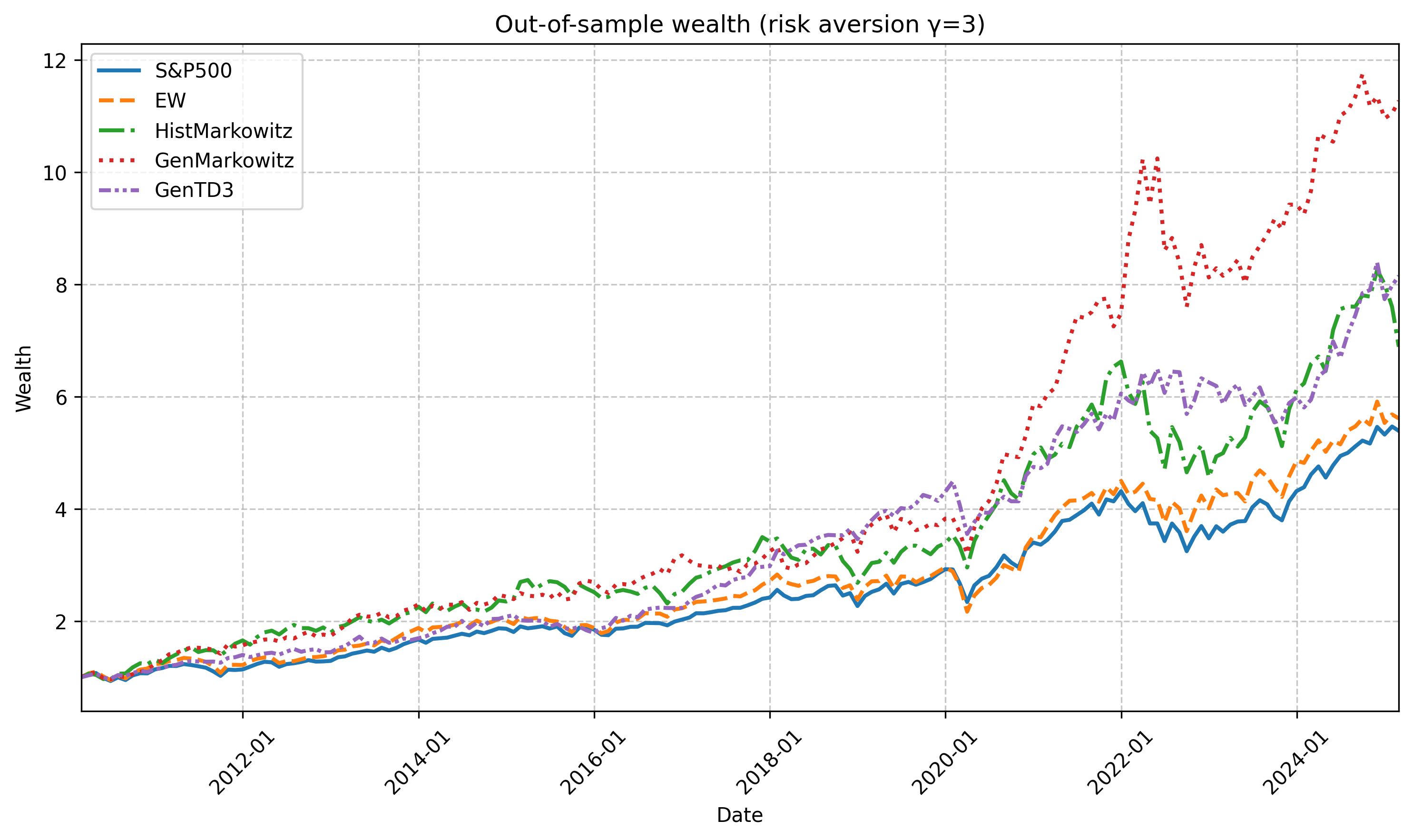}
\caption{Wealth trajectories in the FF-30 walk-forward evaluation with $\gamma=3$}
\label{output_ff30_portfolio_wealth}
\end{figure}

GenMarkowitz attains the highest annualized return and the highest Sortino ratio, whereas GenTD3 attains the lowest annualized volatility and the best Sharpe ratio, maximum drawdown, and Calmar ratio. Both generator-based portfolios outperform the S\&P~500, EW, and HistMarkowitz on annualized return, Sharpe ratio, Sortino ratio, and Calmar ratio; GenTD3 improves on all three benchmarks across all six metrics. Taken together, these results suggest that the proposed methods remain effective as the number of assets increases from 10 to 30 after modest architectural adjustments, including a wider score network.
}

\section{Concluding remarks}\label{conclusion}

In this work, we extend the score-based diffusion-model framework to handle time-series data. We introduce an adaptive training and sampling scheme tailored to dynamic control problems, and we demonstrate its effectiveness on the dynamic mean-variance portfolio selection problem. Our theoretical contributions include a quantitative error bound and a model-stability result that make this application feasible. We implement the adaptive sampler, use its scenarios with downstream policy-gradient algorithms, and present experiments that demonstrate the practical effectiveness of the approach.

Several future directions naturally emerge. {\it On the theoretical side}, one could relax the assumptions on the data distribution, particularly Assumption \ref{assumption:P}. With a more refined approximation analysis, it may be possible to remove Lipschitz continuity in the {\it conditional variable $h^{1:t}$} and impose it only on $x^{t+1}$ in \eqref{Lip:conditionalscore}, which aligns with the usual Lipschitz condition for the score function. Moreover, the error bounds established here are unlikely to be optimal; tightening or even optimizing these rates remains an interesting direction. {\it On the practical side}, further fine-tuning of the dynamic diffusion models and downstream policy-gradient algorithms, together with additional large-scale experiments, is highly desirable. Furthermore, to enhance predictive performance, one could incorporate more covariates into the input of the RNN encoder instead of using only price or return data. A quantitative analysis of the RL component of this work (e.g., a regret analysis) would also be valuable. {  As a scope limitation, we point out that this paper treats the price process as an exogenous market environment. Portfolio decisions affect the investor's wealth and gain processes rather than the conditional law of future prices. Thus, problems with market impact, endogenous price formation, or other decision-dependent state dynamics would require learning controlled transition kernels and fall outside the scope of the present analysis.} We leave these directions for future study.

\appendix

\section{Discussion of assumptions}\label{app:assumption:stheta}

In this appendix, we justify Assumption \ref{assumption:P}. We also demonstrate that the approximating network $s_\theta$ can be selected to satisfy Assumption \ref{assumption:stheta} while retaining the desired approximation power, i.e., while satisfying Assumption \ref{assumption:score-matching-error}.

For convenience, in this appendix we work with the pair $(h^{1:t-1},x^t)$ for $t\in \{1,2,\cdots,T\}$, where $h^{1:t-1}$ is the noisy condition and $x^t$ is the clean next coordinate to be predicted. If $t=1$, then $h^{1:0}$ is void and $\Pi_{h^{1:0}}^{\tau_0,1}=\mP_1$. This corresponds to the initial score-matching step (i.e., \eqref{score-error-initial}).

{ 
\subsection{Assumptions on data distribution}

To elaborate on the Lipschitz assumptions in Assumption \ref{assumption:P}-1 and the growth condition in Assumption \ref{assumption:P}-2, we first observe that, by \eqref{eq:real-score-tau}, the noisy-history target score function $\nabla_x \log p^{\tau_0}_{t}(\tau,x|h^{1:t-1})$ can be written as
\begin{align}
    \nabla_x \log p^{\tau_0}_t(\tau,x|h^{1:t-1})&=\frac{\nabla_x\int_{\mR^d}\phi(\tau,x|x^t)\Pi_{h^{1:t-1}}^{\tau_0,t}(\md x^t) }{\int_{\mR^d}\phi(\tau,x|x^t)\Pi_{h^{1:t-1}}^{\tau_0,t}(\md x^t)}\\
     &= \frac{1}{h_1(\tau)}\nabla _z \log \cZ^\tau_{h^{1:t-1}}\big(z\big)\Big|_{z=x/h_1(\tau)},
    %&=-\frac{\int_{\mR^d}(x-x^th_1(\tau))e^{-\frac{|x-x^th_1(\tau)|^2}{2h_2(\tau)}}\mP_{x^{1:t-1}}(\md x^t)}{h_2(\tau)\int_{\mR^d}e^{-\frac{|x-x^th_1(\tau)|^2}{2h_2(\tau)}}\mP_{x^{1:t-1}}(\md x^t)}\\
    %&=-\frac{x}{h_2(\tau)}+\frac{h_1(\tau)}{h_2(\tau)}\frac{\int_{\mR^d}x^t e^{-\frac{|x-x^th_1(\tau)|^2}{2h_2(\tau)}}\mP_{x^{1:t-1}}(\md x^t)}{\int_{\mR^d} e^{-\frac{|x-x^th_1(\tau)|^2}{2h_2(\tau)}}\mP_{x^{1:t-1}}(\md x^t)}\\
    %&=:-\frac{x}{h_2(\tau)}+G(\tau,x^{1:t-1},x).
\end{align}
because
\begin{align}
    \cZ^\tau_{h^{1:t-1}}(x/h_1(\tau))\label{identity-F-Z}
     \propto \int_{\mR^d}\phi(\tau,x|x^t)\Pi_{h^{1:t-1}}^{\tau_0,t}(dx^t).
\end{align}
Here, for simplicity we denote $h_1(\tau):=e^{-\tau}$, $h_2(\tau):=1-e^{-2\tau}$, $\sigma^2(\tau):=h_2(\tau)/h_1(\tau)^2$, and
\begin{align}
  \cZ^\tau_{h^{1:t-1}}(z)=&\int_{\mR^d}K_\tau(z-x^t)\Pi_{h^{1:t-1}}^{\tau_0,t}(dx^t),\\
  K_\tau(z)=& \exp\bigg(-\frac{|z|^2}{2\sigma^2(\tau)}\bigg).
\end{align}
Consider the ground-truth noisy-history conditional score
\[
\nabla_x\log p^{\tau_0}_t(\tau,x\mid h^{1:t-1}).
\]
Its properties can be studied through $\cZ^\tau_{h^{1:t-1}}$. We propose the following {\it sufficient conditions} on the posterior kernels $\Pi_{h^{1:t-1}}^{\tau_0,t}$, which are fully determined by the data distribution $\mP$ and the fixed early-stopping level $\tau_0$:
\begin{assumption}\phantomsection\label{ass:stability}
  \begin{enumerate}
    \item There exists a universal constant $K>0$ such that, for any $t\in \{1,2,\cdots, T\}$, $h^{1:t-1}\in \mR^{d(t-1)}$, either of the following two conditions holds:
    \begin{enumerate}
    \item $\Pi_{h^{1:t-1}}^{\tau_0,t}$
is supported in $B_K:=\{x\in \mR^{d}: |x|\leq K\}$;
\item $\Pi_{h^{1:t-1}}^{\tau_0,t}$ is a log-concave probability measure. Equivalently, there exists an affine subspace $E_{h^{1:t-1}}\subset \mR^d$ and a log-concave density function $\pi_{h^{1:t-1}}(x^t)$ on $E_{h^{1:t-1}}$ such that $\Pi_{h^{1:t-1}}^{\tau_0,t}(dx^t) = \pi_{h^{1:t-1}}(x^t)\lambda_{E_{h^{1:t-1}}}(dx^t)$, where $\lambda_E$ is the Lebesgue measure on $E$.
\end{enumerate}

\item Define
    \begin{align}
        \mu^{\tau,z}_{h^{1:t-1}}(dx^t) = &\frac{K_\tau(z-x^{t})\Pi_{h^{1:t-1}}^{\tau_0,t}(dx^t)}{\cZ^\tau_{h^{1:t-1}}(z)}.
    \end{align}
    The mean of $\mu^{\tau,z}_{h^{1:t-1}}$, denoted by $m^{\tau,z}_{h^{1:t-1}}:=\mE_{X^t\sim \mu^{\tau,z}_{h^{1:t-1}}}[X^t]$, satisfies the following conditions: there exist constants $L_{\tau_0}$, $C_{\tau_0}>0$ depending on $\tau_0$ such that for any $z\in \mR^d$, $t\in \{1,2,\cdots,T\}$, $\tau\in [\tau_0,\Ti]$, $h^{1:t-1},k^{1:t-1}\in \mR^{d(t-1)}$,
    \begin{align}
        & |m^{\tau,z}_{h^{1:t-1}}-m^{\tau,z}_{k^{1:t-1}}|\leq L_{\tau_0}(1+|z|)|h^{1:t-1}-k^{1:t-1}|,\label{posterior-Lip}\\
        &|m^{\tau,z}_{h^{1:t-1}}|\leq C_{\tau_0}(1+|z|+|h^{1:t-1}|).\label{posterior-growth}
    \end{align}
\end{enumerate}
\end{assumption}

\begin{proposition}
    If the posterior kernels $\Pi_{h^{1:t-1}}^{\tau_0,t}$ satisfy Assumption \ref{ass:stability}, then the ground-truth noisy-history conditional score function $(h^{1:t-1},x)\mapsto \nabla_x\log p^{\tau_0}_t(\tau,x|h^{1:t-1})$ satisfies Assumptions \ref{assumption:P}-1 and \ref{assumption:P}-2.
\end{proposition}

\begin{proof}
	    Let us first show that, under Assumption \ref{ass:stability}-1, the score $z\mapsto\nabla_z\log \cZ^\tau_{h^{1:t-1}}(z)$ is Lipschitz continuous, with a Lipschitz constant {\it independent of $h^{1:t-1}$}. Indeed, by direct computation,
    \begin{align}
    \nabla ^2_z (\log \cZ^\tau_{h^{1:t-1}})(z) =& \frac{1}{\sigma^4(\tau)}\Var_{X^t\sim \mu^{\tau,z}_{h^{1:t-1}}}(X^t-z) - \frac{1}{\sigma^2(\tau)}I_d. \label{identity-nabla2log}
    \end{align}
    Under Assumption \ref{ass:stability}-1(a), $|X^{t}|\leq K$, $\Pi_{h^{1:t-1}}^{\tau_0,t}$-a.s., for any $h^{1:t-1}\in \mR^{d(t-1)}$. By definition, $\mu^{\tau,z}_{h^{1:t-1}}$ is equivalent to $\Pi_{h^{1:t-1}}^{\tau_0,t}$; thus, $|X^t|\leq K$, $\mu^{\tau,z}_{h^{1:t-1}}$-a.s., for any $h^{1:t-1}$ and $z$. Equation \eqref{identity-nabla2log} then gives
    \begin{align}
   |\nabla ^2(\log \cZ^\tau_{h^{1:t-1}})(z)|\leq \frac{1}{\sigma^4(\tau)}K^2 +\frac{1}{\sigma^2(\tau)}.
    \end{align}
    It follows from \eqref{identity-F-Z} that the ground-truth score function $x\mapsto \nabla\log p^{\tau_0}_{t}(\tau,x|h^{1:t-1}) $ is Lipschitz, with Lipschitz constant $L_\tau:=K^2/h_2(\tau)^2+1/h_2(\tau)$ (use $h_1(\tau)\leq 1$). Now suppose that Assumption \ref{ass:stability}-1(b) holds instead. To ease the notation, we fix $h^{1:t-1}$ and write the supporting affine subspace as $E = a+E_0$, where $E_0$ is a linear subspace of $\mR^d$. Let $P_{E_0}$ and $P_{E_0^\perp}$ be the corresponding orthogonal projections. Then $z=a+P_{E_0}(z-a)+P_{E_0^\perp}(z-a)$, and
    \begin{align}
        \cZ^\tau_{h^{1:t-1}}(z) = e^{-\frac{|P_{E_0^\perp}(z-a)|^2}{2\sigma^2(\tau)}}(K_\tau*\pi_{h^{1:t-1}})(a+P_{E_0}(z-a)).
    \end{align}
Here, $*$ denotes convolution on the Euclidean subspace $E_0$. By Pr\`ekopa's theorem \citep[cf. Theorem 7 of][]{Prekopa1973LogConcave}, the function $z\mapsto K_\tau *\pi_{h^{1:t-1}}(a+P_{E_0}(z-a))$ is log-concave. As a result, $\cZ^\tau_{h^{1:t-1}}$ is also log-concave, which, combined with \eqref{identity-nabla2log}, gives
\begin{align}
       -\frac{1}{\sigma^2(\tau)}I_d \preceq\nabla^2(\log \cZ^\tau_{h^{1:t-1}})(z) \preceq 0.
\end{align}
Therefore, by \eqref{identity-F-Z}, the ground-truth score function is $L_\tau$-Lipschitz, with $L_\tau = 1/h_2(\tau)$.

To prove the Lipschitz continuity with respect to the noisy observation $h^{1:t-1}$, we use
\begin{align}
\nabla_z (\log \cZ^\tau_{h^{1:t-1}})(z) = \frac{1}{\sigma^2(\tau)}(m^{\tau,z}_{h^{1:t-1}} - z),\label{score:decomposition}
\end{align}
and directly conclude from \eqref{posterior-Lip} that
\begin{align}
  |\nabla_x\log p^{\tau_0}_t(\tau,x|h^{1:t-1}) - \nabla _x\log p^{\tau_0}_t(\tau, x|k^{1:t-1})|\leq &\frac{1}{h_1(\tau)\sigma^2(\tau)}|m^{\tau,z}_{h^{1:t-1}}-m^{\tau,z}_{k^{1:t-1}}|\Big|_{z = x/h_1(\tau)}\\
  \leq& \frac{L_{\tau_0}}{h_2(\tau)}(h_1(\tau)+ |x|)|h^{1:t-1}-k^{1:t-1}|\\
  \leq &\frac{L_{\tau_0}}{h_2(\tau)}(1+|x|)|h^{1:t-1}-k^{1:t-1}|.
\end{align}
Combined with the Lipschitz continuity with respect to $z$ (with a Lipschitz constant independent of $h^{1:t-1}$), we obtain \eqref{Lip:conditionalscore}. The Lipschitz constants are
\[
\frac{K^2}{h_2(\tau_0)^2}+\frac{1+L_{\tau_0}}{h_2(\tau_0)}
\quad\text{under Assumption \ref{ass:stability}-1(a),}
\qquad
\frac{1+L_{\tau_0}}{h_2(\tau_0)}
\quad\text{under Assumption \ref{ass:stability}-1(b).}
\]
Finally, \eqref{posterior-growth} and \eqref{score:decomposition} give Assumption \ref{assumption:P}-2 with the constant $C_{\tau_0}/h_2(\tau_0)$.
\end{proof}

\begin{remark}
\begin{enumerate}
\item Assumption \ref{ass:stability}-1(b) is weaker than standard strong log-concavity conditions because we allow the distribution to be supported on affine subspaces. More importantly, in the case of our noisy-history conditional score matching, we do not require Assumption \ref{ass:stability}-1 to be {\it uniform in} $t$ and $h^{1:t-1}$ (except for the support bound $K$). For example, the supporting affine subspaces of $\Pi_{h^{1:t-1}}^{\tau_0,t}$ may vary with $h^{1:t-1}$; for some noisy observations $h^{1:t-1}$, $\Pi_{h^{1:t-1}}^{\tau_0,t}$ may even be discrete with finite support. These weak conditions already imply that the Lipschitz constant with respect to $x$ in Assumption \ref{assumption:P} can be chosen uniformly in $h^{1:t-1}$. Nevertheless, these remain convenient {\it sufficient} conditions for Assumption \ref{assumption:P}, and we observe from the proof (see \eqref{identity-nabla2log}) that we only need to impose a uniform upper bound on the posterior variance $\Var_{\mu^{\tau,z}_{h^{1:t-1}}}(X^t)$. This may be verified in explicit examples.

	   \item  The decomposition \eqref{score:decomposition} is essentially the celebrated Tweedie's formula, and $\mu^{\tau,z}_{h^{1:t-1}}$ is the posterior distribution of $X^t\sim \Pi_{h^{1:t-1}}^{\tau_0,t}$ given the noisy observation $X^t+\sigma(\tau)\xi = z$. Therefore, Assumption \ref{ass:stability}-2 requires certain regularity of this posterior distribution. While challenging to obtain explicit sufficient conditions directly on the true data law $\mP$ in full generality, Examples \ref{exm:posterior:finite} and \ref{exm:posterior:gaussian} below show that $\Pi_{h^{1:t-1}}^{\tau_0,t}$ can be computed explicitly from $\mP$ in finite-state and Gaussian cases. We also record a non-Gaussian exponential example as a tractable posterior-level check of Assumption \ref{ass:stability}.
\end{enumerate}
\end{remark}

\begin{example}\label{exm:posterior:finite}
  Suppose the data law itself has finite path support:
  \begin{align}
      \mP=\sum_{\ell=1}^N w_\ell\delta_{x_\ell^{1:T}},
      \qquad w_\ell>0,\qquad \sum_{\ell=1}^N w_\ell=1.
  \end{align}
  For fixed $t\geq2$, since $H^{1:t-1}=h_1(\tau_0)X^{1:t-1}+\sqrt{h_2(\tau_0)}Z^{1:t-1}$, Bayes' rule gives
  \begin{align}
      \Pi_{h^{1:t-1}}^{\tau_0,t}
      =\sum_{\ell=1}^N r_\ell(h^{1:t-1})\delta_{x_\ell^t},
  \end{align}
  where
  \begin{align}
      r_\ell(h)
      =
      \frac{w_\ell\exp\{-|h-h_1(\tau_0)x_\ell^{1:t-1}|^2/(2h_2(\tau_0))\}}
      {\sum_{m=1}^N w_m\exp\{-|h-h_1(\tau_0)x_m^{1:t-1}|^2/(2h_2(\tau_0))\}}.
  \end{align}
  For $t=1$, the same example reduces to $\Pi_{h^{1:0}}^{\tau_0,1}=\mP_1$. Because the path support is finite, Assumption \ref{ass:stability}-1(a) holds with $K=\max_{\ell,t}|x_\ell^t|$.

  We next verify Assumption \ref{ass:stability}-2. The posterior distribution $\mu^{\tau,z}_{h^{1:t-1}}$ is supported on $\{x_\ell^t\}_{\ell=1}^N$ with weights
  \begin{align}
      \nu_\ell^{\tau,z}(h)
      =
      \frac{w_\ell\exp\{-|h-h_1(\tau_0)x_\ell^{1:t-1}|^2/(2h_2(\tau_0))\}K_\tau(z-x_\ell^t)}
      {\sum_{m=1}^Nw_m\exp\{-|h-h_1(\tau_0)x_m^{1:t-1}|^2/(2h_2(\tau_0))\}K_\tau(z-x_m^t)}.
  \end{align}
  Therefore,
  \begin{align}
      m_h^{\tau,z}:=m^{\tau,z}_{h^{1:t-1}}
      =\sum_{\ell=1}^N\nu_\ell^{\tau,z}(h)x_\ell^t,
  \end{align}
  and the growth condition \eqref{posterior-growth} is immediate from $|m_h^{\tau,z}|\leq K$.
  For a fixed finite horizon, the following estimate is automatic from bounded finite support. To keep the constant independent of $T$, we assume the finite path support satisfies the uniform cross-covariance bound
  \begin{align}
      \sup_{1\leq t\leq T}\sup_{\rho\in\Delta_N}
      \bigg\|
      \sum_{\ell=1}^N\rho_\ell
      \big(x_\ell^t-\bar x_\rho^t\big)
      \big(x_\ell^{1:t-1}-\bar x_\rho^{1:t-1}\big)^\mt
      \bigg\|_{\rm op}
	      \leq C_{\rm cov}<\infty,\label{finite-cross-cov}
  \end{align}
	  where $C_{\rm cov}$ is independent of $T$, $\Delta_N:=\{\rho\in[0,1]^N:\sum_{\ell=1}^N\rho_\ell=1\}$, $\bar x_\rho^t:=\sum_{\ell=1}^N\rho_\ell x_\ell^t$, and $\bar x_\rho^{1:t-1}:=\sum_{\ell=1}^N\rho_\ell x_\ell^{1:t-1}$.
  Direct differentiation gives the covariance-form identity
  \begin{align}
      \nabla_h m_h^{\tau,z}
      =
      \frac{h_1(\tau_0)}{h_2(\tau_0)}
      \sum_{\ell=1}^N\nu_\ell^{\tau,z}(h)
      (x_\ell^t-m_h^{\tau,z})
      \bigg(x_\ell^{1:t-1}-\sum_{m=1}^N\nu_m^{\tau,z}(h)x_m^{1:t-1}\bigg)^\mt .
  \end{align}
  Hence, by \eqref{finite-cross-cov},
  \begin{align}
      \|\nabla_h m_h^{\tau,z}\|_{\rm op}
      \leq C_{\tau_0}.
  \end{align}
  Therefore,
  \begin{align}
      |m_h^{\tau,z}-m_k^{\tau,z}|
      \leq C_{\tau_0}|h-k|
      \leq C_{\tau_0}(1+|z|)|h-k|.
  \end{align}
  This proves \eqref{posterior-Lip}.
\end{example}

\begin{example}\label{exm:posterior:gaussian}
  Suppose that, for each $t$, $X^{1:t}$ is jointly Gaussian under the data law $\mP$ (possibly with a singular covariance matrix). In this case, the noisy-history posterior kernel $\Pi_{h^{1:t-1}}^{\tau_0,t}$ can be computed directly from the original data distribution. For $t\geq 2$, write
  \begin{align}
      &\mu_{<t}=\mE[X^{1:t-1}],\qquad \mu_t=\mE[X^t],\\
      &\Sigma_{<t,<t}=\Cov(X^{1:t-1}),\qquad
      \Sigma_{t,<t}=\Cov(X^t,X^{1:t-1}),\qquad
      \Sigma_{t,t}=\Cov(X^t).
  \end{align}
  Since $H^{1:t-1}=h_1(\tau_0)X^{1:t-1}+\sqrt{h_2(\tau_0)}Z^{1:t-1}$, the pair $(X^t,H^{1:t-1})$ is jointly Gaussian, with
  \begin{align}
      \mE[H^{1:t-1}]&=h_1(\tau_0)\mu_{<t},\\
      \Cov(H^{1:t-1})&=h_1(\tau_0)^2\Sigma_{<t,<t}+h_2(\tau_0) I_{d(t-1)},\\
      \Cov(X^t,H^{1:t-1})&=h_1(\tau_0)\Sigma_{t,<t}.
  \end{align}
  The matrix $h_1(\tau_0)^2\Sigma_{<t,<t}+h_2(\tau_0) I_{d(t-1)}$ is positive definite because $h_2(\tau_0)>0$. Therefore, by the standard Gaussian conditioning formula,
  \begin{align}
      \Pi_{h^{1:t-1}}^{\tau_0,t}=\N(a_t(h^{1:t-1}),C_t),
  \end{align}
  where
  \begin{align}
      A_t&=h_1(\tau_0)\Sigma_{t,<t}\big(h_1(\tau_0)^2\Sigma_{<t,<t}+h_2(\tau_0) I_{d(t-1)}\big)^{-1},\\
      a_t(h^{1:t-1})&=\mu_t+A_t\big(h^{1:t-1}-h_1(\tau_0)\mu_{<t}\big),\\
      C_t&=\Sigma_{t,t}-h_1(\tau_0)^2\Sigma_{t,<t}\big(h_1(\tau_0)^2\Sigma_{<t,<t}+h_2(\tau_0) I_{d(t-1)}\big)^{-1}\Sigma_{<t,t}.
  \end{align}
  For $t=1$, we use the same notation with the convention $A_1=0$, $a_1=\mu_1$, and $C_1=\Sigma_{1,1}$.
  If $C_t$ is singular, the above Gaussian law is understood as a Gaussian probability measure supported on the affine subspace $a_t(h^{1:t-1})+\operatorname{Range}(C_t)$. In either case, $\Pi_{h^{1:t-1}}^{\tau_0,t}$ is log-concave on its supporting affine subspace, and hence Assumption \ref{ass:stability}-1(b) holds.

  We next verify Assumption \ref{ass:stability}-2. For $X^t\sim \N(a_t(h^{1:t-1}),C_t)$ and the additional Gaussian observation $X^t+\sigma(\tau)\xi=z$, Gaussian conditioning gives
  \begin{align}
      m^{\tau,z}_{h^{1:t-1}}
      &=a_t(h^{1:t-1})+C_t(C_t+\sigma^2(\tau)I_d)^{-1}(z-a_t(h^{1:t-1}))\\
      &=(I_d-B_t^\tau)a_t(h^{1:t-1})+B_t^\tau z,
      \qquad B_t^\tau:=C_t(C_t+\sigma^2(\tau)I_d)^{-1}.
  \end{align}
  Since $C_t$ is positive semidefinite, both $B_t^\tau$ and $I_d-B_t^\tau$ have operator norm bounded by $1$, uniformly over $\tau>0$. Hence
  \begin{align}
      |m^{\tau,z}_{h^{1:t-1}}-m^{\tau,z}_{k^{1:t-1}}|
      &\leq \|A_t\|_{\rm op}|h^{1:t-1}-k^{1:t-1}|,\\
      |m^{\tau,z}_{h^{1:t-1}}|
      &\leq |z|+|\mu_t-h_1(\tau_0)A_t\mu_{<t}|+\|A_t\|_{\rm op}|h^{1:t-1}|.
  \end{align}
  Consequently, if
  \begin{align}
      \sup_{1\leq t\leq T}\big(\|A_t\|_{\rm op}+|\mu_t-h_1(\tau_0)A_t\mu_{<t}|\big)<\infty,
  \end{align}
  then \eqref{posterior-Lip} and \eqref{posterior-growth} hold. For a fixed finite horizon, this is automatic; for a sequence of Gaussian models with increasing horizon, it is precisely the natural uniform stability condition on the Gaussian regression coefficients induced by the noisy-history observation.
\end{example}

\begin{example}\label{exm:posterior:exp}
     This example provides a non-Gaussian scenario in which Assumption \ref{assumption:P} can be verified. In this case, the exact form of $\Pi_h^{\tau_0,t}$ may be challenging to obtain, so we use it directly as an approximation to the clean kernel $\mP$. More specifically, take $T=2$, $d=1$, and suppose
    \begin{align}
        \Pi_h^{\tau_0,2}={\rm Exp}(\lambda(h)),
        \qquad
        \lambda(h)=\frac{1}{1+\ell(h)},
    \end{align}
    where $\ell$ is a nonnegative Lipschitz function. Since the exponential distribution is log-concave on $\mR$, Assumption \ref{ass:stability}-1(b) holds. For simplicity, denote $\lambda=\lambda(h)$ and $\sigma=\sigma(\tau)$. The posterior distribution $\mu^z_h$ is the normal distribution $\N(\mu,\sigma^2)$ truncated on $(0,\infty)$, where $\mu=z-\lambda\sigma^2$, and
    \begin{align}
	    m^{\tau,z}_{h} = \mu + \sigma \frac{\varphi(\mu/\sigma)}{\Phi_{\rm N}(\mu/\sigma)}\geq 0.
    \end{align}
    To investigate the growth condition \eqref{posterior-growth}, we consider the following subcases:
    \begin{enumerate}
\item $\mu\geq 0$ (which implies $z\geq 0$). In this case, we have $\sigma\leq \sqrt{z/\lambda}=\sqrt{z(1+\ell(h))}\leq (z+1+\ell(h))/2$. Hence,
\begin{align}
  0\leq m^{\tau,z}_{h} \leq z + \sigma\sqrt{\frac{2}{\pi}}\leq \bigg(1+\sqrt{\frac{1}{2\pi}}\bigg)(1+z+\ell(h)).
\end{align}
\item $\mu< 0$. In this case, we have $\zeta:=-\mu/\sigma> 0$, and we use \eqref{hbound} to obtain
\begin{align}
    0\leq m^{\tau,z}_{h} \leq \min\bigg\{\sigma\sqrt{\frac{2}{\pi}},\frac{\sigma^2}{\lambda\sigma^2-z}\bigg\}.
\end{align}
If $\sigma^2>2z/\lambda$, use the bound $\sigma^2/(\lambda\sigma^2-z)<2/\lambda=2(1+\ell(h))$; if $\sigma^2\leq 2z/\lambda$, use the bound $\sigma\sqrt{2/\pi}\leq \sqrt{2z(1+\ell(h))/\pi}\leq (z+1+\ell(h))/\sqrt{\pi}$.
    \end{enumerate}
    Combining these bounds and using $\ell(h)\leq C(1+|h|)$, we obtain a constant $C$, independent of $h$, $z$ and $\sigma$, such that
    \begin{align}
        |m^{\tau,z}_{h}|\leq C(1+|z|+|h|).
    \end{align}
We now consider the Lipschitz condition \eqref{posterior-Lip}. At differentiability points of $\ell$, denote $\zeta=-\mu/\sigma=\lambda\sigma-z/\sigma$. From the discussion above,
	$m^{\tau,z}_{h}=\sigma g(\zeta)$, where $g(\zeta)=-\zeta+\varphi(\zeta)/\Phi_{\rm N}(-\zeta)$. By Lemma \ref{lemma:mills},
\begin{align}
   |\partial_\lambda m^{\tau,z}_{h}|
	   = \sigma^2\Big(1+r(\zeta)\big(\zeta-r(\zeta)\big)\Big),
\end{align}
	where $r(\zeta)=\varphi(\zeta)/\Phi_{\rm N}(-\zeta)$ is the Mills-ratio notation used in Lemma \ref{lemma:mills}. Since $|\lambda'(h)|\leq C/(1+\ell(h))^2=C\lambda^2$, if $\zeta\leq 0$ (thus $\lambda\sigma^2\leq z$), then
\begin{align}
    |\partial_h m^{\tau,z}_{h}|
    \leq C\sigma^2\lambda^2
    \leq C(1+|z|).
\end{align}
If $\zeta>0$, then \eqref{hprimebound} implies
\begin{align}
    |\partial_h m^{\tau,z}_{h}|
    \leq C\sigma^2\lambda^2
    \min\bigg\{\frac{1}{(\lambda\sigma - z/\sigma)^2},1-\frac{2}{\pi}\bigg\}.
\end{align}
If $\lambda\sigma<2z/\sigma$, then the right-hand side is bounded by $C\sigma^2\lambda^2\leq C(1+|z|)$. If $\lambda\sigma\geq 2z/\sigma$, then $\lambda\sigma-z/\sigma\geq c\lambda\sigma$ for a universal constant $c>0$, and the right-hand side is bounded by $C$. Thus $|\partial_h m^{\tau,z}_{h}|\leq C(1+|z|)$ at differentiability points of $\ell$. Integrating this estimate along line segments gives
\begin{align}
    |m^{\tau,z}_{h}-m^{\tau,z}_{k}|
    \leq C(1+|z|)|h-k|,
\end{align}
which proves \eqref{posterior-Lip}. In this example, the relevant estimates are uniform in $\sigma=\sigma(\tau)>0$.
\end{example}

The following lemma is used in Example \ref{exm:posterior:exp}; we provide a short proof for completeness.
\begin{lemma}\label{lemma:mills}
Let $\varphi$ and $\Phi_{\rm N}$ denote the density function and cumulative distribution function of the standard normal distribution, respectively. Define $r(\zeta) = \varphi(\zeta)/\Phi_{\rm N}(-\zeta)$. Then, for $\zeta>0$, we have
    \begin{align}
	        &\frac{2}{\zeta+\sqrt{\zeta^2+8}} \leq r(\zeta)-\zeta\leq  \min\bigg\{\frac{1}{\zeta},\sqrt{\frac{2}{\pi}}\bigg\}  \label{hbound}\\
	        &0\leq 1+r(\zeta)\big(\zeta-r(\zeta)\big)\leq \min\bigg\{\frac{1}{\zeta^2},1-\frac{2}{\pi}\bigg\}.\label{hprimebound}
    \end{align}
\end{lemma}
\begin{proof}
	    By classical bounds for the Mills ratio $(1-\Phi_{\rm N}(\zeta))/\varphi(\zeta)$ (e.g., use the lower bounds with $n=1,2$ in Theorem 1 of \citet{GasullUtzet2014ApproximatingMillsRatio}), we obtain
    \begin{align}\label{mills:1}
	        \zeta\leq \frac{\varphi(\zeta)}{\Phi_{\rm N}(-\zeta)}\leq \zeta+\frac{1}{\zeta}.
    \end{align}
	    Moreover, the lower bound in \eqref{hbound} follows directly from the upper bound in (32) of \citet{GasullUtzet2014ApproximatingMillsRatio}. On the other hand, let $g(\zeta) = -\zeta+r(\zeta)$, so $g'(\zeta) = -1+r(\zeta)(r(\zeta)-\zeta)$. Suppose $Z_\zeta \sim \N(0,1)$ is truncated on $(\zeta,\infty)$. Integration by parts then gives
    \begin{align}
	    \mE[Z_\zeta]=& \frac{\int_\zeta^\infty t\varphi(t)dt}{\Phi_{\rm N}(-\zeta)}=\frac{\varphi(\zeta)}{\Phi_{\rm N}(-\zeta)}  \\
	    \mE[Z_\zeta^2 ] =& \frac{\int_\zeta^\infty t^2\varphi(t) dt}{\Phi_{\rm N}(-\zeta)}= \frac{\int_\zeta^\infty -t\varphi'(t) dt}{\Phi_{\rm N}(-\zeta)}\\
	    =&\frac{\zeta\varphi(\zeta)+\Phi_{\rm N}(-\zeta)}{\Phi_{\rm N}(-\zeta)}\\
	    =&\zeta r(\zeta)+1.
    \end{align}
	    In particular, we have $1+\zeta r(\zeta)=\mE[Z^2_\zeta] \geq (\mE[Z_\zeta])^2=r(\zeta)^2$, implying
    \begin{align}
	        r(\zeta)(r(\zeta)-\zeta)\leq 1, \quad \forall \zeta\in \mR.
    \end{align}
    Therefore, when $\zeta>0$,
    \begin{align}
	        r(\zeta)-\zeta \leq r(0) = \sqrt{\frac{2}{\pi}}.\label{mills:2}
    \end{align}
    Combining \eqref{mills:1} and \eqref{mills:2} completes the proof of \eqref{hbound}. It remains to prove the upper bound of \eqref{hprimebound}. By direct computation,
    \begin{align}
	        g''(\zeta) = r(\zeta)\big(2g(\zeta)^2+\zeta g(\zeta) -1 ).
    \end{align}
	    Because $s = 2/(\zeta+\sqrt{\zeta^2+8})$ is the positive root of $2s^2+\zeta s -1$, the lower bound in \eqref{hbound} implies $g''(\zeta)\geq 0$, hence $1+r(\zeta)(\zeta-r(\zeta))=-g'(\zeta)\leq -g'(0)=1-2/\pi$, which is the constant bound in \eqref{hprimebound}. To prove the tail bound $1/\zeta^2$, recall that $1+r(\zeta)(\zeta-r(\zeta))=\Var(Z_\zeta)$; thus, the desired bound is equivalent to $\Var(Y_\zeta)\leq 1$, where $Y_\zeta = \zeta (Z_\zeta-\zeta)\geq 0$. Note that, for $t\geq 0$,
    \begin{align}
\mP(Y_\zeta \geq t) = \mP(Z_\zeta \geq \zeta+t/\zeta)=\frac{\Phi_{\rm N}(-\zeta-t/\zeta)}{\Phi_{\rm N}(-\zeta)} = e^{-H_\zeta(t)},
    \end{align}
	where $H_\zeta(t) = -\log \Phi_{\rm N}(-\zeta-t/\zeta)+{\rm const.}$ satisfies (use $r(\zeta')\geq \zeta'$ for $\zeta'> 0$)
\begin{align}
	    H_\zeta'(t) = r(\zeta+t/\zeta)/\zeta \geq 1+t/\zeta^2>1.
\end{align}
Consequently, $G_\zeta(t):=H_\zeta^{-1}(t)$ is an increasing, 1-Lipschitz function. Because $\mP(Y_\zeta\geq t)=e^{-H_\zeta(t)}=\mP(E\geq H_\zeta(t))=\mP(G_\zeta(E)\geq t)$ for $E\sim {\rm Exp}(1)$, we have
\begin{align}
    \Var(Y_\zeta) = \Var(G_\zeta(E))\leq \Var(E)=1.
\end{align}
This concludes the proof.
\end{proof}

}

\subsection{Assumptions on the approximating networks}

Denote
\begin{align}
    G(\tau, h^{1:t-1},x):=\frac{h_1(\tau)}{h_2(\tau)}m^{\tau,x/h_1(\tau)}_{h^{1:t-1}}
    =\frac{h_1(\tau)}{h_2(\tau)}\frac{\int_{\mR^d}x^t e^{-\frac{|x - x^th_1(\tau)|^2}{2h_2(\tau)}}\Pi_{h^{1:t-1}}^{\tau_0,t}(dx^t)}{\int_{\mR^d}e^{-\frac{|x - x^th_1(\tau)|^2}{2h_2(\tau)}}\Pi_{h^{1:t-1}}^{\tau_0,t}(dx^t)}.
\end{align}
The Tweedie-type decomposition \eqref{score:decomposition} motivates us to construct an appropriate $G_\theta$ that approximates $G$ and define
\begin{equation}
    s_\theta(\tau,h^{1:t-1},x):=-\frac{x}{h_2(\tau)} + G_\theta(\tau,h^{1:t-1},x).
\end{equation}
If $G_\theta$ can be chosen to be Lipschitz continuous and compactly supported in $\{|h^{1:t-1}|\leq R\}\times \{|x|\leq R\}$, then, for any fixed $\delta\in(0,1)$, Assumption \ref{assumption:stheta} is satisfied with $D_{\tep}\sim C_{\tau_0}^2R^2$. Indeed, noting that $h_2(\tau)<1$, we have
\begin{align}
    2x\cdot s_\theta(\tau,h^{1:t-1},x)&=-\frac{2|x|^2}{h_2(\tau)}+2x\cdot G_\theta(\tau,h^{1:t-1},x)\\
    &\leq -2|x|^2 + 2\|G_\theta\|_{\infty}|x|\\
    &\leq -(1+\delta)|x|^2+\frac{\|G_\theta\|_{\infty}^2}{1-\delta}.
\end{align}
By the growth condition \eqref{posterior-growth}, we have $\|G_\theta\|_{\infty}\leq C_{\tau_0}(1+R)$ on this support, and this is exactly Assumption \ref{assumption:stheta}-2 with $D_{\tep}=C C_{\tau_0}^2(1+R^2)$ for a universal constant $C>0$. The Lipschitz conditions in Assumption \ref{assumption:stheta} are also inherited via Lemma \ref{universalapproximation} below.

In fact, we can construct a compactly supported $G_\theta$ that approximates $G$ to any accuracy $\epsilon_{\rm score}$ and has support size $R$ of logarithmic order in $\epsilon_{\rm score}^{-1}$. To prove this, we need an improved version of the universal approximation property of ReLU networks from \citet{Chen2020}:
\begin{lemma}[Lemma 10 in \citet{Chen2020}]\label{universalapproximation}
	    For any $\epsilon>0$, there exists a ReLU network architecture on $\mR^m$ such that, for every scalar-valued 1-Lipschitz function $f$ on $[0,1]^m$ with $f(0)=0$, some realization $\hat f$ of that architecture satisfies $\|f-\hat f\|_{\infty}\leq \epsilon$. Moreover, $\hat f$ is $10m$-Lipschitz with respect to the infinity norm. The scalar-valued statement can be applied coordinatewise to vector-valued targets.
\end{lemma}

\begin{lemma}
   For any $\tau\in [\tau_0,\Ti]$ and $\epsilon_{\rm score}>0$, choose
    \begin{align}
	R=C\sqrt{T}\Big(\log(eT)+\log \big(e+h_2(\tau_0)^{-2}\epsilon_{\rm score}^{-2}\big)\Big)
    \end{align}
	    for a constant $C>0$ independent of $\tau_0$, $\Ti$, $\epsilon_{\rm score}$ and $T$. Then there exists a ReLU network $G_\theta$ such that $G_\theta(\tau,\cdot,\cdot)$ is supported on $B_R:=\{|h^{1:t-1}|\leq R\}\times \{|x|\leq R\}$, and
    \begin{equation}
\sup_{\tau\in [\tau_0,\Ti]}\mE_{H^{1:t-1}\sim \mP^{\tau_0}_{1:t-1}}\big[\mE_{X_\tau^{t}\sim p^{\tau_0}_{t}(\tau,\cdot|H^{1:t-1})}[|G_\theta(\tau,H^{1:t-1},X^{t}_\tau)-G(\tau,H^{1:t-1},X^{t}_\tau)|^2]\big]\leq \epsilon^2_{\rm score}.\label{L2error:complete}
    \end{equation}
\end{lemma}

\begin{proof}
   We first control the tail outside $B_R$. Let
    \[
    H^{1:t-1}=h_1(\tau_0)X^{1:t-1}+\sqrt{h_2(\tau_0)}Z^{1:t-1},
    \qquad
    X_\tau^t=h_1(\tau)X^t+\sqrt{h_2(\tau)}\xi,
    \]
    where $Z^{1:t-1}$ and $\xi$ are independent standard Gaussian random vectors and are independent of $X^{1:t}\sim\mP_{1:t}$. This joint construction has the conditional law used in the lemma: conditionally on $H^{1:t-1}=h^{1:t-1}$, the law of $X^t$ is $\Pi_{h^{1:t-1}}^{\tau_0,t}$, and then $X_\tau^t$ has density $p_t^{\tau_0}(\tau,\cdot|h^{1:t-1})$. Moreover, by Bayes' rule and the definition of $G$, for the regular conditional law of $X^t$ given $(H^{1:t-1},X_\tau^t)$,
    \begin{align}
    G(\tau,h^{1:t-1},x)
    =\frac{h_1(\tau)}{h_2(\tau)}
    \mE\big[X^t\mid H^{1:t-1}=h^{1:t-1},\,X_\tau^t=x\big].
    \end{align}
    Write $A_R:=\{(H^{1:t-1},X_\tau^t)\notin B_R\}$, where $B_R=\{|h^{1:t-1}|\le R\}\times\{|x|\le R\}$. Jensen's inequality gives
    \begin{align}
    &\mE_{H^{1:t-1}}\mE_{X_\tau^{t}\sim p^{\tau_0}_{t}(\tau,\cdot|H^{1:t-1})}
    \Big[\big|G(\tau,H^{1:t-1},X^{t}_\tau)\big|^2\ind_{B_R^c}\Big]\notag\\
    &\qquad\leq
    \bigg(\frac{h_1(\tau)}{h_2(\tau)}\bigg)^2
    \mE\Big[\mE\big[|X^t|^2\mid H^{1:t-1},X_\tau^t\big]\ind_{A_R}\Big]\notag\\
    &\qquad=
    \bigg(\frac{h_1(\tau)}{h_2(\tau)}\bigg)^2
    \mE\big[|X^t|^2\ind_{A_R}\big]\notag\\
    &\qquad\leq
    \frac{1}{h_2(\tau_0)^2}
    \mE\Big[|X^t|^2\big(\ind_{\{|H^{1:t-1}|>R\}}+\ind_{\{|X_\tau^t|>R\}}\big)\Big].
    \end{align}
    In the last inequality, we used $h_1(\tau)\leq1$ and $h_2(\tau)\geq h_2(\tau_0)$ for $\tau\in[\tau_0,\Ti]$.
	    We now bound the two terms on the right-hand side of the preceding display. Assumption \ref{assumption:P}-3 gives
	    \[
	    \sup_{1\le s\le T}\mP(|X^s|>r)\leq Ce^{-cr},\qquad r>0.
	    \]
	    Also, the standard Gaussian tail gives $\mP(|Z^s|>r)+\mP(|\xi|>r)\leq Ce^{-cr^2}\leq Ce^{-cr}$ for $r\geq1$. Since $h_1(\tau_0)\leq1$ and $h_2(\tau_0)\leq1$, if $|H^{1:t-1}|>R$, then for some $s\leq t-1$,
	    \[
	    |X^s|+|Z^s|>R/\sqrt{T}.
	    \]
	    Thus, for $R\geq \sqrt T$,
	    \[
	    \mP(|H^{1:t-1}|>R)\leq CT e^{-cR/\sqrt T}.
	    \]
	    Similarly, $|X_\tau^t|\leq |X^t|+|\xi|$ implies $\mP(|X_\tau^t|>R)\leq Ce^{-cR}$ for $R\geq1$. Therefore, by Cauchy's inequality,
	    \begin{align}
	    \mE\big[|X^t|^2 \ind_{\{|H^{1:t-1}|>R\}}\big]
	    &\leq \big(\mE|X^t|^4\big)^{1/2}\mP(|H^{1:t-1}|>R)^{1/2}
	    \leq C\sqrt{T}e^{-cR/\sqrt{T}},\\
	    \mE\big[|X^t|^2 \ind_{\{|X_\tau^t|>R\}}\big]
	    &\leq \big(\mE|X^t|^4\big)^{1/2}\mP(|X_\tau^t|>R)^{1/2}
	    \leq Ce^{-cR}.
	    \end{align}
    Plugging these bounds into the preceding tail-reduction display yields
    \begin{equation}
	    \mE_{H^{1:t-1}}\mE_{X_\tau^{t}\sim p^{\tau_0}_{t}(\tau,\cdot|H^{1:t-1})}
	    \Big[\big|G(\tau,H^{1:t-1},X^{t}_\tau)\big|^2\ind_{B_R^c}\Big]
	    \leq \frac{C\sqrt{T}}{h_2(\tau_0)^2}e^{-cR/\sqrt{T}}.
	    \label{L2error:outsupport}
	    \end{equation}
	    After increasing the universal constant $C$ in the definition of $R$, the right-hand side of \eqref{L2error:outsupport} is smaller than $\epsilon_{\rm score}^2/2$.

  It remains to approximate on $B_R$. On this compact set, Assumption \ref{assumption:P}-1 and \eqref{posterior-growth} imply that $(h,x)\mapsto G(\tau,h,x)$ is Lipschitz and bounded by $C_{\tau_0}(1+R)$, uniformly over $\tau\in[\tau_0,\Ti]$. By Lemma \ref{universalapproximation}\footnote{With an additional ReLU layer, the support can be mapped from a unit cube of the appropriate dimension to $B_R$. See Appendix B.1 of \citet{Chen2023} for details.}, after rescaling $B_R$ to a unit cube, we can choose $G_\theta$ compactly supported on $B_R$ and satisfying
\begin{equation}
\mE_{H^{1:t-1}}\mE_{X_\tau^{t}\sim p^{\tau_0}_{t}(\tau,\cdot|H^{1:t-1})}\Big[\big|G_\theta(\tau,H^{1:t-1},X^{t}_\tau)-G(\tau,H^{1:t-1},X^{t}_\tau)\big|^2\ind_{B_R}\Big]<\frac{\epsilon_{\rm score}^2}{2}.\label{L2error:onsupport}
\end{equation}
	Combining \eqref{L2error:outsupport} and \eqref{L2error:onsupport} establishes \eqref{L2error:complete}. Since Assumption \ref{assumption:stheta}-2 holds with $D_{\tep}=C C_{\tau_0}^2(1+R^2)$ for a universal constant $C>0$, the above choice of $R$ gives
	\[
		D_{\tep}=C\cdot C_{\tau_0}^2T\Big(\log(eT)+\log(e+h_2(\tau_0)^{-2}\epsilon_{\rm score}^{-2})\Big)^2,
	\]
	which is the dependence stated in Proposition \ref{prop:assumption:stheta}.
\end{proof}

\section{Proof of Proposition \ref{prop:DSM}}\label{app:proof:conditionalDSM}

\begin{proof}
	By definition, for every noisy history $h^{1:t}$,
	\begin{equation}\label{eq:real-score-tau}
	p^{\tau_0}_{t+1}(\tau,x|h^{1:t})=\int_{\mR^d}\phi(\tau,x|x_0)\Pi_{h^{1:t}}^{\tau_0,t+1}(\md x_0).
	\end{equation}
The objective in \eqref{conditional-ESM} can therefore be rewritten as
\begin{align}
	 	&\mE_{H^{1:t}\sim \mP^{\tau_0}_{1:t}}\mE_{X_\tau^{t+1}\sim p^{\tau_0}_{t+1}(\tau,\cdot|H^{1:t})}\bigg|s^{t+1}_\theta(\tau,H^{1:t},X_{\tau}^{t+1})-\nabla_x \log p^{\tau_0}_{t+1}(\tau,X_\tau^{t+1}|H^{1:t})\bigg|^2\\
	 	=&\mE_{H^{1:t}\sim \mP^{\tau_0}_{1:t}}\mE_{X_\tau^{t+1}\sim p^{\tau_0}_{t+1}(\tau,\cdot|H^{1:t})}\big|s^{t+1}_\theta(\tau,H^{1:t},X_{\tau}^{t+1})\big|^2+C\\
	 	&-2\mE_{H^{1:t}\sim \mP^{\tau_0}_{1:t}}\int_{\mR^d}s^{t+1}_\theta(\tau,H^{1:t},x)\cdot \nabla_x p^{\tau_0}_{t+1}(\tau,x|H^{1:t})\md x \\
	 	=&\mE_{H^{1:t}\sim \mP^{\tau_0}_{1:t}}\mE_{X_\tau^{t+1}\sim p^{\tau_0}_{t+1}(\tau,\cdot|H^{1:t})}\big|s^{t+1}_\theta(\tau,H^{1:t},X_{\tau}^{t+1})\big|^2+C\\
        &-2\mE_{H^{1:t}\sim \mP^{\tau_0}_{1:t}}\int_{\mR^d}s^{t+1}_\theta(\tau,H^{1:t},x)\cdot \nabla_x \int_{\mR^d}\phi(\tau,x|x_0)\Pi_{H^{1:t}}^{\tau_0,t+1}(\md x_0) \md x\\
        =&\mE_{X^{1:t+1}\sim \mP_{1:t+1}}\mE_{Z^{1:t}}\mE_{X_\tau^{t+1}\sim \phi(\tau,\cdot|X^{t+1})}\big|s^{t+1}_\theta(\tau,H^{1:t},X_{\tau}^{t+1})\big|^2+C\\
        &-2\mE_{X^{1:t+1}\sim \mP_{1:t+1}}\mE_{Z^{1:t}}\mE_{X_\tau^{t+1}\sim \phi(\tau,\cdot|X^{t+1})}\big[s^{t+1}_\theta(\tau,H^{1:t},X_\tau^{t+1})\cdot \nabla_x \log \phi(\tau,X_\tau^{t+1}|X^{t+1})\big]\\
        =& \mE_{X^{1:t+1}\sim \mP_{1:t+1}}\mE_{Z^{1:t}}\mE_{X_\tau^{t+1}\sim \phi(\tau,\cdot|X^{t+1})}\bigg|s^{t+1}_\theta(\tau,H^{1:t},X_{\tau}^{t+1})-\nabla _x \log \phi(\tau,X_\tau^{t+1}|X^{t+1})\bigg|^2+C.
\end{align}
Here, $C$ is a generic constant independent of $\theta$. Thus, the two objectives differ by a $\theta$-independent constant and have the same set of minimizers.
\end{proof}

\section{Proof of Theorem \ref{thm:AWbounds}}\label{app:proof-AWbounds}

\subsection{Preliminary lemmas}
To obtain Wasserstein bounds from total variation bounds, we first provide a uniform-in-time moment estimate for an SDE under a {\it strong dissipativity condition}; see Section 1.2 of \citet{Khasminskii2012}. Let $f:[0,\infty)\times \mR^d\to \mR^d$, and suppose that there exist $D,\delta>0$ such that, for any $t,x$,
\begin{equation}\label{fcondition}
x\cdot f(t,x) \leq -(1+\delta)|x|^2+D.
\end{equation}
We consider the following SDE:
\begin{align}
&\md X_t = (X_t+f(t,X_t))\md t +\sqrt{2}\md B_t,\\
&X_0\sim \N(0,I).
\end{align}
% We aim to prove that, for sufficiently small (but fixed) $\delta >0$, with $t_{n}:=n^\delta$, $\mE \big[|X_{t_{n}}|^2I_{\{|X_{t_n}|>n\}}\big]$ has explicit bounds in terms of $n$, which tends to 0. To this end, we start with estimating $\mP(|X_t|\geq n)$, with explicit dependence on $t$ and $n$.

% For $n\geq 0$, consider a function $g_n:[0,\infty)\to [0,1]$ with the following properties:
% \begin{itemize}
% 	\item $\chi_n(r)\equiv 0$ for $r\leq n-1$; $\chi_n(r)\equiv 1$ for $r\geq n$;
% 	\item $\chi_n\in C^\infty(0,\infty)$, and $\chi'_n$, $\chi''_n$ are uniformly bounded.
% \end{itemize}
% Define $g_n(x)=\chi_n(|x|)$, it is then clear that $I_{\{|x|\geq n\}}\leq g_n(x)\leq I_{\{|x|\geq n-1\}}$.

\begin{lemma}
	There exist constants $c_1,c_2, c_3>0$, which are independent of $t$ and $D$, such that
	\begin{equation}\label{tail:est}
	\sup_{t\geq 0}\mE[e^{c_1|X_t|^2}]\leq c_2e^{c_3 D}.	\end{equation}
\end{lemma}
\begin{proof}
Let $V(x)=e^{\theta |x|^2}$, where $\theta>0$ will be determined later. The generator calculation gives
\begin{align}
\mathcal L_tV(x)
=&2\theta V(x)x\cdot \big(x+f(t,x)\big)+2\theta dV(x)+4\theta^2|x|^2V(x)\notag\\
\leq&\Big((4\theta^2-2\theta\delta)|x|^2+2\theta(D+d)\Big)V(x).
\end{align}
Choose $\theta<\min\{\delta/2,1/2\}$, set $R^2=2(D+d)/(\delta-2\theta)$, and denote
\[
I(x):=\Big((4\theta^2-2\theta\delta)|x|^2+2\theta(D+d)\Big)V(x).
\]
With $\gamma:=2\theta(D+d)$, we have $I(x)\leq -\gamma V(x)$ when $|x|\geq R$. When $|x|\leq R$,
\begin{align}
I(x)\leq 2\theta(D+d)V(R)
=2\theta(D+d)e^{\frac{2\theta(D+d)}{\delta-2\theta}}
\leq c'_1e^{c'_2D},
\end{align}
where $c'_1,c'_2$ are independent of $t$ and $D$. Hence
\[
\mathcal L_tV(x)\leq -\gamma V(x)+c'_3e^{c'_4D}.
\]
By applying It\^o's formula up to the localization times $\tau_n:=\inf\{t\geq0:|X_t|\geq n\}$ and then letting $n\to\infty$ in the localized Dynkin formula, we obtain that $m(t):=\mE[V(X_t)]$ satisfies
\[
m(t)\leq m(0)+\int_0^t\big(-\gamma m(s)+c'_3e^{c'_4D}\big)\md s.
\]
By Gronwall's inequality,
\[
\mE[V(X_t)]\leq e^{-\gamma t}\mE[V(X_0)]+\frac{c'_3e^{c'_4D}}{\gamma}\leq c'_5e^{c'_6D},
\]
where the constants $c'_1,\ldots,c'_6$ are independent of $t$ and $D$. This proves \eqref{tail:est} after renaming constants.
\end{proof}

\begin{corollary}\label{coro:moment-X-R}
	    There exist $C,c_1,c_2>0$, independent of $t$, $R$ and $D$, such that, for all $t,R>0$,
    \begin{align}
\mE[|X_t|^2\ind_{\{|X_t|\geq R\}}]\leq C e^{-c_1R^2+c_2 D}.
    \end{align}
\end{corollary}

\begin{proof}
	   Let $\theta>0$ be the exponential-moment constant from the preceding lemma. Since
	    \[
	    \sup_{u\geq0}u^2e^{-\theta u^2/2}<\infty,
	    \]
	    we have
	    \[
	    |x|^2\ind_{\{|x|\geq R\}}\leq C e^{-\theta R^2/2}e^{\theta |x|^2},\qquad x\in\mR^d.
	    \]
    Taking expectations and using \eqref{tail:est} gives the result after renaming constants.
\end{proof}

{  Due to the iterative structure of the adapted Wasserstein metric, we also need the following elementary lemma:

\begin{lemma}\label{lemma:iteration}
Let \((A_t)_{t\ge 0}\) be a nonnegative sequence with \(A_0=0\). Assume
that, for some integer \(m\ge 1\) and some \(C_1\ge 0\) and
\(C_{2,1},
\ldots,
C_{2,m}\in[0,1]\), all independent of $t$, we have
\[
A_{t+1}\le A_t+C_1\sqrt{A_t}+\sum_{i=1}^m C_{2,i},
\qquad t\ge 0.
\]
Then, for every integer \(t\ge 0\),
\[
A_t\le
\left(1+t\max\left\{\frac{C_1}{2},1\right\}\right)^2
\sum_{i=1}^m C_{2,i}^{2^{-(t-1)}}.
\]
\end{lemma}

\begin{proof}
Set \(E_0:=0\) and, for \(t\ge 1\),
\[
E_t:=\sum_{i=1}^m C_{2,i}^{2^{-(t-1)}}.
\]
Since \(0\le C_{2,i}\le 1\) for each \(i\), we have, for all \(t\ge 0\),
\[
E_t\le E_{t+1},
\qquad
E_1\le E_{t+1},
\qquad
\sqrt{E_t}\le E_{t+1}.
\]
Indeed, the last inequality is trivial when \(t=0\), while for \(t\ge 1\),
by subadditivity of the square root,
\[
\sqrt{E_t}
\le \sum_{i=1}^m C_{2,i}^{2^{-t}}
=E_{t+1}.
\]

Define \((M_t)_{t\ge 0}\) by \(M_0=0\) and
\[
M_{t+1}=M_t+C_1\sqrt{M_t}+1.
\]
We first claim that
\[
A_t\le M_tE_t,
\qquad t\ge 0.
\]
The case \(t=0\) follows from \(A_0=M_0E_0=0\). If the claim holds at time
\(t\), then
\[
\begin{aligned}
A_{t+1}
&\le M_tE_t+C_1\sqrt{M_t}\sqrt{E_t}+E_1 \\
&\le \bigl(M_t+C_1\sqrt{M_t}+1\bigr)E_{t+1}
 =M_{t+1}E_{t+1}.
\end{aligned}
\]
This proves the claim by induction.

It remains to estimate \(M_t\). For \(t\ge 1\), let \(K:=\max\{C_1/2,1\}\) and
\[
B_t:=\bigl(1+K(t-1)\bigr)^2.
\]
Then \(B_1=M_1=1\), and, for every \(t\ge 1\),
\[
B_{t+1}-B_t=2K\sqrt{B_t}+K^2\ge C_1\sqrt{B_t}+1.
\]
Since the map \(x\mapsto x+C_1\sqrt{x}+1\) is increasing on \([0,\infty)\),
induction gives \(M_t\le B_t\) for all \(t\ge 1\). Combining this with the
claim yields, for every \(t\ge 1\),
\[
A_t\le M_tE_t\le B_tE_t
=\bigl(1+K(t-1)\bigr)^2
\sum_{i=1}^m C_{2,i}^{2^{-(t-1)}}.
\]
This is the desired estimate.
\end{proof}

}

\subsection{The proof}

\begin{proof}[Proof of Theorem \ref{thm:AWbounds}]
	\ \\
	(1). We first prove the pointwise conditional estimate in \eqref{conditional-bound}. Fix $t\leq T-1$ and noisy/generated histories $h^{1:t},y^{1:t}$. With initial conditions $\bar X_0\sim p^{\tau_0}_{t+1}(\Ti,\cdot|h^{1:t})$ and $\bar Y_0\sim \N(0,I)$, consider the following two reverse SDEs:
\begin{align}
\md \bar X_\tau
&=\big[\bar X_\tau+2\nabla_x\log p^{\tau_0}_{t+1}(\Ti-\tau,\bar X_\tau|h^{1:t})\big]\md \tau+\sqrt{2}\md B_\tau,\\
\md \bar Y_\tau
&=\big[\bar Y_\tau+2s_\theta^{t+1}(\Ti-\tau,y^{1:t},\bar Y_\tau)\big]\md \tau+\sqrt{2}\md B_\tau.
\end{align}
{ Because of early stopping, we first estimate the distance between $\mP^{\tau_0}_{h^{1:t}}$ and $\mQ^{\Ti-\tau_0}_{y^{1:t}}$. In other words, we consider the two SDEs above only on the truncated time interval $[0,\Ti-\tau_0]$}. Then $\cL(\bar X_{\Ti-\tau_0})=\mP^{\tau_0}_{h^{1:t}}$ and $\cL(\bar Y_{\Ti-\tau_0})=\mQ^{\Ti-\tau_0}_{y^{1:t}}$. By the proof of total variation bounds for continuous-time DDPMs \citep[see, e.g.,][]{Song2021b,chen2023sampling}, we have
\begin{align}
\TV(\mP^{\tau_0}_{h^{1:t}},\mQ^{\Ti-\tau_0}_{y^{1:t}})^2
\leq& C\bigg(e^{-2\Ti}\big(1+M^{\tau_0}_2(h^{1:t})^2\big)\\
&\quad+\int_{\tau_0}^{\Ti}\mE_{X_\tau^{t+1}\sim p^{\tau_0}_{t+1}(\tau,\cdot|h^{1:t})}\Big[
\big|s^{t+1}_\theta(\tau,y^{1:t},X_\tau^{t+1})\notag\\
&\hspace{12em}{}-\nabla_x\log p^{\tau_0}_{t+1}(\tau,X_\tau^{t+1}\mid h^{1:t})\big|^2\Big]\md\tau\bigg)\\
\leq& C_{\tau_0}\bigg(e^{-2\Ti}\big(1+M^{\tau_0}_2(h^{1:t})^2\big)+\E^{\tau_0}(\tau_0,h^{1:t})
\\
&+|h^{1:t}-y^{1:t}|^2\int_{\tau_0}^{\Ti}\mE_{X_\tau^{t+1}\sim p^{\tau_0}_{t+1}(\tau,\cdot|h^{1:t})}[1+|X_\tau^{t+1}|^2]\md\tau\bigg)\\
\leq& C_{\tau_0}\Big(e^{-2\Ti}\big(1+M^{\tau_0}_2(h^{1:t})^2\big)+\E^{\tau_0}(\tau_0,h^{1:t})
+\Ti(1+M^{\tau_0}_2(h^{1:t})^2)|h^{1:t}-y^{1:t}|^2\Big).
\end{align}
Therefore, recalling that $\sqrt{a+b+c}\leq \sqrt a+\sqrt b+\sqrt c$ for $a,b,c\geq0$, we have
\begin{align}
\TV(\mP^{\tau_0}_{h^{1:t}},\mQ^{\Ti-\tau_0}_{y^{1:t}})
\leq C_{\tau_0}\big(e^{-\Ti}\big(1+M^{\tau_0}_2(h^{1:t})\big)+\E^{\tau_0}(\tau_0,h^{1:t})^{1/2}
\!+\Ti^{1/2}(1+M^{\tau_0}_2(h^{1:t}))|h^{1:t}-y^{1:t}|\big).
\end{align}
To bound $\cW_2(\mP^{\tau_0}_{h^{1:t}},\mQ^{\Ti-\tau_0}_{y^{1:t}})$ in terms of their total variation distance, we use the following classical result relating these two metrics \citep[see Theorem 6.15 of][]{villani2008}: for any two probability measures $\mu,\nu$ and any $R>0$,
\[
\cW_2^2(\mu,\nu)\leq R^2\TV(\mu,\nu)+\int_{|x|\geq R}|x|^2(\mu+\nu)(\md x).
\]
Taking $\mu=\mP^{\tau_0}_{h^{1:t}}$ and $\nu=\mQ^{\Ti-\tau_0}_{y^{1:t}}$, we have
\begin{align}
\cW_2^2(\mP^{\tau_0}_{h^{1:t}},\mQ^{\Ti-\tau_0}_{y^{1:t}})
\leq& R^2\TV(\mP^{\tau_0}_{h^{1:t}},\mQ^{\Ti-\tau_0}_{y^{1:t}})
+\mE_{\mP^{\tau_0}_{h^{1:t}}}\big[|H^{t+1}|^2\ind_{\{|H^{t+1}|\geq R\}}\big]\notag\\
&+\mE_{\mQ^{\Ti-\tau_0}_{y^{1:t}}}\big[|Y^{t+1}|^2\ind_{\{|Y^{t+1}|\geq R\}}\big].
\end{align}
We now estimate the three terms on the right-hand side. By Assumption \ref{assumption:P}-3 and the definition of the noisy-history kernel,
\begin{align}
\mE_{\mP^{\tau_0}_{h^{1:t}}}\big[|H^{t+1}|^2\ind_{\{|H^{t+1}|\geq R\}}\big]
\leq Ce^{-cR}\big(1+E_c^{\tau_0}(h^{1:t})\big).
\end{align}
By Assumption \ref{assumption:stheta}-2, the reverse SDE defining $\mQ^{\Ti-\tau_0}_{y^{1:t}}$ satisfies the dissipativity condition \eqref{fcondition} with $D=D_{\tep}$. Hence Corollary \ref{coro:moment-X-R} gives
\begin{align}
\mE_{\mQ^{\Ti-\tau_0}_{y^{1:t}}}\big[|Y^{t+1}|^2\ind_{\{|Y^{t+1}|\geq R\}}\big]
\leq C e^{-c_1R^2+c_2D_{\tep}}.
\end{align}
Plugging the preceding estimates into the total-variation-to-Wasserstein bound, we obtain
\begin{align}
\cW_2^2(\mP^{\tau_0}_{h^{1:t}},\mQ^{\Ti-\tau_0}_{y^{1:t}})
\leq C_{\tau_0}\Big(&R^2e^{-\Ti}\big(1+M^{\tau_0}_2(h^{1:t})\big)\\
&+R^2\E^{\tau_0}(\tau_0,h^{1:t})^{1/2}
+R^2\Ti^{1/2}(1+M^{\tau_0}_2(h^{1:t}))|h^{1:t}-y^{1:t}|\\
&+e^{-cR}\big(1+E_c^{\tau_0}(h^{1:t})\big)+e^{-c_1R^2+c_2D_{\tep}}\Big).
\end{align}
Taking $R=(2c_2D_{\tep}/c_1)^{1/2}$ and absorbing constants yields \eqref{conditional-bound}. The first marginal bound is the same argument with void history and the initial score condition \eqref{score-error-initial}.

\ \\
{ 
(2). We first record the smoothing estimate
\begin{equation}
\cA\cW_2^2(\mP,\mP^{\tau_0})\leq (K')^T\tau_0.
\end{equation}
Indeed, under the canonical observation coupling $\bar H^{1:T}=h_1(\tau_0)\bar X^{1:T}+\sqrt{h_2(\tau_0)}\bar Z^{1:T}$, Assumption \ref{assumption:P}-3 gives
\begin{align}
\mE|\bar X^{1:t}-\bar H^{1:t}|^2\leq Ct\tau_0,\qquad t=1,2,\cdots,T,
\end{align}
where $C$ is independent of $\tau_0$ and $T$\footnote{This estimate cannot be used to bound the adapted Wasserstein metric because $\bar X$ is not causal in $\bar H$.}. To construct a bicausal coupling recursively, suppose a coupling of $(X^{1:t},H^{1:t})$ has been constructed with marginals $\mP_{1:t}$ and $\mP^{\tau_0}_{1:t}$. Let $\Lambda_h^{\tau_0,t}$ be the regular conditional law of the clean history given the noisy history under the canonical observation, i.e.,
\[
\Lambda_h^{\tau_0,t}:=\cL(\bar X^{1:t}\in\cdot\mid \bar H^{1:t}=h),
\qquad
\bar H^{1:t}=h_1(\tau_0)\bar X^{1:t}+\sqrt{h_2(\tau_0)}\bar Z^{1:t}.
\]
Given $(X^{1:t},H^{1:t})=(x,h)$ in the recursively constructed coupling, sample $U^{1:t}\sim\Lambda_h^{\tau_0,t}$ conditionally independently of $X^{1:t}$ given $H^{1:t}=h$. Then $(U^{1:t},H^{1:t})$ has the same joint law as the canonical clean/noisy pair $(\bar X^{1:t},\bar H^{1:t})$. We next specify the coupling of the next clean coordinates. For each pair $(x,u)\in\mR^{dt}\times\mR^{dt}$, Assumption \ref{assumption:P}-4 gives
\[
\cW_2^2(\mP_{x^{1:t}},\mP_{u^{1:t}})\leq K^2|x^{1:t}-u^{1:t}|^2.
\]
By the standard measurable selection theorem for optimal transport kernels on Polish spaces, we may choose a measurable stochastic kernel
\[
\Gamma_t(x,u;\md v,\md \tilde v)\in \Pi(\mP_{x^{1:t}},\mP_{u^{1:t}})
\]
such that
\[
\int |v-\tilde v|^2\Gamma_t(x,u;\md v,\md \tilde v)
\leq K^2|x^{1:t}-u^{1:t}|^2.
\]
Equivalently, one may choose $\varepsilon$-optimal kernels and then let $\varepsilon\downarrow0$. Conditional on the histories taking values $(x,h,u)$, sample $(V,\tilde V)$ from $\Gamma_t(x,u;\cdot)$. Then $V\mid X^{1:t}\sim\mP_{X^{1:t}}$, $\tilde V\mid U^{1:t}\sim\mP_{U^{1:t}}$, and, by construction,
\[
\mE[|V-\tilde V|^2\mid X^{1:t},H^{1:t},U^{1:t}]
\leq K^2|X^{1:t}-U^{1:t}|^2.
\]
Finally set $X^{t+1}=V$ and $H^{t+1}=h_1(\tau_0)\tilde V+\sqrt{h_2(\tau_0)}\xi$ with an independent standard Gaussian $\xi$. The conditional law of $V$ given $(X^{1:t},H^{1:t})=(x,h)$ is $\mP_{x^{1:t}}$, and, by the disintegration formula defining $\Lambda_h^{\tau_0,t}$, the conditional law of $H^{t+1}$ given $H^{1:t}=h$ is $\mP^{\tau_0}_{h^{1:t}}$. Thus the recursive kernel has the correct one-step marginals and defines a bicausal coupling. If $A_t:=\mE|X^{1:t}-H^{1:t}|^2$ denotes the cost of the recursively constructed coupling, then
\begin{align}
A_{t+1}
&\leq A_t+C\mE|X^{1:t}-U^{1:t}|^2+C\tau_0\notag\\
&\leq A_t+C\mE|X^{1:t}-H^{1:t}|^2
   +C\mE|U^{1:t}-H^{1:t}|^2+C\tau_0\notag\\
&\leq (1+C)A_t+Ct\tau_0.
\end{align}
In the last step, $(U^{1:t},H^{1:t})$ has the canonical joint law of the clean and noisy histories, so the preceding canonical estimate applies. Since $A_1\leq C\tau_0$, induction yields $A_T\leq C T(1+C)^T\tau_0\leq (K')^T\tau_0$ after increasing $K'>1$.

It remains to compare $\mP^{\tau_0}$ with $\mQ^{\Ti-\tau_0}$. By the standard measurable selection theorem for optimal transport kernels on Polish spaces, we may choose a coupling kernel
\[
\gamma_0^\epsilon(\md h^1,\md y^1)\in \Pi(\mP^{\tau_0}_1,\mQ^{\Ti-\tau_0}_1)
\]
with cost at most $\cW_2^2(\mP^{\tau_0}_1,\mQ^{\Ti-\tau_0}_1)+\epsilon$, and, for each $t=1,\ldots,T-1$, measurable coupling kernels
\[
\gamma_t^\epsilon(h^{1:t},y^{1:t};\md h^{t+1},\md y^{t+1})
\in \Pi(\mP^{\tau_0}_{h^{1:t}},\mQ^{\Ti-\tau_0}_{y^{1:t}})
\]
with conditional cost at most
$\cW_2^2(\mP^{\tau_0}_{h^{1:t}},\mQ^{\Ti-\tau_0}_{y^{1:t}})+\epsilon$.
Define the joint law
\[
\pi^\epsilon(\md h^{1:T},\md y^{1:T})
:=\gamma_0^\epsilon(\md h^1,\md y^1)
\prod_{t=1}^{T-1}
\gamma_t^\epsilon(h^{1:t},y^{1:t};\md h^{t+1},\md y^{t+1}).
\]
By the characterization of bicausal couplings through regular conditional kernels,
\[
\pi^\epsilon\in\Pi_{\rm bc}(\mP^{\tau_0},\mQ^{\Ti-\tau_0}).
\]
Let
\[
A_t:=\mE_{\pi^\epsilon_{1:t}}|H^{1:t}-Y^{1:t}|^2.
\]
Integrating \eqref{conditional-bound} under $\pi^\epsilon_{1:t}$, using \eqref{score-error}, the moment bounds inherited by $\mP^{\tau_0}$ from Assumption \ref{assumption:P}-3, and Cauchy's inequality, gives
\begin{align}
A_{t+1}\leq A_t+C_{\tau_0}D_{\tep}\Ti^{1/2}\sqrt{A_t}
+\epsilon+\mathfrak b_{\tep}
+C_{\tau_0}D_{\tep}\Ti^{1/2}\epsilon_{\rm score}.
\end{align}
Set $A_0=0$ and regard the first-marginal estimate as the $t=0$ step; it gives the same recursion for $A_1$. Set $a_T=2^{-(T-1)}$. Lemma \ref{lemma:iteration} then yields, after letting $\epsilon\downarrow0$,
\begin{align}
\cA\cW_2^2(\mP^{\tau_0},\mQ^{\Ti-\tau_0})
\leq &C_{\tau_0}(1+T^2D_{\tep}^2\Ti)
\mathfrak b_{\tep}^{a_T}\notag\\
&+C_{\tau_0}(1+T^2D_{\tep}^2\Ti)
\Ti^{a_T/2}(C_{\tau_0}D_{\tep}\epsilon_{\rm score})^{a_T}.
\end{align}
Since $D_{\tep}>1$ and $e^{-c'D_{\tep}}\leq e^{-c'\sqrt{D_{\tep}}}$, we have
\[
\mathfrak b_{\tep}^{a_T}
\leq C_{\tau_0}\Big((D_{\tep}e^{-\Ti})^{a_T}+e^{-c'a_T\sqrt{D_{\tep}}}\Big),
\]
after decreasing $c'$ if necessary. Using also $\Ti>1$ and $1+T^2D_{\tep}^2\Ti\leq CT^2D_{\tep}^2\Ti$, we obtain
\begin{align}
\cA\cW_2^2(\mP^{\tau_0},\mQ^{\Ti-\tau_0})
\leq &C_{\tau_0}T^2D_{\tep}^2\Ti\big(D_{\tep}e^{-\Ti}\big)^{a_T}
+C_{\tau_0}T^2D_{\tep}^2\Ti e^{-c'a_T\sqrt{D_{\tep}}}\notag\\
&+C_{\tau_0}T^2D_{\tep}^2\Ti^{1+a_T/2}\big(D_{\tep}\epsilon_{\rm score}\big)^{a_T}.
\end{align}
Finally, the adapted Wasserstein triangle inequality and the preceding smoothing estimate imply
\[
\cA\cW_2^2(\mP,\mQ^{\Ti-\tau_0})
\leq (K')^T\tau_0+2\cA\cW_2^2(\mP^{\tau_0},\mQ^{\Ti-\tau_0}),
\]
and the factor $2$ is absorbed into $K'$ and $C_{\tau_0}$. This proves \eqref{eq:AWbounds}.
}
\end{proof}

{ 
\section{Proof of Proposition \ref{prop:exp-in-T}}\label{app:proof:sec:2}

\begin{proof}[Proof of Proposition \ref{prop:exp-in-T}]
We first improve the smoothing estimate between $\mP$ and $\mP^{\tau_0}$ under the $\mP$-contraction condition. To construct a bicausal coupling recursively, write
\begin{align}
    e_t:=\mE[|X^t-H^t|^2],\qquad t=1,2,\cdots,T.
\end{align}
For $t=1$, take the canonical coupling $H^1=h_1(\tau_0)X^1+\sqrt{h_2(\tau_0)}Z^1$. By Assumption \ref{assumption:P}-3, $e_1\leq C\tau_0$. Suppose now that the coupling has been constructed up to time $t$ with the correct marginals. Let $\Lambda_{h^{1:t}}^{\tau_0,t}$ be the regular conditional law of the clean history given the noisy history under the canonical observation. Given $(X^t,H^{1:t})=(x,h^{1:t})$, sample $U^{1:t}\sim\Lambda_{h^{1:t}}^{\tau_0,t}$ conditionally independently of the current clean coordinate $X^t$ given $H^{1:t}=h^{1:t}$. Then $(U^{1:t},H^{1:t})$ has the canonical clean/noisy joint law. In this Markovian setting, applying the measurable coupling-kernel argument from the smoothing estimate in the proof of Theorem \ref{thm:AWbounds} to the pair of kernels $\mP_x$ and $\mP_{u^t}$, and using $\cW_2(\mP_x,\mP_{u^t})\leq \kappa_P|x-u^t|$, we may conditionally sample $(V,\tilde V)$ so that $V|X^t\sim\mP_{X^t}$, $\tilde V|U^t\sim\mP_{U^t}$, and
\begin{align}
    \mE[|V-\tilde V|^2\mid X^t,H^{1:t},U^{1:t}]
    \leq \kappa_P^2 |X^t-U^t|^2.
\end{align}
Set $X^{t+1}=V$ and $H^{t+1}=h_1(\tau_0)\tilde V+\sqrt{h_2(\tau_0)}\xi$ with an independent standard Gaussian $\xi$. By the same disintegration argument as above, this preserves the $\mP$ and $\mP^{\tau_0}$ marginals and gives a bicausal coupling. Since $h_2(\tau_0)\lesssim \tau_0$, $(1-h_1(\tau_0))^2\lesssim\tau_0$, and $\mE|\tilde V|^2\leq C$ by Assumption \ref{assumption:P}-3, for any small $\rho>0$,
\begin{align}
    e_{t+1}\leq (1+\rho)\kappa_P^2\mE|X^t-U^t|^2+C_\rho\tau_0.
\end{align}
Moreover, because $(U^t,H^t)$ has the canonical clean/noisy joint law, $\mE|U^t-H^t|^2\leq C\tau_0$. Hence, for any small $\lambda>0$,
\begin{align}
    \mE|X^t-U^t|^2
    \leq (1+\lambda)e_t+C_\lambda\tau_0.
\end{align}
Choose $\rho,\lambda>0$ so that $\eta_P:=(1+\rho)(1+\lambda)\kappa_P^2<1$. Then
\begin{align}
    e_{t+1}\leq \eta_P e_t+C_{\kappa_P}\tau_0.
\end{align}
Induction gives $e_t\leq C_{\kappa_P}\tau_0$ for all $t$, and therefore
\begin{align}
    \cA\cW_2^2(\mP,\mP^{\tau_0})\leq \sum_{t=1}^T e_t
    \leq C_{\kappa_P}T\tau_0.
\end{align}

We next compare $\mP^{\tau_0}$ with $\mQ^{\Ti-\tau_0}$. Construct a bicausal coupling recursively between $\mP^{\tau_0}$ and $\mQ^{\Ti-\tau_0}$. Let $\pi_1$ be an optimal coupling of $\mP^{\tau_0}_1$ and $\mQ_1^{\Ti-\tau_0}$, and denote
\begin{align}
    B_t:=\mE_{\pi_t}[|H^t-Y^t|^2],\qquad t=1,2,\cdots,T.
\end{align}
Then $B_1\leq \Delta_{\tep}$. Suppose that $\pi_t$ has been constructed up to time $t$. For each pair of histories $(h^{1:t},y^{1:t})$, by the triangle inequality and the Markovian contraction of $\mQ^{\Ti-\tau_0}$ in Assumption \ref{ass:contraction},
\begin{align}
    \cW_2(\mP_{h^{1:t}}^{\tau_0},\mQ_{y^{1:t}}^{\Ti-\tau_0})
    &\leq \cW_2(\mP_{h^{1:t}}^{\tau_0},\mQ_{h^{1:t}}^{\Ti-\tau_0})
      +\cW_2(\mQ_{h^{1:t}}^{\Ti-\tau_0},\mQ_{y^{1:t}}^{\Ti-\tau_0})\notag\\
    &\leq \cW_2(\mP_{h^{1:t}}^{\tau_0},\mQ_{h^{1:t}}^{\Ti-\tau_0})+\kappa_Q |h^t-y^t|.
\end{align}
Choose $\lambda=2\kappa_Q^2/(1-\kappa_Q^2)$. Then
\begin{align}
    (a+b)^2\leq (1+\lambda)a^2+(1+\lambda^{-1})b^2,\qquad a,b\geq 0,
\end{align}
gives
\begin{align}
	    \cW_2^2(\mP_{h^{1:t}}^{\tau_0},\mQ_{y^{1:t}}^{\Ti-\tau_0})
	    \leq \bar C_{\kappa_Q} \cW_2^2(\mP_{h^{1:t}}^{\tau_0},\mQ_{h^{1:t}}^{\Ti-\tau_0})+\eta_{\kappa_Q} |h^t-y^t|^2,
    \label{contractive-one-step-bound}
\end{align}
where
\begin{align}
	    \bar C_{\kappa_Q}:=\frac{1+\kappa_Q^2}{1-\kappa_Q^2},\qquad
    \eta_{\kappa_Q}:=\frac{1+\kappa_Q^2}{2}<1.
\end{align}
As in the proof of Theorem \ref{thm:AWbounds}, using almost-optimal couplings for $\cW_2(\mP_{h^{1:t}}^{\tau_0},\mQ_{y^{1:t}}^{\Ti-\tau_0})$ and concatenating them with $\pi_t$ yields a bicausal coupling on the first $t+1$ coordinates. Integrating \eqref{contractive-one-step-bound} under $\pi_t$ and using the fact that the $H^{1:t}$-marginal is $\mP_{1:t}^{\tau_0}$, we get
\begin{align}
    B_{t+1}
	    &\leq \bar C_{\kappa_Q} \mE_{H^{1:t}\sim \mP_{1:t}^{\tau_0}}\big[\cW_2^2(\mP_{H^{1:t}}^{\tau_0},\mQ_{H^{1:t}}^{\Ti-\tau_0})\big]+\eta_{\kappa_Q} B_t \notag\\
	    &\leq \bar C_{\kappa_Q} \Delta_{\tep}+\eta_{\kappa_Q} B_t.
\end{align}
By induction,
\begin{align}
	    B_t\leq \left(1+\frac{\bar C_{\kappa_Q}}{1-\eta_{\kappa_Q}}\right)\Delta_{\tep},
    \qquad t=1,2,\cdots,T.
\end{align}
The adapted Wasserstein cost of the constructed bicausal coupling is bounded by
\begin{align}
    \mE_{\pi}[|H^{1:T}-Y^{1:T}|^2]
    =\sum_{t=1}^T B_t
    \leq C_{\kappa_Q} T\Delta_{\tep}.
\end{align}
Taking the infimum over bicausal couplings proves $\cA\cW_2^2(\mP^{\tau_0},\mQ^{\Ti-\tau_0})\leq C_{\kappa_Q} T\Delta_{\tep}$. Combining this estimate with the smoothing bound above and the adapted Wasserstein triangle inequality proves
\begin{align}
    \cA\cW_2^2(\mP,\mQ^{\Ti-\tau_0})
    \leq C_{\kappa_P}T\tau_0+C_{\kappa_Q}T\Delta_{\tep},
\end{align}
after absorbing the universal factor from the squared triangle inequality into the constants.
\end{proof}

}

\section{Proof of results in Section \ref{stability}}\label{app:proof:stability}

\begin{proof}[Proof of Lemma \ref{predictable-projection}]
			By taking successive differences, $(\vartheta\cdot S)_t=(\bar\vartheta \cdot S)_t$, $\mP$-a.s., for every $t=1,2,\cdots, T$ if and only if $\vartheta_t^\mt \triangle S^t = \bar\vartheta_t^\mt \triangle S^t$, $\mP$-a.s., for every $t=1,2,\cdots,T-1$, which in turn is equivalent to
		\begin{equation}
		\mE_\mP [|(\vartheta-\bar\vartheta)_t^\mt \triangle S^t|^2]=0.
		\end{equation}
			The proof is complete because
		\begin{align}
		0=& \mE_\mP [|(\vartheta_t-\bar\vartheta_t)^\mt \triangle S^t|^2]\\
		 =& \mE_\mP \big[(\vartheta_t-\bar\vartheta_t)^\mt \triangle S^t (\triangle S^t)^\mt(\vartheta_t-\bar\vartheta_t)     \big]\\
		 =& \mE_\mP \big[(\vartheta_t-\bar\vartheta_t)^\mt \mE_\mP[\triangle S^t (\triangle S^t)^\mt|\F_t](\vartheta_t-\bar\vartheta_t)     \big]\\
		 =& \mE_\mP \big[(\vartheta_t-\bar\vartheta_t)^\mt \Pi^S_t(\Pi^S_t)^\mt (\vartheta_t-\bar\vartheta_t)  \big]\\
		 =& \mE_\mP [|\Pi^S_t(\vartheta_t-\bar\vartheta_t)|^2 ]. \qedhere
		\end{align}
\end{proof}

\begin{proof}[Proof of Proposition \ref{prop:closedness}]
		Convexity follows from the convexity of $\mathcal K$ and the linearity of the stochastic integral, so we only need to prove the closedness. Suppose $(\vartheta^n\cdot S)_T\to Y$ in $L^2(\mP)$ for some $\vartheta ^n \in \Theta_{\mathcal K,\mP}$ and $Y\in L^2(\mP)$. Then $\{(\vartheta^n\cdot S)_T\}_{n\geq 1}$ is a Cauchy sequence in $L^2(\mP)$. Consider the semimartingale decomposition of $S$, $S^t=s^1+A_t+M_t$, where $A_1=M_1=0$, $\triangle M_t:=S^{t+1}-\mE_\mP[S^{t+1}|\F_{t}]$, and $\triangle A_t:=\mE_\mP[\triangle S^t|\F_{t}]$, $t=1,2,\cdots, T-1$. For any $\vartheta\in L^2_{\mP}(S)$ and $t=1,2,\cdots, T-1$, we have
		\begin{equation}\label{recursive-S}
			(\vartheta \cdot S)_{t+1} = (\vartheta \cdot S)_t+\vartheta_t^\mt \triangle A_t + \vartheta_t^\mt\triangle M_t.
		\end{equation}
	   For any $n,m\geq 1$, because the third term is orthogonal to the first two, we have
	   \begin{equation}
	   \mE_\mP\Big[\big|\big((\vartheta^n-\vartheta^m)\cdot S \big)_{t+1}	\big|^2\Big]=\mE_\mP\Big[\big| ((\vartheta^n-\vartheta^m) \cdot S)_{t}+(\vartheta^n_{t}-\vartheta^m_{t})^\mt \triangle A_t\big|^2 \Big]+\mE_\mP\Big[\big|(\vartheta^n_{t}-\vartheta^m_{t})^\mt \triangle M_t\big|^2 \Big].
	   \end{equation}
	   Starting from $t=T-1$, because $\{(\vartheta^n\cdot S)_T\}_{n\geq 1}$ is Cauchy in $L^2(\mP)$, we know $\mE_\mP\Big[\big|(\vartheta^n_{T-1}-\vartheta^m_{T-1})^\mt \triangle M_{T-1}\big|^2 \Big]\to 0$ as $n,m\to\infty$. From \eqref{ND}, it follows that
	   \begin{align}
	   \mE_\mP[|\Pi^S_{T-1}(\vartheta^n_{T-1}-\vartheta^m_{T-1})|^2]=&\mE_\mP\Big[(\vartheta^n_{T-1}-\vartheta^m_{T-1})^\mt \mE_\mP[\triangle S^{T-1}(\triangle S^{T-1})^\mt |\F_{T-1}](\vartheta^n_{T-1}-\vartheta^m_{T-1})\Big]	\\
	   &\leq \frac{1}{\delta}\mE_\mP\Big[(\vartheta^n_{T-1}-\vartheta^m_{T-1})^\mt \mE_\mP[\triangle M_{T-1}\triangle M_{T-1}^\mt |\F_{T-1}](\vartheta^n_{T-1}-\vartheta^m_{T-1})\Big]	\\
	   &=  \frac{1}{\delta}\mE_\mP\Big[\big|(\vartheta^n_{T-1}-\vartheta^m_{T-1})^\mt \triangle M_{T-1}\big|^2 \Big]\\
	   &\to 0.
	   \end{align}
	Thus $\{\Pi^S_{T-1}\vartheta ^n_{T-1}\}_{n\geq 1}$ is Cauchy in $L^2(\mP)$. Equivalently, $\{(\vartheta^n_{T-1} )^\mt\triangle S^{T-1}\}_{n\geq 1}$ is Cauchy in $L^2(\mP)$, and hence, by \eqref{recursive-S}, so is $\{(\vartheta^n\cdot S)_{T-1}\}_{n\geq 1}$. Repeating this argument backward, we conclude that $\{\Pi^S_t\vartheta^n_t\}_{n\geq 1}$ is Cauchy in $L^2(\mP)$ for every $t=1,2,\cdots,T-1$. Denote its $\F_t$-measurable $L^2(\mP)$-limit by $\xi_t$. Since there are finitely many $t$, up to a common subsequence, the convergence holds almost surely for every $t$. Since $\mathcal K$ is compact and $\vartheta^n_t\in\mathcal K$, we have $\xi_t\in\Pi^S_t\mathcal K$, $\mP$-almost surely. The multifunction $\Gamma_t(\omega):=\{\theta\in\mathcal K:\Pi^S_t(\omega)\theta=\xi_t(\omega)\}$ is $\F_t$-measurable with nonempty compact values, so the Kuratowski--Ryll-Nardzewski measurable selection theorem yields an $\F_t$-measurable selection $\widehat\vartheta_t\in\Gamma_t$. By the boundedness of $\mathcal K$ and Assumption \ref{assumption-S}, $\widehat\vartheta\in\Theta_{\mathcal K,\mP}$. Moreover, for every $t=1,2,\cdots,T-1$,
	\begin{align}
	\mE_\mP\big[|(\vartheta^n_t-\widehat\vartheta_t)^\mt\triangle S^t|^2\big]
	=\mE_\mP\big[|\Pi^S_t\vartheta^n_t-\xi_t|^2\big]\to 0.
	\end{align}
	Consequently,
	\begin{align}
	\big\|(\vartheta^n\cdot S)_T-(\widehat\vartheta\cdot S)_T\big\|_{L^2(\mP)}
	\leq\sum_{t=1}^{T-1}\big\|(\vartheta^n_t-\widehat\vartheta_t)^\mt\triangle S^t\big\|_{L^2(\mP)}
	\to 0.
	\end{align}
	By uniqueness of the $L^2(\mP)$-limit, $Y=(\widehat\vartheta\cdot S)_T\in G(\Theta_{\mathcal K,\mP})$.
\end{proof}

\begin{proof}[Proof of Theorem \ref{DPP:thm}] 
Choose full-$\mP_{1:t}$-measure Borel sets $H_t$ that are invariant under the conditional kernels and on which the conditional future second moments are finite. On the augmented state space $\{t\}\times\mR\times H_t$, with wealth transition $w'=w+\theta^\mt\triangle S^t$, the problem is a finite-horizon Borel model with compact action set $\mathcal K$, zero running costs, and nonnegative terminal cost $|c-w|^2$. Assumption \ref{assumption-S}, boundedness of $\mathcal K$, backward induction, and dominated convergence make each Bellman criterion finite, Borel measurable, and continuous in $(w,\theta)$ for each fixed history; hence the measurable maximum theorem gives Borel minimizing selectors. Propositions 8.2 and 8.5 of \citet{BertsekasShreve1978} then yield \eqref{DPP-1}--\eqref{DPP-2} and the optimal nonrandomized Markov policy; extending the selectors arbitrarily outside $\mR\times H_t$ gives \eqref{hattheta} and the stated recursive optimizer.
\end{proof}

\begin{proof}[Proof of Theorem \ref{DPP-stability}]
For simplicity, we use $V$ and $\tilde V$ to denote $V(\cdot;c,\mP)$ and $V(\cdot;c,\mQ)$, respectively. We prove \eqref{DPP-stability-eq} by backward induction. The terminal case is clear from the terminal condition. Suppose the estimate is known at $t+1$. Because $\pi\in \Pi_{\rm bc}(\mP,\mQ)$, for $\pi_{1:t}$-almost every $(s^{1:t},\tilde s^{1:t})$ we have $\pi_{s^{1:t},\tilde s^{1:t}}\in \Pi(\mP_{s^{1:t}}, \mQ_{\tilde s^{1:t}})$ \citep[cf. Proposition 5.1 of][]{BBLZ2017}. Therefore, considering a near-optimal $\vartheta$ in \eqref{DPP-2} gives
	\begin{align}
	|V(t,w,s^{1:t})-\tilde V(t,\tilde w,\tilde s^{1:t})|\leq & \sup_{\vartheta\in \mathcal K}\mE_{\pi_{s^{1:t},\tilde s^{1:t}}}\big|V(t+1,w+\vartheta^\mt(S^{t+1}-s^t),(s^{1:t},S^{t+1})) 	\\
	&-\tilde V(t+1,\tilde w+\vartheta^\mt (\tilde S^{t+1}-\tilde s^{t}),(\tilde s^{1:t},\tilde S^{t+1}))\big|\label{induction-t}
	\end{align}
Note that in this proof, we adopt the notation $S\sim \mP$, $\tilde S\sim \mQ$, and lowercase letters denote realized paths used in regular conditional kernels. Moreover, $\vartheta\in \mathcal K$ is uniformly bounded by assumption. We thus derive from \eqref{induction-t} and the induction hypothesis that
\begin{align}
		&|V(t,w,s^{1:t})-\tilde V(t,\tilde w,\tilde s^{1:t})|\\
		\leq{} &C\mE_{\pi_{s^{1:t},\tilde s^{1:t}}}\Bigg[
		\big(|w-\tilde w|+|s^t-\tilde s^t|+|S^{t+1}-\tilde S^{t+1}|\big)\notag\\
		&\qquad\times\big(1+c+|w|+|\tilde w|+|s^t|+|\tilde s^t|
		+|S^{t+1}|+|\tilde S^{t+1}|\notag\\
        &\qquad\qquad\quad +M_2^\mP(s^{1:t},S^{t+1})
        +M_2^{\mQ}(\tilde s^{1:t},\tilde S^{t+1})\big)\notag\\
		&\quad+\big(1+c+|w-\tilde w|+|s^t-\tilde s^t|
		+|S^{t+1}-\tilde S^{t+1}|+|w|+|\tilde w|\notag\\
        &\qquad\quad +|s^t|+|\tilde s^t|+|S^{t+1}|+|\tilde S^{t+1}|
        +M_2^\mP(s^{1:t},S^{t+1})\notag\\
        &\qquad\quad +M_2^{\mQ}(\tilde s^{1:t},\tilde S^{t+1})\big)
		\Big(\mE_{\pi_{S^{1:t+1},\tilde S^{1:t+1}}}
		|S^{t+2:T}-\tilde S^{t+2:T}|^2\Big)^{1/2}\notag\\
		&\quad+\mE_{\pi_{S^{1:t+1},\tilde S^{1:t+1}}}
		|S^{t+2:T}-\tilde S^{t+2:T}|^2\Bigg]\\
		=:{}&J_1+J_2+J_3.
\end{align}
Let $A_t:=|w-\tilde w|+|s^t-\tilde s^t|$ and $D_t:=D^\pi_t(s^{1:t},\tilde s^{1:t})$. We estimate $J_i$ as follows:
\begin{align}
 J_1&\leq C\Big\{A_tB_t+\big(B_t+A_t\big)D_t+D_t^2\Big\},\\
 J_2&\leq C\mE_{\pi_{s^{1:t},\tilde s^{1:t}}}\Bigg[\bigg(B_t+A_t+|S^{t+1}|+|\tilde S^{t+1}|+M_2^\mP(s^{1:t},S^{t+1})\\
 &\qquad\qquad\qquad\qquad +M_2^{\mQ}(\tilde s^{1:t},\tilde S^{t+1})+|S^{t+1}-\tilde S^{t+1}|\bigg)\\
 &\times \Big(\mE_{\pi_{S^{1:t+1},\tilde S^{1:t+1}}}|S^{t+2:T}-\tilde S^{t+2:T}|^2\Big)^{1/2}\Bigg]\\
 &\leq C\Big\{\big(B_t+A_t\big)D_t+D_t^2\Big\}.\\
 J_3&=C \mE_{\pi_{s^{1:t},\tilde s^{1:t}}}\Big|S^{t+2:T}-\tilde S^{t+2:T}\Big|^2\leq CD_t^2.
   \end{align}
Here, we used Cauchy-Schwarz, the tower property, and the bounds
\[
\mE_{\pi_{s^{1:t},\tilde s^{1:t}}}\big[\big(M_2^\mP(s^{1:t},S^{t+1})\big)^2\big]\leq \big(M_2^\mP(s^{1:t})\big)^2,\qquad
\mE_{\pi_{s^{1:t},\tilde s^{1:t}}}\big[\big(M_2^{\mQ}(\tilde s^{1:t},\tilde S^{t+1})\big)^2\big]\leq \big(M_2^{\mQ}(\tilde s^{1:t})\big)^2.
\]
Combining the bounds for $J_1$, $J_2$ and $J_3$, we establish \eqref{DPP-stability-eq}. Take $t=1$, $w=\tilde w=0$ and $\tilde s^1=s^1$ in \eqref{DPP-stability-eq}. For any such coupling,
\[
M_2^{\mQ}(s^1)\leq M_2^\mP(s^1)+D^\pi_1(s^1,s^1),
\]
and hence
\[
|V(1,0,s^1;c,\mP)-V(1,0,s^1;c,\mQ)|
\leq C\Big\{\big(1+|s^1|+c+M_2^\mP(s^1)\big)D^\pi_1(s^1,s^1)+\big(D^\pi_1(s^1,s^1)\big)^2\Big\}.
\]
Taking the infimum over all $\pi\in \Pi_{\rm bc}(\mP,\mQ)$ gives \eqref{stability-initial}.
\end{proof}

\begin{proof}[Proof of Corollary \ref{coro:MVstability}]
A standard verification argument gives
\[
V(1,0,s^1;c,\mQ)
=\inf_{\vartheta\in \Theta_{\mathcal K,\mQ}}
\mE_\mQ|c-(\vartheta\cdot S)_T|^2.
\]
Thus, for $\mu\in\{\mP,\mQ\}$,
\begin{equation}\label{vstar-V}
v^*(\mu)=\sup_{a\in \mR_+}\bigg\{-\frac{\gamma}{2}V\Big(1,0,s^1;\frac{1}{\gamma}+a,\mu\Big)+\frac{1}{2\gamma}+a\bigg\}.
\end{equation}
Note that for any $\vartheta\in \Theta_{\mathcal K,\mQ}$,
\begin{align}
|\mE_\mQ(\vartheta\cdot S)_T|&\leq \mE_\mQ|(\vartheta\cdot S)_T|\\
&\leq \mE_\mQ\sum_{t=1}^{T-1}|\vartheta_t||S^{t+1}-S^t|\\
&\leq C\mE_\mQ|S^{1:T}|\\
&\leq C\Big(\cW_2(\mP,\mQ)+\big(\mE_\mP|S^{1:T}|^2\big)^{1/2}\Big)\\
&\leq C.
\end{align}
Here and below, $C$ denotes a generic constant that depends only on $T$, $K$ and $M_S$, and it may vary from line to line. Importantly, it does not depend on $\mQ$ as long as $\cA\cW_2(\mP,\mQ)<1$. Consequently, for $\mu\in\{\mP,\mQ\}$,
\begin{equation}
V(1,0,s^1;c,\mu)\geq c^2-2Cc,\quad \forall c>0.
\end{equation}
 Therefore, the maximum in \eqref{vstar-V} is attained in a bounded interval $[0,\tilde L]$, and we conclude from \eqref{stability-initial} that
 \begin{align}
 |v^*(\mP)-v^*(\mQ)|&\leq \sup_{a\in [0,\tilde L]}\bigg\{\frac{\gamma}{2}\Big|V\Big(1,0,s^1;\frac{1}{\gamma}+a,\mP\Big)	-V\Big(1,0,s^1;\frac{1}{\gamma}+a,\mQ\Big)	\Big|\bigg\}\\
 &\leq C\cA\cW_2(\mP,\mQ),
 \end{align}
completing the proof.
\end{proof}

{ 
\begin{proof}[Proof of Lemma \ref{lemma:fixed-policy-stability}]
We first record two deterministic estimates. Since $\phi_t$ is $\mathcal K$-valued,
\begin{equation}\label{gain-linear-growth}
|(\phi\cdot s)_T|\leq C|s^{1:T}|.
\end{equation}
Moreover, \eqref{regular-feedback-primitive} implies the following local Lipschitz estimate:
\begin{equation}\label{gain-local-lip}
|(\phi\cdot s)_T-(\phi\cdot \tilde s)_T|
\leq C\big(1+|s^{1:T}|^{T-1}+|\tilde s^{1:T}|^{T-1}\big)|s^{1:T}-\tilde s^{1:T}|.
\end{equation}
Indeed, writing $d_t:=|(\phi\cdot s)_t-(\phi\cdot \tilde s)_t|$ and $\Delta s^t=s^{t+1}-s^t$, we have
\begin{align}
d_{t+1}
&\leq d_t+|\phi_t((\phi\cdot s)_t,s^{1:t})-\phi_t((\phi\cdot \tilde s)_t,\tilde s^{1:t})||\Delta s^t|\\
&\quad+K|(\Delta s^t)-(\Delta\tilde s^t)|\\
&\leq (1+C |s^{1:T}|)d_t+C(1+|s^{1:T}|)|s^{1:T}-\tilde s^{1:T}|.
\end{align}
Starting from $d_1=0$ and iterating over $t=1,\ldots,T-1$ gives
\[
d_T\leq C(1+|s^{1:T}|^{T-1})|s^{1:T}-\tilde s^{1:T}|.
\]
Interchanging $s$ and $\tilde s$ gives \eqref{gain-local-lip}. Combining \eqref{gain-linear-growth} and \eqref{gain-local-lip} also yields
\begin{equation}\label{gain-square-local-lip}
|(\phi\cdot s)_T^2-(\phi\cdot \tilde s)_T^2|
\leq C\big(1+|s^{1:T}|^{T}+|\tilde s^{1:T}|^{T}\big)|s^{1:T}-\tilde s^{1:T}|.
\end{equation}

Let $\eta>0$ be arbitrary. Choose $\pi\in\Pi(\mP,\mQ)$ such that
\[
D_\pi\leq \cW_2(\mP,\mQ)+\eta,
\quad
D_\pi:=\Big(\mE_\pi|S^{1:T}-\tilde S^{1:T}|^2\Big)^{1/2}.
\]
For any $m\in[0,T]$, Cauchy-Schwarz, Assumption \ref{assumption-S}, and \eqref{price-moment-fixed-policy} imply
\begin{equation}\label{holder-weighted-AW}
\mE_\pi\Big[\big(1+|S^{1:T}|^m+|\tilde S^{1:T}|^m\big)|S^{1:T}-\tilde S^{1:T}|\Big]
\leq CD_\pi.
\end{equation}
Indeed, since $2m\leq 2T$, the moment bounds under $\mP$ and $\mQ$ give
\[
\Big(\mE_\pi\big(1+|S^{1:T}|^m+|\tilde S^{1:T}|^m\big)^2\Big)^{1/2}\leq C,
\]
where $C$ depends only on $T$, $M_S$ and $M_S'$. Thus \eqref{holder-weighted-AW} follows.

Set
\[
X:=(\phi\cdot S)_T,\qquad \tilde X:=(\phi\cdot \tilde S)_T.
\]
By \eqref{gain-local-lip} and \eqref{holder-weighted-AW} with $m=T-1$,
\begin{equation}\label{mean-fixed-policy-bound}
|\mE_\mP X-\mE_\mQ\tilde X|
\leq \mE_\pi|X-\tilde X|
\leq CD_\pi.
\end{equation}
By \eqref{gain-square-local-lip} and \eqref{holder-weighted-AW} with $m=T$,
\begin{equation}\label{second-fixed-policy-bound}
|\mE_\mP X^2-\mE_\mQ\tilde X^2|
\leq CD_\pi.
\end{equation}
Finally, \eqref{gain-linear-growth}, Assumption \ref{assumption-S}, and \eqref{price-moment-fixed-policy} imply
\[
|\mE_\mP X|+|\mE_\mQ\tilde X|\leq C,
\]
and therefore, by \eqref{mean-fixed-policy-bound},
\begin{equation}\label{mean-square-fixed-policy-bound}
\big|(\mE_\mP X)^2-(\mE_\mQ\tilde X)^2\big|\leq CD_\pi.
\end{equation}
The bounds \eqref{mean-fixed-policy-bound}--\eqref{mean-square-fixed-policy-bound} give
\[
|v(\mP;\phi)-v(\mQ;\phi)|\leq CD_\pi.
\]
Letting $\eta\to 0$ completes the proof.
\end{proof}
}

{ 
\begin{proof}[Proof of Corollary \ref{coro:regular-policy-transfer}]
By Lemma \ref{lemma:fixed-policy-stability},
\[
v(\mP;\phi^\epsilon)\geq v(\mQ;\phi^\epsilon)-C\cW_2(\mP,\mQ).
\]
Using the $\epsilon$-suboptimality of $\phi^\epsilon$ under $\mQ$ and Corollary \ref{coro:MVstability}, we obtain
\begin{align}
v(\mP;\phi^\epsilon)
&\geq v^*(\mQ)-\epsilon-C\cA\cW_2(\mP,\mQ)-C\cW_2(\mP,\mQ)\\
&\geq v^*(\mP)-\epsilon-C\cA\cW_2(\mP,\mQ).
\end{align}
Here we use the bound $\cW_2(\mP,\mQ)\leq \cA\cW_2(\mP,\mQ)$.
\end{proof}

\begin{proof}[Proof of Lemma \ref{lemma:log-price-transfer}]
Fix $\pi\in\Pi_{\rm bc}(\mu,\nu)$ and let $(X,Y)\sim\pi$. Since $\Phi$ and $\Phi^{-1}$ are componentwise and preserve the time filtration, $(\Phi,\Phi)_\#\pi$ is a bicausal coupling between $\Phi_\#\mu$ and $\Phi_\#\nu$. For scalar $u,v$,
\[
|e^u-e^v|^2\leq e^{2R}|u-v|^2+4e^{2|u|}\ind_{\{|u|>R\}}+4e^{2|v|}\ind_{\{|v|>R\}}.
\]
Summing over time and coordinates, taking expectations, and using
\[
e^{2|z|}\ind_{\{|z|>R\}}\leq e^{-(a-2)R}e^{a|z|}
\]
gives
\[
\mE_{(\Phi,\Phi)_\#\pi}[|S^{1:T}-\tilde S^{1:T}|^2]
\leq e^{2R}\mE_\pi[|X^{1:T}-Y^{1:T}|^2]+C_{d,T}K_a e^{-(a-2)R}.
\]
Taking the infimum over $\pi\in\Pi_{\rm bc}(\mu,\nu)$ proves the first claim. The case $\delta:=\cA\cW_2(\mu,\nu)=0$ is immediate. For $0<\delta\leq1$, choose $R=(2/a)\log(\delta^{-1})$ to obtain the stated H\"older bound. Finally, if $a>2T$, the same exponential-moment bound implies $\mE[|\Phi(X)^{1:T}|^{2T}]<\infty$ and $\mE[|\Phi(Y)^{1:T}|^{2T}]<\infty$.
\end{proof}
}

\bibliographystyle{informs2014}
\bibliography{ref.bib}

\end{document}